\documentclass{article}

\usepackage{amsmath,amsthm,amssymb,bm,color,mathrsfs,mathtools}
\usepackage{amscd,verbatim}
\usepackage{bbding,wasysym,pifont}
\usepackage[breaklinks=true,colorlinks=true,linkcolor=blue,urlcolor=blue,citecolor=blue]{hyperref}
\usepackage{graphicx}
\usepackage{multirow}
\usepackage{comment}
\usepackage[arrow,matrix,curve]{xy}
\usepackage{authblk}
\usepackage{tikz-cd}
\usepackage{subcaption}
\usepackage{colortbl}
\usepackage{rotating}

\usepackage{tikz-cd}
\usetikzlibrary{decorations.markings}
\tikzset{->-/.style={decoration={
  markings,
  mark=at position #1 with {\arrow{>}}},postaction={decorate}}}
\tikzset{->-/.default=0.5}

\newcommand\wt{\widetilde}
\newcommand\wh{\widehat}

\newcommand{\CC}{{\mathbb C}}
\newcommand{\RR}{{\mathbb R}}
\newcommand{\TT}{{\mathbb T}}
\newcommand{\ZZ}{{\mathbb Z}}

\newcommand{\C}{\mathbb{C}}

\newcommand{\Z}{\mathbb{Z}}

\newcommand{\sG}{{\mathscr G}}

\newcommand{\im}{\mathrm{i}}

\newcommand{\Sflip}{S^1_{\rm flip}}
\newcommand{\Striv}{S^1_{\rm triv}}
\newcommand{\Sfree}{S^1_{\rm free}}
\newcommand{\hSflip}{\hat{S}^1_{\rm flip}}
\newcommand{\hStriv}{\hat{S}^1_{\rm triv}}
\newcommand{\TR}{{\rm T}_R}
\newcommand{\TZtwo}{{\rm T}_{\ZZ_2}}

\newcommand{\pt}{\mathrm{pt}}

\newcommand{\shade}{\cellcolor[gray]{0.9}}

\newtheorem{theorem}{Theorem}[section]

\newtheorem{lemma}[theorem]{Lemma}
\newtheorem{proposition}[theorem]{Proposition}
\newtheorem{corollary}[theorem]{Corollary}
\newtheorem{definition}[theorem]{Definition}
\newtheorem{remark}[theorem]{Remark}

\begin{document}

\title{Crystallographic T-duality}

\author[1]{Kiyonori Gomi}
\affil[1]{Department of Mathematical Sciences, Shinshu University, Matsumoto, Nagano 390-8621, Japan}

\author[2]{Guo Chuan Thiang
}
\affil[2]{School of Mathematical Sciences, University of Adelaide, SA 5005, Australia}


\maketitle

\begin{abstract}
We introduce the notion of crystallographic T-duality, inspired by the appearance of $K$-theory with graded equivariant twists in the study of topological crystalline materials. Besides giving a range of new topological T-dualities, it also unifies many previously known dualities, motivates generalisations of the Baum--Connes conjecture to graded groups, provides a powerful tool for computing topological phase classification groups, and facilitates the understanding of crystallographic bulk-boundary correspondences in physics.
\end{abstract}


\tableofcontents

\section{Introduction}
Mathematical interest in T-duality was stimulated by the discovery in string theory that for a circle bundle over a manifold $X$ with H-flux (an integral degree-3 cohomology class), there is a T-dual circle bundle with dual H-flux, such that the H-twisted $K$-theories on either side coincide despite the bundles generally being topologically distinct \cite{BEM}. For instance, a 3-dimensional lens space $L(p)$ with $k$ units of H-flux is T-dual to a generally non-homeomorphic lens space $L(k)$ with $p$ units of H-flux \cite{MMS,BEM}. The desire to understand the general mechanism behind ``topological T-dualities'' of this kind led to a rekindling of interest in twisted $K$-theory, and a very fruitful $C^*$-algebraic approach \cite{Rosenberg,MR,Rosenberg2,DM-DR} even relates T-duality to the deep Baum--Connes isomorphisms \cite{BCH}.

Recently, twisted $K$-theory also started to appear in solid-state physics due to the influential work of Freed--Moore \cite{FM} which generalised the Bott-``Periodic Table'' of topological insulators \cite{Kitaev} to the \emph{crystallographic} setting. Here, the relevant twists are \emph{graded} and \emph{equivariant} (i.e.\ have a $H^1_G(X,\ZZ_2)$ part) over a so-called \emph{Brillouin} torus of \emph{quasimomenta}, and are generally \emph{torsion} cohomology classes. The physical intuition of position-momentum duality, together with the insight that T-duality may be understood as a type of ``topological Fourier transform'' (Fourier--Mukai transform), lead us to introduce the notion of \emph{crystallographic T-duality} in this paper. 

Our central result is Theorem \ref{thm:crystalTdual}, which says that for each crystallographic space group $\mathscr{G}$ in $d$-dimensions (and there are many such groups), there is an isomorphism of twisted $K$-theories,
$${\rm T}_\mathscr{G}:K^{-\bullet+\sigma_\mathscr{G}}_G(T^d_\mathscr{G})\overset{\cong}{\longrightarrow}K^{-\bullet-d+\tau_{\mathscr{G}}}_{G}(\hat{T}^d).$$
On the LHS, $T^d_\mathscr{G}$ is a ``position space'' $d$-torus equipped with a naturally defined affine action of a finite quotient $G$ of $\mathscr{G}$ (alternatively, a flat orbifold), while $\sigma_{\mathscr{G}}$ is a \emph{graded} $G$-equivariant twist (\S\ref{sec:twistgeneral}) from the $K$-nonorientability of this $G$-action. On the RHS, $\hat{T}^d$ is the ``momentum space'' Brillouin torus equipped with the dual $G$-action (generally a \emph{different} flat orbifold), and $\tau_{\mathscr{G}}$ is another equivariant twist constructed from group-theoretic properties of $\mathscr{G}$. Strikingly, the data on one side appears at first glance to be of a different nature to the data on the other side, yet the total $K$-theoretic information is ``conserved''. The unifying object is the crystallographic group $\mathscr{G}$, which is capable of providing both sets of data. 

{\bf Super-Baum--Connes conjecture}: A similar phenomenon arises in the simplest nontrivial example verifying the Baum--Connes conjecture: the group $\ZZ$ has classifying space $B\ZZ=S^1$ a circle, and character space/Pontryagin dual (spectrum of reduced group $C^*$-algebra $C^*_r(\ZZ)$) another circle, $C^*_r(\ZZ)\cong C(\hat{S^1})$, and there is an assembly map $\mu_\ZZ$ implementing T-duality isomorphisms 
$${\rm T}_\ZZ:K^{1-\bullet}(S^1)\cong K_\bullet(B\ZZ)\xrightarrow{\mu_\ZZ} K_\bullet(C^*_r(\ZZ))\cong K^\bullet(\hat{S^1}).$$
Our general crystallographic T-duality ${\rm T}_\mathscr{G}$ is shown by a chain of isomorphisms involving the Baum--Connes assembly map applied to $\mathscr{G}$, and so implicitly passes through a (graded) $C^*$-algebraic formulation. 

The simplest example in the setting of ($\ZZ_2$-)graded groups is $\ZZ$ equipped with the nontrivial even/odd grading, which can be thought of as the \emph{frieze group} \cite{KL} (a generalised crystallographic group) generated by (an odd) glide reflection, usually denoted $\mathsf{p11g}$, see Fig.\ \ref{fig:glideaxis}. The T-duality associated to $\mathsf{p11g}$, Eq.\ \eqref{p11gTduality} of \S\ref{sec:p11g}, may be rewritten as
\begin{equation*}
{\rm T}_{\sf{p11g}}: K_{\bullet+c}(S^1)\overset{\cong}{\longrightarrow} K^{\rm graded}_\bullet(C^*_r({\sf p11g})),\label{p11gTduality}
\end{equation*}
where on the LHS, $c\in H^1(S^1,\ZZ_2)$ is the generating orientation twist. Thinking of $(S^1,c)$ as an appropriate notion of classifying space for the \emph{graded} group ${\sf p11g}\cong\ZZ$, the isomorphism
$K_{\bullet+c}(S^1)\longrightarrow K^{\rm graded}_\bullet(C^*_r({\sf p11g}))$ essentially verifies a ``super-Baum--Connes conjecture'' for the graded group ${\sf p11g}$. Remarkably, it turns out that $K^{\rm graded}_0(C^*_r({\sf p11g}))\cong\ZZ/2$ \cite{SSG1,GT}. Correspondingly, it is easy to see that the $c$-twisted boundary of a loop winding around $S^1$ is twice of a point, so that after passing to homology, $K_{0+c}(S^1)\cong H_{0+c}(S^1)\cong\ZZ/2$ on the LHS. An extraordinary amount of effort has been directed towards the Baum--Connes conjecture for ordinary (ungraded) groups, and it is hoped that our paper motivates its study in the general setting of super (i.e.\ graded) groups, for which crystallography provides ample well-motivated examples.

{\bf Computations and unification of known T-dualities.} We provide numerous computable examples of crystallographic T-duality (\S\ref{sec:H3twistexamples}-\ref{sec:H1twistexamples}). Because tori appear on both sides, the duality actually becomes a tool to compute many previously unknown twisted equivariant $K$-theory groups ``for free'', and to supplement spectral sequence methods by resolving extension problems (\S\ref{sec:extensionproblem}). Let us also emphasise that \emph{torsion} $K$-theory classes are of particular interest in physics, so rational methods are not necessarily desirable. Via a large number of explicit examples, we further show that the crystallographic T-duality factorises through several circle bundle T-dualities --- ``partial Fourier transforms'' --- such as the `Real' T-duality of \cite{Gomi1} involving $K_\pm$ groups. From this, we obtain intricate webs of T-dualities whose individual links are sometimes already known, but are now assembled together in coherent patterns (\S\ref{sec:2Ddualities}, \S\ref{sec:1Ddualities}). Since crystallographic groups basically correspond to finite group actions on tori, crystallographic T-duality is the ``most general'' notion of T-duality with finite group equivariance, at least in the sense of dualising the fibres of trivial equivariant torus bundles $T^d_\sG\times X$ (we primarily study $X=\pt$) in this paper. We anticipate that we can generalise even further to non-trivial $G$-equivariant torus bundles, whence equivariant characteristic classes should come into play. The case $G=\ZZ_2$ is discussed in \S\ref{sec:ztwoTduality} as a natural cousin of `Real' T-duality \cite{Gomi1}, and we leave the general ``fibred'' version of crystallographic T-duality for a subsequent work.

{\bf Physics applications.} Let us discuss briefly some physics motivations and possible applications. The role of $K$-theory in string theory \cite{MM,Witten} and in solid state physics \cite{Bellissard} has been known for several decades. In the former, D-brane charges in various flavours of (super)string theory live in appropriate $K$-(co)homology groups of spacetime, as initially argued for the Type II case by \cite{MM} and for several other cases by \cite{Witten}. In the latter, invariants of topological phases live in the $K$-theory of some (noncommutative) momentum space \cite{Kitaev,FM,Thiang}. In both fields, dualities play key conceptual roles. For instance, T-dualities relate complementary features of and account for different types of string theories \cite{Buscher, Vafa}, while a closely related position-momentum space duality was already observed in \cite{Cartier} and features in the Bloch--Floquet--Fourier transform used extensively in solid state physics. Furthermore, index theory as formulated in $K$-theoretic language appears in T-duality in string theory \cite{Hori}, in accounting for the quantum Hall effect \cite{Bellissard2}, and in formulating the bulk-edge correspondence rigorously \cite{GT, PSB, Kubota2, Bourne}. D-brane transformation under T-duality can also be represented by a geometric Fourier--Mukai/Nahm transform \cite{Hori}. Thus we also define a general \emph{crystallographic Fourier--Mukai transform} ${\rm T}_{\mathscr{G}}^{\rm FM}$ (\S\ref{sec:Poincarebundle}), which is expected to implement crystallographic T-duality ${\rm T}_\mathscr{G}$ in a more geometric way suitable for string theory applications, although we do not pursue the latter in any detail in this paper.

In solid state physics, the RHS of the crystallographic T-duality, $K^{-\bullet+\tau_{\mathscr{G}}}_{G}(\hat{T}^d)$, serves as a convenient classification group for \emph{bulk} topological crystalline insulator phases (roughly: equivalence classes of $\mathscr{G}$-invariant Hamiltonians with a spectral gap at zero), assuming that one is working in the single-particle (i.e.\ non-interacting) framework \cite{FM,Thiang,SSG2,SSG3}. For $\bullet=0$, these are Class A insulators, whereas $\bullet=1$ is relevant for Class AIII ones characterised by the possession of an additional \emph{chiral} symmetry (an odd ``supersymmetry''). Quite aside from tabulating the possible topological phases, it is critical that the actual experimental signatures of nontrivial topological (crystalline) insulators are expected to be \emph{topological zero modes} localised at an appropriate boundary cut into the sample. Thus the somewhat ``invisible'' bulk topological invariant $K^{-\bullet+\tau_{\mathscr{G}}}_{G}(\hat{T}^d)$ appears on the boundary as a phenomenon which is simultaneously topological \emph{and} analytic in nature. This suggests that the so-called \emph{bulk-boundary correspondence} is index-theoretic in nature, and indeed justifies the appropriateness of the $K$-theoretic classification in the first place. In the non-crystalline case (i.e.\ $\mathscr{G}=\Pi\cong\ZZ^d$), such correspondences have been studied in mathematical physics for some time \cite{Hatsugai,GP,PSB,Bourne}, and with nonequivariant H-flux introduced in \cite{HMT, HMT2}. An approach using coarse geometry and $C^*$-algebras appears in \cite{Kubota2} and some crystalline symmetries were studied there. T-duality as a topological Fourier transform was introduced into this field in \cite{MT1}, and used to understand why certain Gysin (topological index) maps should implement the bulk-to-boundary homomorphisms \cite{HMT, HMT2}. 

The main roadblock one encounters when trying to study general crystallographic bulk-boundary correspondences rigorously is that the appropriate ``index'' for the topological boundary zero modes is not known. The main insight of \cite{GT} is that the symmetries of the boundary should be generalised from a lower-dimensional crystallographic space group, to a \emph{subperiodic group} (e.g.\ the frieze group ${\sf p11g}$ above), which is generally \emph{graded} according to the data of how the boundary sits inside the bulk (\S\ref{sec:bulkboundaryapplication}). The linear space of boundary zero modes should therefore host a graded representation of the boundary symmetries. This expectation was verified in \cite{GT} through a new mod 2 index theorem (valued in $K^{\rm graded}_\bullet(C^*_r({\sf p11g}))$) which counts the ``glide zero modes'' that appear along a boundary of a 2D topological insulator with glide reflection symmetry. More generally, the framework of graded groups and graded equivariant twistings of $K$-theory allows us to formulate appropriate ``super-indices'' for exotic topological boundary modes arising in crystalline topological phases.

\medskip

{\bf Notation}: $\ZZ_2$ denotes the 2-element group $\{\pm 1\}$ written multiplicatively, while $\ZZ/2=\{0,1\}$ is the additive version. When necessary, objects (e.g.\ bundle, projection map, twist) on one side of a T-duality are denoted with a \emph{small} hat $\hat{(\cdot)}$ to distinguish them from similar objects on the other side. The Pontryagin dual of an abelian group $A$ is denoted by $\wh{A}$ (\emph{wide} hat). Equations involving $K$-theory groups $K^\bullet(\cdot)$ hold for each $\bullet \in \ZZ/2$.

\section{Generalities on crystallographic space groups}\label{sec:crystalgeneral}
Let $R^d$ be $d$-dimensional (affine) Euclidean space, which can be identified with its vector group $\RR^d$ of translations upon choosing an origin. The Euclidean group $\mathscr{E}(d)$ of isometries of $R^d$ is then isomorphic to the semidirect product $\mathscr{E}(d)\cong \RR^d\rtimes {\rm O}(d)$ where ${\rm O}(d)$ is the orthogonal group fixing the origin. 

\begin{definition}[e.g.\ \cite{Hiller, Schwarz}]
A $d$-dimensional \emph{crystallographic space group}, or simply \emph{space group}, is a discrete cocompact subgroup $\mathscr{G}\subset\mathscr{E}(d)$.
\end{definition}
 From various classical theorems of Bieberbach \cite{Bie}, the \emph{lattice} $\Pi\coloneqq\mathscr{G}\cap \RR^d$ of translations in $\mathscr{G}$ is free abelian of rank $d$ (so isomorphic to $\ZZ^d$), with \emph{finite} quotient $G=\mathscr{G}/\Pi$. In fact, an abstract group $\mathscr{G}$ is characterised as a $d$-dimensional space group, by virtue of it having a finite-index normal free abelian subgroup of rank $d$ which is maximal abelian \cite{Zass}.

To summarise, there is a commutative diagram of groups
$$
\begin{CD}
0@>>>\RR^d@>>>\mathscr{E}(d)@>>>{\rm O}(d)@>>>1
 \\
@.@AAA @ AAA @ AA\rho A \\
0@>>>\Pi\cong\ZZ^d@>>> \mathscr{G} @>>> G @ >>> 1
\end{CD}
$$
in which the vertical maps are inclusions. In crystallography, the subgroup $\rho:G\hookrightarrow{\rm O}(d)$, is called the \emph{point group}. Note that $\mathscr{G}$ need not be isomorphic to a semi-direct product $\ZZ^d\rtimes G$; if it is, we say that $\mathscr{G}$ is \emph{symmorphic}.

\subsection{Actions of the point group}\label{sec:pointgroupaction}
\subsubsection{Linear and affine actions on position space torus}
Since $\Pi$ is normal in $\mathscr{G}$, the Euclidean action of $\mathscr{G}$ on $R^d$ descends to the $\Pi$-orbit space $T^d=R^d/\Pi$, which is an affine torus with translation group $\RR^d/\Pi \cong\TT^d$. Thus, there is a homomorphism $G=\mathscr{G}/\Pi\rightarrow {\rm Isom}(T^d)$. 

More concretely, pick an origin for $R^d$ and thus $T^d$. Then ${\rm Isom}(T^d)\cong \TT^d\rtimes {\rm Aut}(\TT^d)$ and we obtain a homomorphism
\begin{equation}
\bm{\alpha}\equiv(s,\alpha):G\rightarrow \TT^d\rtimes {\rm Aut}(\TT^d),\label{affinetorusaction}
\end{equation}
with $s:G\rightarrow\TT^d$ the ``translational'' part and $\alpha:G\rightarrow{\rm Aut}(\TT^d)$ the ``linear'' part fixing the origin. This terminology is based on the following. The linear action of $g\in G\subset{\rm O}(d)$ on $\RR^d$ taking $x\mapsto gx\in \RR^d$ is, inside the Euclidean group, implemented by choosing a section
 \begin{equation}
 g\mapsto \wt{g}\equiv(\wt{s}(g),g)\in\mathscr{G}\subset\RR^d\rtimes{\rm O}(d)=\mathscr{E}(d)\label{splittingmap}
 \end{equation}
 and then conjugating $(x,1)\in \mathscr{E}(d)$ by $(\wt{s}(g),g)$ to get $(gx,1)$. This conjugation preserves the subgroup of lattice translations $\Pi$, so there is an action $\alpha:G\rightarrow{\rm Aut}(\TT^d)$ on the quotient group $\TT^d=\RR^d/\Pi$. Note that the lifts $\wt{g}$, and thus the translational parts $\wt{s}(g)\in\RR^d$, are specified up to a $\Pi$ ambiguity, so there is a well-defined map $s:G\rightarrow\TT^d$. Points of $T^d$ are labelled, with respect to the origin, by equivalence classes $[x]\in\TT^d$ of translations $x\in\RR^d$ modulo $\Pi$. Then the affine $G$-action $\bm{\alpha}$ on $T^d$ can be written as
\begin{equation}
\bm{\alpha}(g)[x]\coloneqq[\wt{g}\cdot x]\equiv[\wt{s}(g)+gx]=[\wt{s}(g)]+[gx]=s(g)+\alpha(g)[x]=(s(g),\alpha(g))[x].\label{affinetorusaction2}
\end{equation}

\begin{definition}\label{defn:spacegrouptorus}
Let $\mathscr{G}$ be a $d$-dimensional space group with point group $G\subset{\rm O}(d)$. We write $T^d_{\mathscr{G}}$ for the quotient torus $T^d=R^d/\Pi$ equipped with the induced affine $G$-action $\bm{\alpha}$ described above. 
\end{definition}

We can identify $\Pi$ with $H_1(\TT^d,\ZZ)$, and ${\rm Aut}(\TT^d)$ with ${\rm Aut}(\Pi)$ via the induced map on homology. Upon choosing a (not necessarily orthonormal) basis for $\Pi\cong\ZZ^d$, the injective homomorphism $\alpha:G\rightarrow{\rm Aut}(\TT^d)={\rm Aut}(\Pi)\cong {\rm GL}(d,\ZZ)$
expresses $G$ as a finite subgroup of integral $d\times d$ matrices. The conjugacy class of $G$ in ${\rm GL}(d,\ZZ)$ is called the \emph{arithmetic crystal class}. There may be several non-isomorphic space groups within the same arithmetic crystal class, due to possible translational parts in $\bm{\alpha}\equiv(s,\alpha)$, and $s=0$ gives the symmorphic $\ZZ^d\rtimes_\alpha G$. 

The fundamental domain $\mathscr{G}\backslash R^d=G\backslash(R^d/\Pi)=G\backslash T^d$ is an \emph{orbifold} that is finitely covered by $T^d$. If $\bm{\alpha}$ happens to be a \emph{free} action of $G$ (thus $\mathscr{G}$ acts freely on $R^d$, is torsion-free, and is nonsymmorphic), then the fundamental domain is a flat \emph{manifold}. An orbifold approach to crystallography can be found in \cite{CFHT}.

\noindent
{\bf 1D examples.}
The 1D torus $T^1=S^1= \{ u \in \C |\ \lvert u \rvert = 1 \}$ admits three inequivalent $\ZZ_2$-actions: $S^1_{\mathrm{triv}}$, $S^1_{\mathrm{flip}}$ and $S^1_{\mathrm{free}}$ have the trivial action $u \mapsto u$, the flip action $u \mapsto \bar{u}$ and the free action $u \mapsto -u$ respectively.

The two space groups in 1D are $\ZZ$, and $\ZZ\rtimes\ZZ_2$ with $\ZZ_2$ acting by reflection on $\RR\supset\ZZ$, and are sometimes referred to as $\sf{p1},\sf{p1m1}$ respectively. The $G$-tori $R^1/\Pi=T^1$ are respectively $T^1_{\sf{p1}}=S^1$ (trivial $G$), and the involutive space $T^1_{\sf{p1m1}}=\Sflip$. 

The two other involutive circles $\Striv, \Sfree$ do not come directly from 1D space groups, but they appear in a generalisation to \emph{frieze groups} (Section \ref{sec:frieze}).

\noindent
{\bf 2D examples.}
There are 13 arithmetic crystal classes in $d=2$, whereas there are 17 wallpaper groups (2D space groups); the four extra ones are nonsymmorphic, see Table \ref{table:wallpapergroups}. For example, $\sf{pm}\cong\ZZ^2\rtimes\ZZ_2$ has point group $\ZZ_2$ acting by reflection in one coordinate, while the nonsymmorphic version $\sf{pg}$ has instead a \emph{glide reflection} --- a reflection followed by half a lattice translation along the orthogonal coordinate --- which squares to a lattice translation, so it is of infinite order (\S\ref{sec:pmpgcm}). The quotient of the 2-torus $T^2_{\sf{pg}}$ by its free involution $\bm{\alpha}$ is the Klein bottle, whose torsion-free fundamental group recovers $\sf{pg}$ .
\begin{table}
$$
\begin{array}{|c|c|c|c|c|c|c|c|c|c|c|c|c|}
\hline
\sf{p1} & \sf{p2} & \sf{p3} & \sf{p4} & \sf{p6} & \sf{pm} & \sf{cm} & \sf{pmm} & \sf{cmm} & \sf{p3m1} & \sf{p31m} & \sf{p4m} & \sf{p6m} \\
 & & & & &\sf{pg} & &\sf{pmg} & & & &\sf{p4g} & \\
 & & & & & & &\sf{pgg} & & & & & \\
\hline
1 & \ZZ_2 & \ZZ_3 & \ZZ_4 & \ZZ_6 & D_1 & D_1  & D_2 & D_2 & D_3 & D_3 & D_4 & D_6  \\
\hline
\end{array}
$$
\caption{The thirteen arithmetic crystal classes, with symmorphic representatives listed first. $D_n$ is the dihedral group of order $2n$. Note that we have written $D_1$ here for the point group of $\sf{pm}, \sf{pg}, \sf{cm}$ to emphasise that it contains a \emph{reflection}, whereas $\sf{p2}$ with isomorphic point group $\ZZ_2$ contains a \emph{rotation}.}\label{table:wallpapergroups}
\end{table}

\subsubsection{Crystallography and group cohomology}
Unlike $\alpha$, the map ${s}$ is not generally a homomorphism but satisfies the condition $${s}(g_1g_2)={s}(g_1)+\alpha(g_1)({s}(g_2)).$$
Thus ${s}$ is a \emph{group 1-cocycle} with values in $\TT^d$ (regarded as a $G$-module via $\alpha$). A different choice of origin shifted by $t\in \TT^d$ causes $({s}(g),\alpha(g))\in \TT^d\rtimes{\rm Aut}(\TT^d)$ to be conjugated by $(t,1)$ into $(s'(g),\alpha(g))=(t+{s}(g)-\alpha(g)(t),\alpha(g))$, thereby modifying ${s}$ by the 1-coboundary $g\mapsto t-\alpha(g)(t)$. Therefore, it is only the \emph{cohomology class} $[{s}]\in H^1_{\rm group}(G,\TT^d)$ which matters. 

We may specify all the possible space groups within an arithmetic crystal class $\alpha:G\hookrightarrow{\rm GL}(d,\ZZ)$ by specifying a group cohomology class\footnote{Technically, we are considering space groups up to isomorphism, and there is a redundancy given by an action on $H^1_{\rm group}(G,\TT^d)$ of the normaliser of $G$ in ${\rm Aut}(\TT^d)$.} $[s]\in H^1_{\rm group}(G,\TT^d)$, see e.g.\ Theorem 5.2 in \cite{Hiller}. Via the connecting homomorphism $\delta$ coming from the exact sequence of $G$-modules $0\rightarrow \Pi \rightarrow\RR^d\rightarrow\TT^d\rightarrow 0$, we have $H^1_{\rm group}(G,\TT^d)\cong H^2_{\rm group}(G,\Pi)$, so that we are equivalently looking for inequivalent extensions of the point group $G$ by $\Pi$ (e.g.\ \S3.4 of \cite{Schwarz}, remark after Theorem 5.2 of \cite{Hiller}). Explicitly, a lifting map $\wt{s}$ as in Eq.\ \eqref{splittingmap} determines the 2-cocycle
\begin{equation}
\nu(g,h)\coloneqq\delta s(g,h)\equiv\wt{s}(g)+g\cdot\wt{s}(h)-\wt{s}(gh)\in\Pi,\label{2cocycleforspacegroup}
\end{equation}
which twists the product rule in $\Pi\times_\alpha G$ to give the space group $\mathscr{G}$ as an extension of $G$ by $\Pi$. The extension is symmorphic iff the cocycle class of $\nu$ is trivial. Starting from $\mathscr{G}$, we see that its $\Pi$-valued 2-cocycle is
\begin{align*}
\nu'(g,h)\coloneqq \wt{g}\wt{h}\wt{gh}^{-1} &\equiv (\wt{s}(g),g)(\wt{s}(h),h)(\wt{s}(gh),gh)^{-1}\\
&=(\wt{s}(g)+g\cdot\wt{s}(h)-\wt{s}(gh),1)
\end{align*}
recovering the formula Eq.\ \eqref{2cocycleforspacegroup}. To emphasise that (the class of) $s$ is determined by $\mathscr{G}$, we shall sometimes write the affine torus action as $\bm{\alpha}=(s_\mathscr{G},\alpha)$.

\begin{table}
 \begin{tabular}{|| c | c | c | c | c | c | c | c ||} 

\end{tabular}
\end{table}

\subsubsection{Dual action on Brillouin torus}
Under Pontryagin duality ${\rm Hom}(\cdot,{\rm U}(1))\equiv \wh{(\cdot)}$, conventionally denoted with a \emph{wide} hat, the sequence of abelian groups 
$$0\rightarrow \Pi\cong\ZZ^d\rightarrow \RR^d\rightarrow \TT^d\rightarrow 0$$
is \emph{self-dual}, in that the dual sequence $0\rightarrow \wh{\TT^d}\rightarrow \wh{\RR^d}\rightarrow \wh{\Pi}\rightarrow 0$ again has $\wh{\TT^d}$ a lattice in $\wh{\RR^d}\cong\RR^d$ and quotient $\wh{\Pi}$ another torus. In fact, the dual lattice $\wh{\TT^d}$ is just the annihilator subgroup $\Pi^\perp\subset\wh{\RR^d}\equiv {\rm Hom}(\RR^d,{\rm U}(1))$ for $\Pi$ (so that they indeed define characters of $\TT^d=\RR^d/\Pi$). 

In solid state physics, $\TT^d=\RR^d/\Pi$ (actually the affine torus $T^d=R^d/\Pi$) is called the position space \emph{unit cell}, while $\Pi^\perp$ is the \emph{reciprocal lattice} in momentum space $\wh{\RR^d}$ with quotient $\wh{\RR^d}/\Pi^\perp\cong\wh{\Pi}$ the \emph{Brillouin torus} of crystal/quasi-momenta. 

{\bf Notation.} We will mostly be regarding the Brillouin torus $\wh{\Pi}$ as the T-dual topological space to the unit cell $T^d$, in which case we write it as $\hat{T}^d$ where the \emph{small} hat notation is usual in the string theory literature. This is not to be confused with taking the Pontryagin dual of the Brillouin torus as an abelian group (which would give back the lattice $\Pi$).

Given $\alpha:G\rightarrow{\rm GL}(d,\ZZ)\cong{\rm Aut}(\Pi)$, there is a canonical (linear) \emph{dual} action $\hat{\alpha}$ of $G$ on the Brillouin torus $\hat{T}^d=\wh{\Pi}$, defined in the usual way: $\hat{\alpha}(g)(\chi)=\chi\circ \alpha(g^{-1}), \,\,\chi\in \hat{T}^d=\wh{\Pi}$. For convenience, we also write this as the equation
\begin{equation*}
(g\cdot\chi)(n)=\chi(g^{-1}\cdot n), \qquad g\in G,\; \chi\in\hat{T}^d=\wh{\Pi},\; n\in \Pi.\label{dualaction}
\end{equation*}

\begin{remark}
Since we can think of $\hat{\alpha}$ as $\hat{\alpha}:G\rightarrow {\rm Aut} (\Pi^\perp)\cong{\rm GL}(d,\ZZ)$, there are \emph{two} maps $\alpha, \hat{\alpha}$ into $GL(d,\ZZ)$ (upon choosing bases for $\Pi, \Pi^\perp$), which are \emph{not} necessarily conjugate. This is due to the fact that a subgroup of $GL(d,\ZZ)$ need not be conjugate to its contragredient subgroup (inverse transpose). In $d=2$, $\alpha$ and $\hat{\alpha}$ are conjugate for any space group, except for $\sf{p31m}$ and $\sf{p3m1}$ (Lemma 2.4 of \cite{Gomi2}), see Section \ref{sec:p31mp3m1}. In higher dimensions, the general relation between $\alpha$ and $\hat{\alpha}$ appears to be difficult to ascertain, but see \cite{Michel} for some 3D examples.
\end{remark}

\subsubsection{Dual cocycle on Brillouin torus}\label{sec:dualcocycle}
Whether $\mathscr{G}$ is symmorphic or not, the Brillouin torus $\hat{T}^d$ is a $G$-space under $\hat{\alpha}$ with no translational part, so it is itself associated to the \emph{symmorphic} space group for the \emph{dual} arithmetic crystal class $\hat{\alpha}$ of $\mathscr{G}$. To achieve a full duality, the nonsymmorphicity data $s$ should also appear on the Brillouin torus side.

Let us write $g\cdot\chi\coloneqq \hat{\alpha}(g)(\chi)$ for $\chi\in\wh{\Pi}=\hat{T}^d$ to simplify notation. The group 2-cocycle $\nu:G\times G\rightarrow\Pi$ for $\mathscr{G}$ has a Fourier transformed version as a ${\rm U}(1)$-valued function $\tau_\mathscr{G}:G\times G\rightarrow C(\wh{\Pi},{\rm U}(1))\equiv{\rm U}(C(\hat{T}^d))$; explicitly,
\begin{equation}
\tau_\mathscr{G}(g_1,g_2)(\chi)=(g_1g_2\cdot\chi)(\nu(g_1,g_2)) \in {\rm U}(1),\qquad g_1,g_2\in G.\label{dualcocycle}
\end{equation}
The algebra $C(\hat{T}^d)$ and also its unitary group ${\rm U}(C(\hat{T}^d))$ admit a natural left action of $G$ by taking 
$(g\cdot f)(\chi)\coloneqq f(g^{-1}\cdot\chi),\;f\in C(\hat{T}^d)$.
Then we see that $\tau_\mathscr{G}$ is a group 2-cocycle with values in the $G$-module ${\rm U}(C(\hat{T}^d))$. As explained in \S\ref{sec:H3twistexamples}, the class of $\tau_\mathscr{G}$ in $H^2_{\rm group}(G,{\rm U}(C(\hat{T}^d)))$ can be regarded as a $G$-equivariant twist $\tau_\mathscr{G}\in H^3_G(\hat{T}^d,\ZZ)$. 

{\bf Idea of crystallographic T-duality.} 
The nonsymmorphicity data of a space group appears on the position space side in the affine action, whereas it is a $K$-theory twist on the momentum space side. The basic idea behind crystallographic T-duality for a space group $\mathscr{G}$, is that the position space data $\bm{\alpha}=(s_\mathscr{G},\alpha)$ determines dual data $(\hat{\alpha},\tau_\mathscr{G})$ in momentum space, and that despite this drastic-looking change, the $G$-equivariant $K$-(co)homology theories adapted to $(T^d,(s_\mathscr{G},\alpha))$ and $(\hat{T}^d,(\hat{\alpha},\tau_\mathscr{G}))$ are isomorphic in a natural way. Its precise statement requires a discussion of graded $K$-theory twists.

\section{Generalities on twistings of $K$-theory}\label{sec:twistgeneral}
In what follows, we are concerned with compact spaces $E$ equipped with continuous actions of a finite group $G$. The transformation groupoid $E/\!\!/G$ is a special case of a local quotient groupoid, and its complex (equivariant) $K$-theory has a category $\mathfrak{Twist}(E/\!\!/G)$ of twists, whose isomorphism classes $\pi_0(\mathfrak{Twist}(E/\!\!/G))$ fit into a short exact sequence of groups \cite{FHT1},
$$1\longrightarrow H^3_G(E,\ZZ)\longrightarrow \pi_0(\mathfrak{Twist}(E/\!\!/G))\longrightarrow H^1_G(E,\ZZ_2)\longrightarrow 1.$$
Under the bijection $\pi_0(\mathfrak{Twist}(E/\!\!/G))=H^3_G(E,\ZZ)\times H^1_G(E,\ZZ_2)\ni (h,w)$ of \emph{sets}, the group law is
\begin{equation}
(h_1,w_1)(h_2,w_2)=(h_1+h_2+\beta(w_1\cup w_2),w_1+w_2),\label{twistcomposition}
\end{equation}
where $\beta:H^2_G(E,\ZZ_2)\rightarrow H^3_G(E,\ZZ)$ is the Bockstein homomorphism associated to the mod 2 reduction $\ZZ\rightarrow\ZZ_2$ of coefficients. The \emph{ungraded twists} have $w=0$. 

In the non-equivariant case, $h\in H^3(E,\ZZ)$ is the Dixmier--Douady invariant when the $K$-theory of $E$ twisted by $h$ is modelled by gerbes \cite{BCMMS} or by continuous-trace algebras \cite{Rosenberg}, and is also called a H-flux in string theory. In solid state physics, it arises after a partial Fourier transform of a screw-dislocated lattice \cite{HMT}, see Section \ref{sec:partialTflux}. In the equivariant world, $H^3_G(E,\ZZ)$ need not vanish even if dim$(E)<3$. In fact, there is a nice interpretation of the various ``lower-dimensional'' terms that appear in the Leray--Serre spectral sequence computing $H^3_G(E,\ZZ)$ \cite{Gomi2}, intimately related to crystallographic groups when $E$ is a torus. Some examples are given in Section \ref{sec:H3twistexamples}.

Much less studied are the \emph{gradings} $w\in H^1_G(E,\ZZ_2)$ of twists. This cohomology group classifies $G$-equivariant real line bundles, or equivalently $G$-equivariant principal ${\rm O}(1)$ bundles over $E$, which we can think of as an ``orientation field''. Graded twists are required in crystallographic T-duality because of orientation reversing operations like reflections; they are needed to equivariantly implement push-forwards and Poincar\'{e} duality.

\subsection{Examples of special equivariant $H^3$ twists}\label{sec:H3twistexamples}
{\bf Twist from ${\rm U}(1)$ central extension of group.} For a finite group $G$ acting on a point, $H^3_G(\pt ,\ZZ)\cong H^3(BG,\ZZ)\cong H^3_{\rm group}(G,\ZZ)\cong H^2_{\rm group}(G,{\rm U}(1))$, so the equivariant $H^3$-twists come from ${\rm U(1)}$-valued group 2-cocycles of $G$, i.e.\ central extensions of $G$ by ${\rm U}(1)$. This cannot occur for $G=\ZZ_n$, but $D_2=\ZZ_2\times\ZZ_2$ has a non-trivial 2-cocycle $\omega$ specified by
\begin{equation}
\omega(((-1)^{k_1},(-1)^{k_2}),((-1)^{l_1},(-1)^{l_2}))=(-1)^{k_2l_1},\label{2cocyclepoint}
\end{equation}
which generates $H^3_{D_2}(\pt ,\ZZ)\cong\ZZ/2$. When a $D_2$-space $E$ has a fixed point, the pullback of $H^3_{D_2}$ from a point to $E$ is split injective and we continue to write the pullback as $\omega\in H^3_{D_2}(E,\ZZ)$. Similarly, if $G\rightarrow D_2$ splits, we continue to write $\omega$ for its pullback in $H^3_G(\pt ,\ZZ)$.

For the dihedral groups $D_n$ of order $2n$, which appear as point groups in crystallography, it is known that (e.g.\ Theorem 5.2 of \cite{Handel})
\begin{equation*}
H^3_{D_n}(\pt ,\ZZ)=H^3_{\rm group}(D_n,\ZZ)=\begin{cases}0 \qquad\, n\,\,{\rm odd},\\ \ZZ/2\quad n\,\,{\rm even}. \end{cases}\label{dihedralcohomology}
\end{equation*} 
For $n$ even, there are split surjections $D_n\rightarrow D_2$ so $\omega$ generates $H^3_{D_n}(\pt ,\ZZ)$.

{\bf Twist from group 2-cocycle in ${\rm U}(C(E))$.}
We have $H^3_{\ZZ_2}(\Sfree,\ZZ)\cong H^3(S^1,\ZZ)=0$, while $H^3_{\ZZ_2}(\Sflip,\ZZ)\cong 0$ and $H^3_{\ZZ_2}(\Striv,\ZZ)\cong \ZZ/2$ are recalled in \S\ref{sec:cohomologycalculations}. The generating twist for the latter is represented by a group 2-cocycle $\tau_{S^1}$ for $\ZZ_2$ with coefficients in the (trivial) $\ZZ_2$-module ${\rm U}(C(\Striv))$,
\begin{equation}
\tau_{S^1}(-1,-1)= \{k\mapsto e^{\im k}\}\in {\rm U}(C(S^1_{\rm triv})).\label{2cocyclecircle}
\end{equation}
In fact, $\tau_{S^1}$ is the Fourier transformed version, in the sense of \S\ref{sec:dualcocycle}, of the $\ZZ$-valued group 2-cocycle $\nu$ corresponding to the extension
$$0\rightarrow\ZZ\overset{\times 2}{\longrightarrow}\ZZ\overset{(-1)^{(\cdot)}}{\longrightarrow}\ZZ_2\rightarrow 1,$$
which has value $\nu(-1,-1)=1$ and 0 otherwise.

Generally, for a $\ZZ_2$-space $E$ with an equivariant map to $\Striv$, we will continue to write $\tau_{S^1}$ for its pullback to $E$, unless several such maps are possible and cause ambiguity. 

{\bf Non-cocycle twists.} For $T^2$, there may be general $G$-equivariant $H^3$-twists $h$ which are not represented by a 2-cocycle. Examples arising from crystallography will be discussed in Section \ref{sec:partialTflux}. We mention that these non-cocycle twists can be represented by central extensions of the groupoid $T^2/\!\!/G$, see \cite{Gomi2} for a discussion. In higher dimensions, there may also be $H^3$ twists which are non-equivariantly nontrivial.

\subsection{Examples of special equivariant $H^1$ twists}\label{sec:H1twistexamples}
{\bf $d=0$: Twist from grading of group.} For a finite group $G$ acting on a point, $H^1_G(\pt ,\ZZ_2)\cong {\rm Hom}(H_1(BG,\ZZ),\ZZ_2)\cong{\rm Hom}(G,\ZZ_2)$, so a $H^1$-twist is essentially a homomorphism $c:G\rightarrow\ZZ_2$ making $G$ into a \emph{graded} group. The equivariant line bundle over $\pt$ is just $\pt \times \RR$ with $g\in G$ acting by multiplication by $c(g)$. 

The pullback of $c$ from $\pt$ to a $G$-space $E$ gives a twist in $H^1_G(E,\ZZ_2)$, corresponding to the $G$-equivariant product line bundle $E\times\RR$ with $G$ action taking the fibre over $x\in E$ to the fibre over $g\cdot x$ followed by multiplication by $c(g)$; we call such twists \emph{$c$-type}. When there is only one possible surjective homomorphism, e.g.\ when $G$ is $\ZZ_2$ or $D_3$, we will just write $c$ for the unique nontrivial $H^1$-twist of $c$-type.

For a space group $\mathscr{G}$ with point group $\rho:G\rightarrow{\rm O}(d)$, there is a distinguished twist from the orientability homomorphism
\begin{equation}
c_\mathscr{G}:G\ni g\mapsto {\rm det}\,\,g \in \ZZ_2.\label{orientationhom}
\end{equation} 
Note that $\mathscr{G}, \mathscr{G}'$ can have the same $G$ as an abstract group but different $c_\mathscr{G}$, e.g.\ $\sf{p2}$ contains rotations so $c_{\sf{p2}}$ is the trivial map, while $\sf{pm}$ contains a reflection so $c_{\sf{pm}}$ is the identity map on $G\cong\ZZ_2$.

\medskip
\noindent
{\bf $d=1$: M\"{o}bius -type equivariant twists on circles.}
 If $E$ has a $G$-fixed point, there is a splitting $H^1_G(E,\ZZ_2)=H^1_G(\pt ,\ZZ_2)\oplus \wt{H}^1_G(E,\ZZ_2)$, and the reduced part may be represented by \emph{non-trivial} line bundles over $E$ made $G$-equivariant. This is a consequence of a calculation with the Leray--Serre spectral sequence. In such cases, equivariant $H^1$-twists come in $c$-type (from a point), in ``$M$-type'' (for {\bf M}\"{o}bius, coming from the reduced part), or a sum of the two types.
 
For $\Striv$, we have
$$H^1_{\ZZ_2}(\Striv,\ZZ_2)\cong H^1_{\ZZ_2}(\pt ,\ZZ_2)\oplus\wt{H}^1_{\ZZ_2}(\Striv,\ZZ_2)\cong \ZZ/2\oplus\ZZ/2,$$
where the first generator $c$ comes from $\ZZ_2\overset{\rm id}{\rightarrow}\ZZ_2$, and the second generator $M$ is the M\"{o}bius bundle over $S^1$ with $\ZZ_2$ acting trivially on the total space. The mixed twist $c+M$ is the M\"{o}bius bundle with $-1\in\ZZ_2$ acting fibrewise by multiplication by $-1$.

For $\Sflip$, we have 
$$H^1_{\ZZ_2}(\Sflip,\ZZ_2)\cong \ZZ/2\oplus\ZZ/2,$$ with one generator $c$ as above. The other generator $M$ is the M\"{o}bius line bundle with $\ZZ_2$-action given locally by $(e^{\im k},v)\mapsto(e^{-\im k},v), k\in(-\frac{2\pi}{3},\frac{2\pi}{3})$ and $(e^{\im k},v)\mapsto(e^{-\im k},-v), k\in(\frac{\pi}{3},\frac{5\pi}{3})$. Thus on the fibre over the fixed point $k=0$ (resp.\ $k=\pi$), the $\ZZ_2$-representation is trivial (resp.\ sign). Similarly, the mixed twist $c+M$ is the M\"{o}bius bundle with trivial (resp.\ sign) representation at $k=\pi$ (resp.\ $k=0$).

For $\Sfree$, we have 
$$H^1_{\ZZ_2}(\Sfree,\ZZ_2)\cong H^1(S^1,\ZZ_2)\cong\ZZ/2.$$
The generator $c$ is pulled back from a point; explicitly, take the product bundle $\Sfree\times\RR$ with involution $(e^{\im k},v)\mapsto(-e^{\im k},-v), k\in[0,2\pi]/_{0\sim 2\pi}$.

\subsection{Twisted composition rule for $H^3$ and $H^1$ twists}
An important example where the composition rule for graded twists, Eq.\ \eqref{twistcomposition}, is modified nontrivially, is $\pt /\!\!/D_2$. From the usual cohomologies of $BD_2=B\ZZ_2\times B\ZZ_2$, we obtain the $D_2$-equivariant cohomologies of $\pt$:
$$
\begin{array}{|c|c|c|c|c|c|}
\hline
& n = 0 & n = 1 & n = 2 & n = 3 \\
\hline
\hline
H^n_{D_2}(\pt,\ZZ_2) & \ZZ_2 & \ZZ_2^2 & \ZZ_2^3 & \ZZ_2^4 \\
\hline
\mbox{basis} & 1 & c_1, c_2 & c_1^2, c_1c_2, c_2^2 & c_1^3, c_1^2c_2, c_1c_2^2, c_2^3  \\
\hline
\hline
H^n_{D_2}(\pt,\ZZ) & \ZZ & 0 & \ZZ_2^2 & \Z_2  \\
\hline
\mbox{basis} & 1 & 0 & t_1, t_2 & \omega \\
\hline
\end{array}
$$
Here, the generators $c_i\in H^1_{D_2}(\pt,\ZZ_2)\cong\ZZ_2^2$ come from the $i$-th projections $p_i:D_2=\ZZ_2\times\ZZ_2\rightarrow\ZZ_2,\; i=1,2$. We had already seen that $H^3_{D_2}(\pt ,\ZZ)\cong\Z_2$, generated by the group 2-cocycle $\omega$ in Eq.\ \eqref{2cocyclepoint}. 
\begin{proposition}
For the composition of graded twists of $\pt /\!\!/D_2$, we have $(0,c_i)+(0,c_i)=(0,0), i=1,2$, but 
\begin{equation}
(0,c_1)+(0,c_2) = (0,c_2)+(0,c_1)=(\omega,c_1+c_2).\label{D2twistmodified}
\end{equation}
\end{proposition}
\begin{proof}
Putting the $D_2$-equivariant cohomology groups for $\pt$ tabulated above into the Bockstein sequence
$$
\cdots\rightarrow H^n_{D_2}(\pt,\ZZ)\overset{\times 2}{\rightarrow} H^n_{D_2}(\pt,\ZZ)\rightarrow H^n_{D_2}(\pt,\ZZ_2)\overset{\beta}{\rightarrow}H^{n+1}_{D_2}(\pt,\ZZ)\rightarrow\cdots,
$$
we find that $\beta(c_i)=x_i, \beta(c_i^2)=0, i=1,2$ and $\beta(c_1c_2)=\omega$.
\end{proof}
Although $\pi_0(\mathfrak{Twist}(\pt/\!\!/D_2))\cong \ZZ_2^3$, it is not split as $H^3_{D_2}(\pt,\ZZ)\times H^1_{D_2}(\pt,\ZZ_2)$.

\subsection{Equivariant $K$-orientability of torus}
Non-equivariantly, the torus $T^d$ is ${\rm Spin}^c$ thus $K$-oriented, and Poincar\'{e} duality has the simple form $K_\bullet(T^d)\cong K^{d-\bullet}(T^d)$. If $T^d$ is replaced by an oriented compact $d$-manifold $M$, a twist by the ${\rm Spin}^c$ obstruction class $W_3(M)\in H^3(M,\ZZ)$ is needed on the $K$-theory side, e.g.\ Prop.\ 9.2 of \cite{ABG}.

When $T^d$ has an action $\bm{\alpha}$ of a finite group $G$, the obstruction to it being oriented in $G$-equivariant $K$-theory is
\begin{equation}
\sigma\equiv(W_3^G(T^d), w_1^G(T^d))\in H^3_G(T^d,\ZZ)\times H^1_G(T^d,\ZZ_2),\label{Korientabilityobstruction}
\end{equation}
where $W_3^G(T^d)$ and $w_1^G(T^d)$ are respectively the $G$-equivariant third integral Stiefel--Whitney class and first Stiefel--Whitney class of the tangent bundle $\mathcal{T}T^d$ of $T^d$. The Poincar\'{e} duality in this case is
\begin{equation}
K_\bullet^G(T^d)\cong K^{d-\bullet+\sigma}_G(T^d),\label{poincareduality}
\end{equation}
a special case of general dualities, e.g.\ Theorem 2.1 of \cite{Tu}, Theorem 2.9 of \cite{EEK}.

We are interested in $T^d_{\mathscr{G}}$ as in Definition \ref{defn:spacegrouptorus}, i.e.\ the affine $G$-actions $\bm{\alpha}$ on $T^d$ arising from a space group $\mathscr{G}$, as explained in Section \ref{sec:pointgroupaction}. Recall that $g\in G$ is an orthogonal transformation under $\rho:G\hookrightarrow {\rm O}(d)$, and via a lift $\wt{g}=(\wt{s}(g),g)\in\mathscr{G}\subset \RR^d\rtimes O(d)$ specified by a splitting map $\wt{s}:G\rightarrow\RR^d$, there is the induced $G$-action $\bm{\alpha}$ on $T^d_{\mathscr{G}}$ by $\bm{\alpha}(g)[x]=[\wt{s}(g)+gx]$ as in Eq.\ \eqref{affinetorusaction2}. 

\begin{lemma}
Let $\mathscr{G}$ be a space group and $\bm{\alpha}$ be the associated affine action of the point group $G\subset{\rm O}(d)$ on $T^d_{\mathscr{G}}$. The tangent bundle of $T^d_{\mathscr{G}}$ is a $G$-equivariant vector bundle isomorphic to $T^d_{\mathscr{G}}\times \RR^d$ with the product $G$-action.
\end{lemma}
\begin{proof}
 The tangent bundle $\mathcal{T}R^d$ of $R^d$ is trivialised by the translation action of $\RR^d$, and so an element $(r,\mathcal{O})\in\RR^d\rtimes O(d)= \mathscr{E}(d)$ acts on $\mathcal{T}R^d=R^d\times\RR^d$ by $(x,v)\mapsto(r+\mathcal{O}x,\mathcal{O}v)$ in this trivialisation. In particular, a lift $\wt{g}$ of $g\in G$ in $\mathscr{G}$ acts this way, and passing to quotients we get 
 $$\mathcal{T}T^d_{\mathscr{G}}=\mathcal{T}R^d/\Pi=T^d_{\mathscr{G}}\times\RR^d\ni([x],v)\overset{g}{\mapsto}([\wt{s}(g)+gx],gv)=(\bm{\alpha}(g)[x],gv).$$
 \end{proof}
From this lemma, the obstruction classes $W_3^G(T^d_{\mathscr{G}}), w_1^G(T^d_{\mathscr{G}})$ are respectively the pullback of
\begin{equation}
W_3^G(\RR^d_\rho)\in H^3_G(\pt ,\ZZ),\qquad w_1^G(\RR^d_\rho)\in H^1_G(\pt ,\ZZ_2), \label{W3obstructionpoint}
\end{equation}
where $\RR^d_\rho\rightarrow\pt $ is the $G$-representation given by $\rho:G\hookrightarrow {\rm O}(d)$. Then $w_1^G(\RR^d_\rho)$ vanishes iff $\rho$ is orientable, i.e.\ factors through ${\rm SO}(d)$. Thus under the identification $H^1_G(\pt ,\ZZ_2)\cong {\rm Hom}(G,\ZZ_2)$, we have $w_1^G(\RR^d_\rho)$ being the orientability homomorphism $c_{\mathscr{G}}$ of Eq.\ \eqref{orientationhom}. 

Similarly, $W_3^G(\RR^d_\rho)$ vanishes iff $\rho$ is ${\rm Pin}^c$, factoring through the projection in
\begin{equation}
1\rightarrow {\rm U}(1)\rightarrow {\rm Pin}^c(d)\rightarrow{\rm O}(d)\rightarrow 1.\label{pincextension}
\end{equation}
It is well-known that for finite $G$, central extensions of $G$ by ${\rm U}(1)$ are classified by $H^2_{\rm group}(G,{\rm U}(1))\cong H^3_G(\pt ,\ZZ)$. Then $W_3^G(\RR^d_\rho)\in H^3_G(\pt ,\ZZ)$ corresponds to the pullback of Eq.\ \eqref{pincextension} under $\rho:G\rightarrow{\rm O}(d)$, which only depends on the point group. For the 2D space groups, only the point groups $G=D_2, D_4, D_6$ have a potential obstruction ($0\neq H^3_G(\pt ,\ZZ)\cong\ZZ/2$ in these cases), and by Lemma \ref{lem:2DW3obstruction} in the Appendix, the obstruction in all these cases are nontrivial and therefore equal to the unique nontrivial $\omega$.

\begin{definition}\label{defn:obstructionclass}
For the $G$-space $T^d_{\mathscr{G}}$ of Definition \ref{defn:spacegrouptorus}, we will write its $K$-orientability obstruction class $\sigma=(W_3^G(T^d_{\mathscr{G}}), w_1^G(T^d_{\mathscr{G}}))$ as $\sigma_\mathscr{G}$.
\end{definition}

\section{Crystallographic T-duality}\label{sec:crystalTdual}

\subsection{T-duality of circle bundles}
\subsubsection{Ordinary T-duality for a circle}
The basic T-duality in string theory exchanges a circle $S^1$ of radius $L$ with a circle $\hat{S}^1$ of radius $\frac{1}{L}$. There are degree-shifted isomorphisms
\begin{equation}
K^\bullet(S^1)\overset{{\rm T}}{\longleftrightarrow} K^{\bullet-1}(\hat{S}^1)\label{circledual}
\end{equation}
It is important to remember that $S^1$ and $\hat{S}^1$ are not canonically the same space, so that Eq.\ \eqref{circledual} is not the (non-trivial) observation that $K^0(S^1)\cong\ZZ\cong K^{-1}(S^1)$, but is rather a type of topological Fourier transform. This is apparent from its implementation as a Fourier--Mukai transform \cite{Hori}, Eq.\ \eqref{FMbasic}.

For applications in solid-state and quantum physics, it is useful to view the two circles $S^1, \hat{S}^1$ as coming from the group $\ZZ$ in two different ways. On the one hand, the affine $S^1$ is the quotient of the Euclidean line $R^1$ by a lattice $\ZZ\subset\RR$ of translations, i.e.\ the unit cell. Since $R^1$ is contractible, $S^1=B\ZZ$ is a classifying space for $\ZZ$. On the other hand, the Pontryagin dual of $\ZZ$ is topologically another circle $\hat{S}^1$, namely the 1D Brillouin torus. Notice that if $S^1$ has radius $L$ in the Euclidean metric, then $\hat{S}^1$ has radius $\frac{1}{L}$ in the dual metric on $\wh{\RR}$, exactly as in the string theory story.

As explained in Section \ref{sec:BC}, this duality between $S^1=B\ZZ$ and $\hat{S}^1=\wh{\ZZ}$ can be regarded in noncommutative topology language as a Baum--Connes isomorphism, which is useful for formulating our crystallographic generalisation in which $\ZZ$ is replaced by a space group.

\subsubsection{Topology change from $H^3$-twists}\label{sec:topologychange}
When trying to formulate the T-dual of a circle bundle $E\rightarrow X$ with some $h\in H^3(E,\ZZ)$, one finds the remarkable fact that the dual circle bundle $\hat{E}\rightarrow X$ may be non-isomorphic to $E$ \cite{BEM}. Furthermore, a dual twist $\hat{h}\in H^3(\hat{E},\ZZ)$ which completes the duality in the reverse direction always exists.

More precisely, for a pair $(E,h)$ on $X$ as above, there exists another pair $(\hat{E},\hat{h})$ on $X$ such that
\begin{equation}
\pi_*h=c_1(\hat{E}),\qquad \hat{\pi}_*\hat{h}=c_1(E),\qquad p^*h=\hat{p}^*\hat{h}\label{chern-flux-exchange}
\end{equation}
where $p:\pi^*\hat{E}\rightarrow E$ and $\hat{p}:\hat{\pi}^*E\rightarrow\hat{E}$ are projections from the fibre product $E\times_X\hat{E}=\pi^*\hat{E}=\hat{\pi}^*E$ as summarised in the diagram
\begin{equation}
\xymatrix@R.8pc@C.8pc{
& E\times_X\hat{E} \ar[dl]_{p}\ar[dr]^{\hat{p}}&  \\
 E \ar[dr]_{\pi}& \mbox{\phantom{$x$}}   & \hat{E} \ar[dl]^{\hat{\pi}}\\
& X & 
}\label{Tdualitydiagram}
\end{equation}
Such dual pairs enjoy the property that there is a T-duality isomorphism
$${\rm T}:K^{\bullet+h}(E)\cong K^{\bullet-1+\hat{h}}(\hat{E}).$$ 
In this way, there is a ``conservation of topological invariants'' on each side of the duality, even though the two sides may look very different.

The special case where $X$ is a point has $E=S^1, \hat{E}=\hat{S}^1$, and recovers the basic circle T-duality in Eq.\ \eqref{circledual}. There are no possible $H^3$-twists, and an explicit formula is the Fourier--Mukai transform
\begin{equation}
{\rm T}^{\rm FM}: K^0(S^1)\ni[\mathcal{L}]\mapsto\hat{p}_*([\mathcal{P}\otimes p^*\mathcal{L}])\in K^{-1}(\hat{S}^1),\label{FMbasic}
\end{equation}
where $\mathcal{P}\rightarrow S^1\times\hat{S}^1$ is the \emph{Poincar\'{e} line bundle}, recalled in \S\ref{sec:Poincarebundle}.

\subsection{Crystallographic T-duality and Baum--Connes assembly}\label{sec:BC}
Recall that the Baum--Connes assembly map \cite{BCH} for a discrete group $\mathscr{G}$ is a map
\begin{equation}
K_\bullet^\mathscr{G}(\underline{E}\mathscr{G})\overset{\mu}{\longrightarrow} K_\bullet(C_r^*(\mathscr{G}))\label{BaumConnes}
\end{equation}
where the left-hand-side is the equivariant $K$-homology (with compact supports) of the universal space $\underline{E}\mathscr{G}$ for proper $\mathscr{G}$-actions, and the right-hand-side is the $K$-theory of the reduced group $C^*$-algebra of $\mathscr{G}$.

A crystallographic space group $\mathscr{G}$ is a discrete subgroup of the Euclidean group $\RR^d\rtimes{\rm O}(d)$, and it also acts properly on $R^d$. Then we see that $R^d$ is an $\underline{E}\mathscr{G}$ (cf.\ \cite{BCH} Section 2), and it has only finite isotropy groups. Furthermore, the lattice subgroup $\Pi$ acts freely, so after quotienting by $\Pi$, we can rewrite
\begin{equation*}
K_\bullet^\mathscr{G}(\underline{E}\mathscr{G})\cong K_\bullet^{\mathscr{G}/\Pi}(R^d/\Pi)=K_\bullet^G(T^d_{\mathscr{G}}).\label{equivKhomology}
\end{equation*}
where the finite point group $G$ acts on $T^d_{\mathscr{G}}$ by $\bm{\alpha}$ given in Eq.\ \eqref{affinetorusaction}.

$C_r^*(\mathscr{G})$ is a noncommutative $C^*$-algebra, but since $\mathscr{G}$ is virtually abelian with twisted product $\mathscr{G}=\Pi\times_{\alpha,\nu} G$, we can understand $C_r^*(\mathscr{G})$ in a ``virtually commutative'' way by decomposing it as a twisted crossed product \cite{PR}
\begin{equation}
C_r^*(\mathscr{G})\cong C^*_r(\Pi)\rtimes_{\alpha,\nu}G\cong C(\hat{T}^d)\rtimes_{(\hat{\alpha},\tau_{\mathscr{G}})}G,\label{twistedcrossedproduct}
\end{equation}
where $\tau_{\mathscr{G}}$ is the Fourier transformed 2-cocycle of Eq.\ \eqref{dualcocycle} valued in the $G$-module ${\rm U}(C(\hat{T}^d))$. The $K$-theory of Eq.\ \eqref{twistedcrossedproduct} turns into a $\tau_{\mathscr{G}}$-twisted G-equivariant $K$-theory of $\hat{T}^d$ via a twisted Green--Julg theorem (Theorem 4.10 in \cite{Kubota}),
\begin{equation}
K_\bullet(C^*_r(\mathscr{G}))\cong K_\bullet(C(\hat{T}^d)\rtimes_{(\hat{\alpha},\tau_{\mathscr{G}})}G)\cong K_{\bullet+\tau_{\mathscr{G}}}^G(C(\hat{T}^d))\cong K_G^{-\bullet+\tau_{\mathscr{G}}}(\hat{T}^d),\label{GreenJulg}
\end{equation}
where $G$ acts on $\hat{T}^d$ via the dual action $\hat{\alpha}$.

The Baum--Connes conjecture is verified for $\mathscr{G}$ so that $\mu$ in \eqref{BaumConnes} is an isomorphism, which we rewrite using Eq.\ \eqref{GreenJulg} as
\begin{equation*}
\mu:K_\bullet^G(T^d_{\mathscr{G}}) \overset{\cong}{\longrightarrow} K^{-\bullet+\tau_{\mathscr{G}}}_G(\hat{T}^d).
\end{equation*}
We can convert the LHS to $K$-theory using $\sigma_\mathscr{G}$-twisted Poincar\'{e} duality Eq.\ \eqref{poincareduality}. Assembling these isomorphisms, we finally obtain:
\begin{theorem}[Crystallographic T-duality]\label{thm:crystalTdual}
Let $T^d_\mathscr{G}$ be the $d$-torus with the affine action $\bm{\alpha}\equiv(s_\mathscr{G},\alpha)$ associated to a space group $\mathscr{G}$, as in Definition \ref{defn:spacegrouptorus}, and let the graded twist $\sigma_\mathscr{G}$ be its $G$-equivariant Spin${}^c$ obstruction class as in Definition \ref{defn:obstructionclass}. Let $\hat{T}^d$ be the $d$-torus with the dual $G$-action $\hat{\alpha}$ and 2-cocycle twist $\tau_{\mathscr{G}}$ as defined in Section \ref{sec:dualcocycle}. Then $(T^d_\mathscr{G},\sigma_\mathscr{G})$ and $(\hat{T}^d,\tau_{\mathscr{G}})$ are \emph{crystallographic T-dual} in the sense that there is an isomorphism
\begin{equation}
{\rm T}_\mathscr{G}:K^{-\bullet+\sigma_\mathscr{G}}_G(T^d_\mathscr{G})\overset{\cong}{\longrightarrow}K^{-\bullet-d+\tau_{\mathscr{G}}}_{G}(\hat{T}^d).\label{crystalTdual}
\end{equation}
\end{theorem}

\begin{remark}\label{rem:extratwist}
Twists from a homomorphism $c:\mathscr{G}\rightarrow G\rightarrow \ZZ_2$ are relevant whenever $(\mathscr{G},c)$ arises as a \emph{graded} group in physical applications. We anticipate that such $c$-twists can be added to both sides of Eq.\ \eqref{crystalTdual}, amounting to the statement that a ``\emph{super Baum--Connes} assembly map'' for $(\mathscr{G},c)$ is an isomorphism. We leave the verification and the application of these conjectures for a future work. 
\end{remark}

\subsubsection{Crystallographic T-duality and Poincar\'{e} bundle}\label{sec:Poincarebundle}
We wish to define a variant of the Fourier--Mukai transform adapted to $\mathscr{G}$,
\begin{equation}
{\rm T}^{\rm FM}_\mathscr{G}:K^{-\bullet+\sigma_\mathscr{G}}_G(T^d_\mathscr{G})\rightarrow K^{-\bullet-d+\tau_{\mathscr{G}}}_{G}(\hat{T}^d).\label{crystalFM}
\end{equation}
Recall that the Poincar\'{e} line bundle $\mathcal{P}\rightarrow \TT^d\times\wh{\Pi}$ is defined as the quotient of the product line bundle $\RR^d\times\wh{\Pi}\times\CC\rightarrow \RR^d\times\wh{\Pi}$ under the $\Pi\cong\ZZ^d$ action,

\begin{equation}
n: \RR^d\times\wh{\Pi}\times\CC\ni (x,\chi,z)\mapsto (x+n,\chi,\chi(n)z),\qquad n\in\Pi.\label{Poincareaction}
\end{equation}
To incorporate $\mathscr{G}$, choose an origin to identify $T^d_\mathscr{G}$ with  $\TT^d=\RR^d/\Pi$, and recall that the $G$-action $\bm{\alpha}=(s_\mathscr{G},\alpha)$ on $T^d_\mathscr{G}\simeq\TT^d$ is obtained by first picking a map $\wt{s}:G\rightarrow\RR^d$ which lifts $G\in g\mapsto\wt{g}=(\wt{s}(g),g)\in\mathscr{G}\subset \RR^d\rtimes{\rm O}(d)$, and then taking $\bm{\alpha}(g)[x]=[\wt{s}(g)+gx]$ for $[x]\in \TT^d\simeq T^d_\mathscr{G}$ (Eq.\ \eqref{affinetorusaction2}). On $\RR^d\times\wh{\Pi}\times\CC$, we can further define a $G$-action $\wt{\gamma}_{\bm{\alpha}}$
\begin{equation}
\tilde{\gamma}_{\bm{\alpha}}(g):\RR^d\times\wh{\Pi}\times\CC\ni(x,\chi,z)\mapsto(\wt{g}\cdot x, g\cdot\chi,z)\equiv (\wt{s}(g)+gx, g\cdot\chi, z), \label{GactionPoincare}
\end{equation}
where we write $g\cdot\chi\equiv \hat{\alpha}(g)(\chi)$.
\begin{theorem}\label{thm:poincare}
$\tilde{\gamma}_{\bm{\alpha}}$ descends to a twisted $G$-action $\gamma_{\bm{\alpha}}:\mathcal{P}\rightarrow\mathcal{P}$ on the Poincar\'{e} line bundle, with cocycle the pullback of $\tau_{\mathscr{G}}^{-1}\in Z^2({\rm U}(C(\wh{\Pi})))$ under $\hat{p}:\TT^d\times\wh{\Pi}\rightarrow\wh{\Pi}$.
\end{theorem}
\begin{proof}
We can readily check the commutativity of the diagram
$$
\xymatrix@R.8pc@C.8pc{
\RR^d\times\wh{\Pi}\times\CC\ar[rrr]^{\tilde{\gamma}_{\bm{\alpha}}(g)} \ar[dd]_{n}& &&  \RR^d\times\wh{\Pi}\times\CC\ar[dd]^{ng} \\
 & & & \\
 \RR^d\times\wh{\Pi}\times\CC\ar[rrr]^{\gamma_{\bm{\alpha}}(g)}&   & & \RR^d\times\wh{\Pi}\times\CC}
$$
for all $g\in G, n\in\Pi$, so that each $\tilde{\gamma}_{\bm{\alpha}}(g)$ gives a bundle map $\gamma_{\bm{\alpha}}(g):\mathcal{P}\rightarrow\mathcal{P}$ covering the $G$-action $\bm{\alpha}\times\hat{\alpha}$ on $\TT^d\times\wh{\Pi}$. Using the action formulae Eq.\ \eqref{GactionPoincare} and Eq.\ \eqref{Poincareaction}, along with Eq.\ \eqref{2cocycleforspacegroup} and Eq.\ \eqref{dualcocycle}, we compute for $g,h\in G$,
\begin{align*}
\tilde{\gamma}_{\bm{\alpha}}(g)\tilde{\gamma}_{\bm{\alpha}}(h)[x,\chi,z]&=\tilde{\gamma}_{\bm{\alpha}}(g)[\wt{s}(h)+hx, h\cdot\chi, z]\\
&= [\wt{s}(g)+g(\wt{s}(h)+hx),gh\cdot\chi, z]\\
&=[\nu(g,h)+\wt{s}(g,h)+ghx, gh\cdot\chi, z]\\
&=[\wt{s}(g,h)+ghx, gh\cdot\chi, gh\cdot\chi(-\nu(g,h))z]\\
&=\tilde{\gamma}_{\bm{\alpha}}(gh)[x,\chi,z]gh\cdot\chi(-\nu(g,h))\\
&=(\tau_\mathscr{G}^{-1}(g,h)(\chi))\tilde{\gamma}_{\bm{\alpha}}(gh)[x,\chi,z],
\end{align*}
verifying that the $\tilde{\gamma}_{\bm{\alpha}}$ action on $\mathcal{P}$ is twisted by the pullback of $\tau_\mathscr{G}^{-1}$ in $Z^2(G,{\rm U}(C(\wh{\Pi})))$.
\end{proof}
Thus $\mathcal{P}$ is a $\tau_\mathscr{G}^{-1}$-twisted $G$-equivariant line bundle, and Theorem \ref{thm:poincare} implies
\begin{corollary}\label{cor:trivialisepoincare}
The dual Poincar\'{e} line bundle $\mathcal{P}^*$ trivialises $\hat{p}^*\tau_\mathscr{G}$ as a twist in $H^3_G(\TT^d\times\wh{\Pi},\ZZ)$.
\end{corollary}

Corollary \ref{cor:trivialisepoincare} allows us to construct a ``Fourier--Mukai'' transform adapted to $\mathscr{G}$, based on the diagram (cf.\ Eq.\ \eqref{Tdualitydiagram}),
$$
\xymatrix@R.8pc@C.8pc{
&\mathcal{P} \ar[d]& \\
& \TT^d\times\wh{\Pi}\ar[dl]_{p}\ar[dr]^{\hat{p}}& \\
\TT^d\ar[dr]_{\pi}& &\wh{\Pi}\ar[dl]^{\hat{\pi}}\\
&\pt &
}
$$
as the composition
$$K_G^{\bullet+\sigma_\mathscr{G}}(\TT^d)\overset{p^*}{\rightarrow}K_G^{\bullet+p^*\sigma_\mathscr{G}}(\TT^d\times\wh{\Pi})\overset{\mathcal{P}^*\otimes}{\rightarrow}K_G^{\bullet+p^*\sigma_\mathscr{G}+\hat{p}^*\tau_{\mathscr{G}}}(\TT^d\times\wh{\Pi})\overset{\hat{p}_*}{\rightarrow}K_G^{\bullet-d+\tau_{\mathscr{G}}}(\wh{\Pi}),$$
where $\sigma_\mathscr{G}$ is the $K_G$-orientability obstruction for $\TT^d$ needed for the push-forward $\hat{p}_*$ along $\TT^d$. Recalling that $T^d_{\mathscr{G}}\simeq \TT^d$ and $\hat{T}^d=\wh{\Pi}$, we obtain the desired map ${\rm T}^{\rm FM}_\mathscr{G}$ in Eq.\ \eqref{crystalFM}.

It is anticipated that ${\rm T}^{\rm FM}_\mathscr{G}$ is an isomorphism, implementing crystallographic T-duality, Eq.\eqref{crystalTdual}, and that twists that are pulled back from the common base (a point in the above case) can be added to both sides, cf.\ Remark \ref{rem:extratwist}. The latter is a common feature of well known dualities like ordinary circle bundle T-duality ${\rm T}$, as well as $\TZtwo, \TR$ recalled below.

\subsection{T-dualities for circle bundles with involution}
There are several variants of ordinary $K$-theory groups that we can apply to `Real' or \emph{involutive} spaces, i.e., spaces $X$ equipped with a continuous $\ZZ_2$ action $x\mapsto\bar{x}$. We shall assume that $X$ is a finite $\ZZ_2$-CW complex for simplicity. There is equivariant $K_{\ZZ_2}$, and also a variant $K_\pm$ introduced by Witten in his study of orientifold string theory \cite{Witten} and studied by Atiyah--Hopkins \cite{AH} in connection with Dirac operators. A definition of $K_\pm$ is
$$ K^\bullet_\pm(X)=K^{\bullet+1}_{\ZZ_2}(X\times\tilde{I},X\times\partial\tilde{I}),$$
where $\tilde{I}$ is the interval $[-1,1]$ and $X\times\tilde{I}$ has the involution $(x,t)\mapsto (\bar{x},-t)$. By a Thom isomorphism, one gets the relation \cite{Gomi1}
\begin{equation}
K^{\bullet+h}_\pm\cong K^{\bullet+(h,c)}_{\ZZ_2}(X),\qquad K^{\bullet+h}_{\ZZ_2}\cong K^{\bullet+(h,c)}_{\pm}(X)\label{KpmKZtworelation}
\end{equation}
where $h\in H^3_{\ZZ_2}(X,\ZZ)$ and $c\in H^1_{\ZZ_2}(X,\ZZ_2)$ is the graded twist coming from the unique nontrivial homomorphism $\ZZ_2\overset{{\rm id}}{\rightarrow}{\ZZ_2}$.

There is also Atiyah's $KR$-theory \cite{A-KR}, constructed out of complex vector bundles equipped with \emph{antilinear} involutions lifting the involution on the base. It turns out that $\Striv$ and $\hSflip$ are T-dual in this context, in that there is a naturally defined isomorphism between $KR^\bullet(\Striv)\equiv KO^\bullet(\Striv)$ and $KR^{\bullet-1}(\hSflip)$. Such dualities were studied in the context of orientifold string theories in \cite{Hori,DM-DR} and in the context of topological insulators in \cite{MT1}, and are related to the Baum--Connes conjecture over the reals \cite{Rosenberg2}. In the latter setting, $\hSflip$ can be thought of as a 1D Brillouin torus $\wh{\ZZ}$ with the flip involution induced by complex conjugating characters. It is possible to understand $KR$-theory as equivariantly twisted (complex) $K$-theory, provided we expand the notion of twists to ``$\phi$-twists'' \cite{FM, Gomi3}. This roughly means that $\ZZ_2$ (or more generally $G$) is allowed to act complex \emph{antilinearly} on fibres, and is motivated to a large extent by quantum physics where time-reversal is a basic example of such an antilinear symmetry operator. 

In this paper, we study T-dualities in the \emph{purely complex} (twisted) equivariant setting, i.e.\ $K_{\ZZ_2}, K_\pm, K_G$ and their twisted versions (with no further $\phi$-twisting), their relationship with crystallographic T-duality, and therefore their remarkable appearance in solid state physics.

\subsubsection{T-duality for `Real' circle bundles and $K_\pm$-theory}\label{sec:RealTduality}
{\bf Notation:} For a space $X$ with $\Z_2$-action, $H^n_{\Z_2}(X; \Z)$ denotes its Borel equivariant cohomology with integer coefficients. We will sometimes write $H^n_{\Z_2}(X)\equiv H^n_{\Z_2}(X; \Z)$ for simplicity. The variant $H^n_\pm(X)\coloneqq H^{n+1}_{\ZZ_2}(X\times\tilde{I},X\times\partial\tilde{I},\ZZ)$ is, by a Thom isomorphism \cite{Gomi1}, isomorphic to $H^n_{\Z_2}(X; \Z(1))$, which is the equivariant cohomology with coefficients in the local system $\Z(1)$ in which $\ZZ_2$ acts by $n\mapsto -n\in \ZZ$.

A `Real' circle bundle $E$ over a space $X$ with involution is defined to be a principal $S^1$ bundle $E\rightarrow X$ with an involution $\varrho$ lifting that on $X$, such that $\varrho(\xi u)=\varrho(\xi)\bar{u}$ for all $\xi\in E, u\in S^1\cong {\rm U}(1)$. A basic example is $\Sflip\rightarrow\pt$ with trivial involution on $\pt$. Such bundles are classified \cite{Kahn} by their first `Real' Chern class (Euler class) $c_1^R(E)\in H^2_\pm(X)$. An important tool for computing $H^n_{\ZZ_2}, H^n_\pm$ for `Real' circle bundles $E$ is the `Real' Gysin sequence (Corollary 2.11 of \cite{Gomi1}),
\begin{align*}
\cdots&\rightarrow H^{n-2}_{\ZZ_2}(X)\overset{c_1^R(E)\cup}{\rightarrow}H^n_\pm(X)\overset{\pi^*}{\rightarrow} H^n_\pm(E)\overset{\pi_*}{\rightarrow}H^{n-1}_{\ZZ_2}(X)\rightarrow\cdots\\
\cdots&\rightarrow H^{n-2}_\pm(X)\overset{c_1^R(E)\cup}{\rightarrow}H^n_{\ZZ_2}(X)\overset{\pi^*}{\rightarrow} H^n_{\ZZ_2}(E)\overset{\pi_*}{\rightarrow}H^{n-1}_{\pm}(X)\rightarrow\cdots
\end{align*}

A (`Real') \emph{pair} $(E,h)$ on $X$ comprises a `Real' circle bundle $E\rightarrow X$ and a class $h\in H^3_{\ZZ_2}(E,\ZZ)$.
\begin{definition}[Theorem 1.1 of \cite{Gomi1}]\label{defn:Realdualpair}
$(E,h)$ and $(\hat{E},\hat{h})$ are called (`Real') T-dual pairs if
$$
c_1^R(\hat{E})=\pi_*(h),\quad c_1^R(E)=\hat{\pi}_*(\hat{h}),\quad p^*h=\hat{p}^*\hat{h},
$$
where $p, \hat{p}, \pi, \hat{\pi}$ are the maps in the correspondence diagram Eq.\ \eqref{Tdualitydiagram} regarded in the `Real' circle bundle sense.
\end{definition}
Existence and uniqueness of T-dual pairs was established in \cite{Gomi1}.
\begin{theorem}\label{thm:RealTdual}
Let $(E,h)$ and $(\hat{E},\hat{h})$ be T-dual pairs over an involutive space $X$. Then there are $K^*_{\ZZ_2}(X)$-module isomorphisms
\begin{equation}
\TR:K^{\bullet+h}_{\ZZ_2}(E)\rightarrow K_{\pm}^{\bullet-1+\hat{h}}(\hat{E}),\;\qquad \TR:K^{\bullet+h}_\pm(E)\rightarrow K_{\ZZ_2}^{\bullet-1+\hat{h}}(\hat{E}).
\end{equation}
\end{theorem}

\subsubsection{$\ZZ_2$-equivariant T-duality}\label{sec:ztwoTduality}
A $\ZZ_2$-equivariant circle bundle $E\rightarrow X$ over an involutive space $X$ is a principal $S^1$ bundle $E\rightarrow X$ with involution $\varrho$ lifting that on $X$, such that $\varrho(\xi u)=\varrho(\xi)u$ for all $\xi\in E, u\in S^1\cong {\rm U}(1)$. A basic nontrivial example is $\Sfree\rightarrow\pt$ with trivial involution on $\pt$. Such bundles are classified by their first equivariant Chern class $c_1^{\ZZ_2}(E)\in H^2_{\ZZ_2}(X,\ZZ)$. The Gysin sequence for such an $E$ is 
\begin{equation*}
\cdots\rightarrow H^{n-2}_{\ZZ_2}(X)\overset{c_1^{\ZZ_2}(E)\cup}{\rightarrow}H^n_{\ZZ_2}(X)\overset{\pi^*}{\rightarrow} H^n_{\ZZ_2}(E)\overset{\pi_*}{\rightarrow}H^{n-1}_{\ZZ_2}(X)\rightarrow\cdots
\end{equation*}

A ($\ZZ_2$-equivariant) pair $(E,h)$ on $X$ comprises a $\ZZ_2$-equivariant circle bundle $E\rightarrow X$ and a class $h\in H^3_{\ZZ_2}(E,\ZZ)$. Generalising Definition \ref{defn:Realdualpair} and Theorem \ref{thm:RealTdual} along the lines of the arguments for `Real' T-duality in \cite{Gomi1}, we have
\begin{definition}\label{defn:Z2dualpair}
$(E,h)$ and $(\hat{E},\hat{h})$ are called ($\ZZ_2$-equivariant) T-dual pairs if
$$
c_1^{\ZZ_2}(\hat{E})=\pi_*(h),\quad c_1^{\ZZ_2}(E)=\hat{\pi}_*(\hat{h}),\quad p^*h=\hat{p}^*\hat{h},
$$
where $p, \hat{p}, \pi, \hat{\pi}$ are as in Eq.\ \eqref{Tdualitydiagram} regarded in the $\ZZ_2$-equivariant sense.
\end{definition}
\begin{theorem}\label{thm:ZtwoTdual}
Let $(E,h)$ and $(\hat{E},\hat{h})$ be ($\ZZ_2$-equivariant) T-dual pairs over an involutive space $X$. Then there are $K^*_{\ZZ_2}(X)$-module isomorphisms
\begin{equation}
\TZtwo:K^{\bullet+h}_{\ZZ_2}(E)\rightarrow K_{\ZZ_2}^{\bullet-1+\hat{h}}(\hat{E}),\;\qquad \TZtwo:K^{\bullet+h}_\pm(E)\rightarrow K_\pm^{\bullet-1+\hat{h}}(\hat{E}).
\end{equation}
\end{theorem}
Note that $T_{\ZZ_2}$ is in particular a $R(\ZZ_2)=K^0_{\ZZ_2}(\pt)$-module map.

\section{`Real' and $\ZZ_2$-equivariant T-dualities over involutive circle base}\label{sec:circleTdualcalculations}
In this section, we will provide all the `Real' or $\ZZ_2$-equivariant circle bundle T-dualities with base space one of the three involutive circles $\Striv, \Sflip, \Sfree$. The total space of the circle bundle is necessarily a 2-torus, which admits six inequivalent $\Z_2$-actions \cite{S}, five of which are products of circles with $\Z_2$-action,
\begin{align*}
&S^1_{\mathrm{triv}} \times S^1_{\mathrm{triv}}, &
&S^1_{\mathrm{triv}} \times S^1_{\mathrm{flip}}, &
&S^1_{\mathrm{flip}} \times S^1_{\mathrm{flip}}, &
&S^1_{\mathrm{triv}} \times S^1_{\mathrm{free}}, &
&S^1_{\mathrm{flip}} \times S^1_{\mathrm{free}}.
\end{align*}
The following crystallographic interpretations are available for three of these:
\begin{equation*}
T^2_{\mathsf{pm}}\cong S^1_{\mathrm{triv}} \times S^1_{\mathrm{flip}},\quad
T^2_{\mathsf{p2}}\cong S^1_{\mathrm{flip}} \times S^1_{\mathrm{flip}}
,\quad
T^2_{\mathsf{pg}}\cong S^1_{\mathrm{flip}} \times S^1_{\mathrm{free}}.\label{pmp2pgtorus}
\end{equation*}
$S^1_{\mathrm{triv}} \times S^1_{\mathrm{triv}}, S^1_{\mathrm{triv}} \times S^1_{\mathrm{free}}$ do not arise directly from 2D wallpaper groups, but from a slight generalisation called \emph{layer groups} \cite{KL} associated to 2D ``layers'' in 3D\footnote{Explicitly, $S^1_{\mathrm{triv}} \times S^1_{\mathrm{triv}}$ would correspond to $\mathsf{p11m}$ (reflection plane symmetries) while $S^1_{\mathrm{triv}} \times S^1_{\mathrm{free}}$ corresponds to $\mathsf{p11a}$ (glide plane symmetries), although the point group needs to be regarded as a graded group, as in the frieze group $\sf{p11g}$.}. The sixth and final $\Z_2$ action is induced from the wallpaper group \textsf{cm},
\begin{equation*}
T^2_{\mathsf{cm}} = S^1 \times S^1, \qquad (u_1, u_2) \mapsto (u_2, u_1).\label{cmtorus}
\end{equation*}

The fibreings of $T^2_{\sf{p2}}, T^2_{\sf{pm}}, T^2_{\sf{pg}}, T^2_{\sf{cm}}$ as `Real' and/or $\ZZ_2$-equivariant circle bundles are illustrated in Fig.\ \ref{fig:fiberings}, and explained in further detail in the subsequent Subsections.

\begin{figure}
\begin{tikzpicture}[scale=1.7,every node/.style={scale=0.8}]

\draw[dotted] (0,0) -- (1,0) -- (1,1) -- (0,1) -- (0,0);
\draw[thick] (0,0) -- (1,0);
\draw[thick] (0,0.5) -- (1,0.5);
\draw[thick] (0,1) -- (1,1);
\node[below] at (0.5,0) {$\Striv$};
\node[left] at (0,0.5) {\begin{sideways}$\Sflip$\end{sideways}};
\node[above] at (0.5,1.1) {$\sf{pm}$};
\draw[<->] (0.4,0.4) -- (0.4,0.6);
\draw[<->] (0.6,0.2) -- (0.6,0.8);

\draw[dotted] (2,0) -- (3,0) -- (3,1) -- (2,1) -- (2,0);
\draw[thick,dashed] (2,0) -- (3,0);
\draw[thick,dashed] (2,0.5) -- (3,0.5);
\draw[thick,dashed] (2,1) -- (3,1);
\node[below] at (2.5,0) {$\Sfree$};
\node[left] at (2,0.5) {\begin{sideways}$\Sflip$\end{sideways}};
\node[above] at (2.5,1.1) {$\sf{pg}$};
\draw[->] (2.2,0.8) -- (2.2,0.2);
\draw[->] (2.2,0.8) -- (2.7,0.8);
\draw[->] (2.4,0.6) -- (2.4,0.4);
\draw[->] (2.4,0.6) -- (2.9,0.6);

\draw[dotted] (1,1.5) -- (2,1.5) -- (2,2.5) -- (1,2.5) -- (1,1.5);
\node at (1,1.5) {$\circ$};
\node at (1.5,1.5) {$\circ$};
\node at (2,1.5) {$\circ$};
\node at (1,2) {$\circ$};
\node at (1.5,2) {$\circ$};
\node at (2,2) {$\circ$};
\node at (1,2.5) {$\circ$};
\node at (1.5,2.5) {$\circ$};
\node at (2,2.5) {$\circ$};
\node[below] at (1.5,1.5) {$\Sflip$};
\node[left] at (1,2) {\begin{sideways}$\Sflip$\end{sideways}};
\node[above] at (1.5,2.6) {$\sf{p2}$};
\draw[<->] (1.75,2.25) arc (45:225:0.354);
\draw[<->] (1.5,1.8) arc (-90:90:0.2);

\draw[dotted] (4,0) -- (5,1) -- (5,2) -- (4,1) -- (4,0);
\draw[thick] (4,0.5) -- (5,1.5);
\draw[ultra thin] (4,0.67) -- (5,1.67);
\draw[ultra thin] (4,0.83) -- (5,1.83);
\draw[ultra thin] (4,0.33) -- (5,1.33);
\draw[ultra thin] (4,0.17) -- (5,1.17);
\draw[thick,dashed] (4,0) -- (5,1);
\draw[thick,dashed] (4,1) -- (5,2);
\node[right] at (5,0.95) {\begin{rotate}{45}$\Sfree$\end{rotate}};
\node[right] at (5,1.95) {\begin{rotate}{45}$\Sfree$\end{rotate}};
\node[right] at (5,1.5) {\begin{rotate}{45}$\Striv$\end{rotate}};
\node[left] at (4,0.5) {\begin{sideways}$\Sflip$\end{sideways}};
\node[above] at (4.2,1.5) {$\sf{cm}$};
\draw[<->] (4.2,0.8) -- (4.3,0.7);
\draw[<->] (4.4,1.4) -- (4.9,0.9);

\draw[dotted] (6,2) -- (7,1) -- (7,0) -- (6,1) -- (6,2);
\draw[thick] (6,1) -- (6.5,1.5);
\draw[thick] (6.5,0.5) -- (7,1);
\draw[ultra thin] (6,1.83) -- (7,0.83);
\draw[ultra thin] (6,1.67) -- (7,0.67);
\draw[ultra thin] (6,1.5) -- (7,0.5);
\draw[ultra thin] (6,1.33) -- (7,0.33);
\draw[ultra thin] (6,1.17) -- (7,0.17);
\draw[thick,dashed] (4,0) -- (5,1);
\draw[thick,dashed] (4,1) -- (5,2);
\node[below] at (6.3,0.5) {\begin{rotate}{315}$\Sflip$\end{rotate}};
\node[left] at (6,1.5) {\begin{sideways}$\Striv$\end{sideways}};
\node[above] at (6.7,1.8) {$\sf{cm}$};
\draw[<->] (6.7,0.5) -- (6.5,0.7);
\draw[<->] (6.2,1.7) -- (6.7,1.2);

\end{tikzpicture}
\caption{For $\mathscr{G}=\sf{p2},\sf{pm},\sf{pg},\sf{cm}$, the lattice $\Pi$ is taken to be square, generated by a horizontal and a vertical translation. Their unit cells $T^2_{\mathscr{G}}=R^2/\Pi$ are drawn, with arrows indicating how pairs of points are exchanged under the induced involution $\bm{\alpha}$ on $T^2_{\mathscr{G}}$. There is a product fibring of $T^2_{\sf{p2}}, T^2_{\sf{pm}}, T^2_{\sf{pg}}$ as $\ZZ_2$-equivariant circle bundles (horizontal fibres) or as `Real' circle bundles (vertical fibres). For $\sf{cm}$ we redraw $T^2_{\sf{cm}}$ as a rhombus in two different ways to illustrate its non-product (diagonal) fibring as a $\ZZ_2$-equivariant circle bundle and as a `Real' circle bundle. Thick lines indicate reflection axes, dashed lines indicate glide axes, and circles indicate $\pi$-rotation centres. }\label{fig:fiberings}
\end{figure}
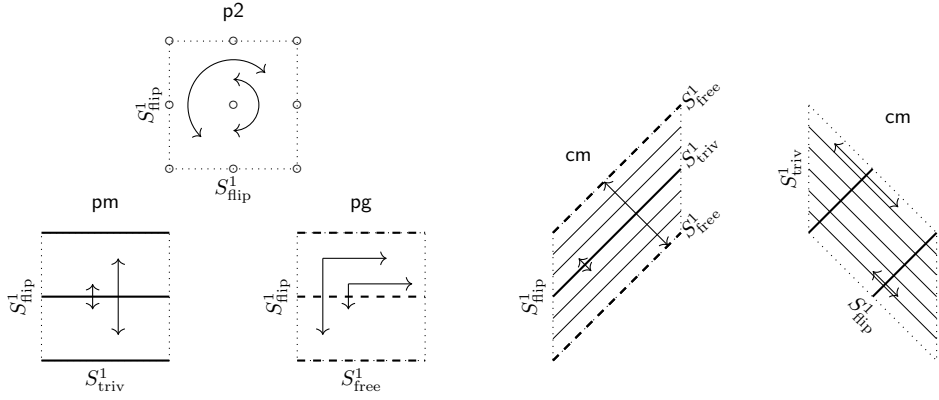

\subsection{Equivariant cohomology of $\pt$, $S^1_{\mathrm{triv}}$ and $S^1_{\mathrm{flip}}$}\label{sec:cohomologycalculations}

Let $\pt$ be the point with the trivial $\Z_2$-action. The equivariant cohomology rings of $\pt$, $S^1_{\mathrm{triv}}$ and $S^1_{\mathrm{flip}}$ were computed as (\cite{Gomi1}, Proposition 2.4, Lemma 2.15, Lemma 2.12 respectively)
\begin{eqnarray}
H^*_{\Z_2}(\pt) \oplus H^*_\pm(\pt)
&\cong& \Z[t^{1/2}]/(2t^{1/2}),\nonumber\\
H^*_{\Z_2}(S^1_{\mathrm{triv}}) \oplus H^*_{\pm}(S^1_{\mathrm{triv}})
&\cong& \Z[t^{1/2}, e]/(2t^{1/2}, e^2),\nonumber\\
H^*_{\Z_2}(S^1_{\mathrm{flip}}) \oplus H^*_{\pm}(S^1_{\mathrm{flip}})
&\cong& \Z[t^{1/2}, \chi]/(2t^{1/2}, \chi^2 - t^{1/2}\chi),\label{Sflipring}
\end{eqnarray}
where
\begin{equation*}
t^{1/2} \in H^1_\pm(\pt) \cong \Z_2,\quad e \in H^1_{\Z_2}(S^1_{\mathrm{triv}}) \cong \Z,\quad \chi \in \tilde{H}^1_\pm(S^1_{\mathrm{flip}}) \cong \Z\label{basiccohomologygenerators}
\end{equation*}
are generators such that the equivariant push-forward along $\tilde{\pi} : S^1_{\mathrm{triv}} \to \pt$ provides $\tilde{\pi}_*e = 1$ and the `Real' push-forward along $\pi: S^1_{\mathrm{flip}} \to \pt$ provides $\pi_*\chi = 1$. In low degrees, these groups are:
$$
\begin{array}{|c|c|c|c|c|c|}
\hline
& n = 0 & n = 1 & n = 2 & n = 3 & n = 4 \\
\hline
\hline
H^n_{\Z_2}(\pt) & \Z & 0 & \Z_2 & 0 & \Z_2 \\
\hline
\mbox{basis} & 1 & & t & & t^2 \\
\hline
\hline
H^n_{\pm}(\pt) & 0 & \Z_2 & 0 & \Z_2 & 0 \\
\hline
\mbox{basis} & & t^{1/2} & & t^{3/2} & \\
\hline
\hline
H^n_{\Z_2}(S^1_{\mathrm{triv}}) & \Z & \Z & \Z_2 & \Z_2 & \Z_2 \\
\hline
\mbox{basis} & 1 & e & t & te & t^2 \\
\hline
\hline
H^n_{\pm}(S^1_{\mathrm{triv}}) & 0 & \Z_2 & \Z_2 & \Z_2 & \Z_2 \\
\hline
\mbox{basis} & & t^{1/2} & t^{1/2}e & t^{3/2} & t^{3/2}e \\
\hline
\hline
H^n_{\Z_2}(S^1_{\mathrm{flip}}) & \Z & 0 & \Z_2^2 & 0 & \Z_2^2 \\
\hline
\mbox{basis} & 1 &  & t, t^{1/2}\chi & 0 & t^2, t^{3/2}\chi \\
\hline
\hline
H^n_{\pm}(S^1_{\mathrm{flip}}) & 0 & \Z_2^2 & 0 & \Z_2^2 & 0 \\
\hline
\mbox{basis} & & t^{1/2}, \chi &  & t^{3/2}, t^{1/2}\chi &  \\
\hline

\end{array}
$$
Note that $te \in H^3_{\Z_2}(S^1_{\mathrm{triv}}) \cong \Z_2$ can be represented by the group $2$-cocycle $\tau_{S^1}$ of Eq.\ \eqref{2cocyclecircle}. 
Also, the generator $\chi$ depends on the choice of the base points on $S^1_{\mathrm{flip}}$. The multiplication by $-1$ is an automorphism of $S^1_{\mathrm{flip}}$ which exchanges the two fixed points. This automorphism acts on the basis by $(-1)^*\chi = \chi + t^{1/2}$ and $(-1)^*t^{1/2} = t^{1/2}$. 

\subsubsection{$\ZZ_2$-equivariant and Real T-duality over a point}\label{sec:Duality_over_pt}
Note also that the generator $t\in H^2_{\ZZ_2}(\pt)$ corresponds to the sign representation of $\ZZ_2$, and the representation ring of $\ZZ_2$ is
$$ R(\ZZ_2)=\ZZ[t]/(t^2-1).$$
The unit circle bundle for this sign representation is just $\Sfree$, so that $c_1^{\ZZ_2}(\Sfree)=t$, and there are no twists on $\Sfree$ since $H^3_{\ZZ_2}(\Sfree)\cong H^3(S^1)=0$. Also, under $\tilde{\pi}:\Striv\rightarrow\pt$, we have $\tilde{\pi}_*(\tau_{S^1})\equiv\tilde{\pi}_*(te)=t$. Obviously $c_1^{\ZZ_2}(\Striv)=0$, so by Definition \ref{defn:Z2dualpair}, we have
\begin{proposition}\label{prop:Z2_dual_over_pt}
The $\ZZ_2$-equivariant T-dual pairs over the point are
$$(\Sfree,0)\leftrightarrow(\Striv,\tau_{S^1})\quad {\rm and}\quad (\Striv,0)\leftrightarrow(\Striv,0) .$$
\end{proposition}
Since $H^2_\pm(\pt)=0$ the only possibility for `Real' T-duality over a point is between $\Sflip$ and itself, and there are no twists since $H^3_{\ZZ_2}(\Sflip)=0$.
\begin{proposition}\label{prop:Real_dual_over_pt}
$(\Sflip,0)$ is a `Real' self-dual pair over the point.
\end{proposition}


\subsection{`Real' T-dualities over circle base}
\subsubsection{`Real' T-duality over $S^1_{\mathrm{flip}}$}\label{sec:Real_duality_over_S1_flip}

By $H^2_\pm(S^1_{\mathrm{flip}}) = 0$, there is only the trivial `Real' circle bundle $S^1_{\mathrm{flip}} \times S^1_{\mathrm{flip}} \to S^1_{\mathrm{flip}}$, which can be identified with $T^2_{\sf{p2}}$. By the splitting of the Gysin sequence, we find $H^3_{\Z_2}(S^1_{\mathrm{flip}} \times S^1_{\mathrm{flip}}) = 0$.

\begin{proposition}\label{prop:Real_dual_over_S1_flip}
There is only one `Real' self-dual pair over $S^1_{\mathrm{flip}}$:
$$
(S^1_{\mathrm{flip}} \times S^1_{\mathrm{flip}}, 0)
\leftrightarrow
(S^1_{\mathrm{flip}} \times S^1_{\mathrm{flip}}, 0).
$$
\end{proposition}

The $K$-theories verifying Prop.\ \ref{prop:Real_dual_over_S1_flip} were computed in \cite{Gomi1} (Proposition 4.1):
\begin{align*}
K^0_{\Z_2}(S^1_{\mathrm{flip}} \times S^1_{\mathrm{flip}})
&\cong (R(\Z_2) \oplus (1 - t))^{\oplus 2}, &
\;
K^1_{\Z_2}(S^1_{\mathrm{flip}} \times S^1_{\mathrm{flip}})
&\cong 0, \\
K^0_{\pm}(S^1_{\mathrm{flip}} \times S^1_{\mathrm{flip}})
&\cong 0, &
\;\;
K^1_{\pm}(S^1_{\mathrm{flip}} \times S^1_{\mathrm{flip}})
&\cong (R(\Z_2) \oplus (1 - t))^{\oplus 2}.
\end{align*}


\subsubsection{`Real' T-duality over $S^1_{\mathrm{triv}}$}
\label{sec:Real_duality_over_S1_triv}
The T-dualities over $S^1_{\mathrm{triv}}$ were given in \cite{Gomi1} as an example (\S5.5), but they are somewhat intricate and we recall some of the details here.


From $H^2_{\pm}(S^1_{\mathrm{triv}}) \cong \Z_2$, there are two inequivalent `Real' line bundles on $S^1_{\mathrm{triv}}$ whose `Real' Chern classes are $0$ and $t^{1/2}e$.
\begin{itemize}
\item
($c_1^R = 0$)
The trivial `Real' line bundle with unit circle bundle $S^1_{\mathrm{triv}} \times S^1_{\mathrm{flip}}$, is identifiable with $T^2_{\textsf{pm}}$.

\item
($c_1^R = t^{1/2}e$)
The non-trivial `Real' line bundle $R\rightarrow\Sflip$ is $S^1_{\mathrm{triv}} \times \C$ with the `Real' structure $(u, z) \mapsto (u, u\bar{z})$, and `Real' Chern class $c_1^R(R) = t^{1/2}e$. We write $\pi_R : S(R) \to S^1_{\mathrm{triv}}$ for its unit circle bundle. As a $\Z_2$-space, $S(R)$ is identified with $T^2_{\textsf{cm}}$.
\end{itemize}

\noindent
{\bf(Case $c_1^R = 0$).} 
The relevant low degree cohomology groups for the (split) Gysin sequence are the middle two columns of:
$$
\begin{array}{c|c|c|c|}
n = 4 & \Z_2 t^{1/2}e & \Z_2 t^2 & \\
\hline
n = 3 & \Z_2 t^{1/2} & \Z_2 te & \shade{\Z_2 \oplus \Z_2} \\
\hline
n = 2 & 0 & \Z_2 t & \shade{\Z_2 \oplus \Z_2} \\
\hline
n = 1 & 0 & \Z e & \shade{\Z} \\
\hline
n = 0 & 0 & \Z & \shade{\Z} \\
\hline
& H^{n-2}_\pm(S^1_{\mathrm{triv}}) & H^n_{\Z_2}(S^1_{\mathrm{triv}}) &
H^n_{\Z_2}(S^1_{\mathrm{triv}} \times S^1_{\mathrm{flip}})
\end{array}
$$
With some extra computations, the generators for the desired shaded column are (suppressing the pullback notation)
$$
\begin{array}{|c|c|c|c|c|}
\hline
& n = 0 & n = 1 & n = 2 & n = 3  \\
\hline
H^n_{\Z_2}(S^1_{\mathrm{triv}} \times S^1_{\mathrm{flip}}) &
\Z & \Z & \Z_2 \oplus \Z_2 & \Z_2 \oplus \Z_2 \\
\hline
\mbox{basis} & 1 & e & t, t^{1/2}\chi & te, t^{1/2}e\chi \\
\hline
\end{array}
$$
Under the projection $\pi : S^1_{\mathrm{triv}} \times S^1_{\mathrm{flip}} \to S^1_{\mathrm{triv}}$, we have $\pi_*\chi = 1$. In $H^3_{\Z_2}(S^1_{\mathrm{triv}} \times S^1_{\mathrm{flip}})$, the basis element $te$ is represented by (the pull-back from $\Striv$ of) the group $2$-cocycle $\tau_{S^1}$, whereas $h_{\mathsf{pm}} \coloneqq t^{1/2}e\chi$ cannot be represented by any group $2$-cocycle, according to the classification of twists \cite{Gomi2}. In the Gysin sequence, 
$$
\begin{CD}
H^3_{\Z_2}(S^1_{\mathrm{triv}}) @>{\pi^*}>>
H^3_{\Z_2}(S^1_{\mathrm{triv}} \times S^1_{\mathrm{flip}}) @>{\pi_*}>>
H^2_{\pm}(S^1_{\mathrm{triv}}). \\
@| @| @| \\
\Z_2 \tau_{S^1} @.
\Z_2 \tau_{S^1} \oplus \Z_2 h_{\mathsf{pm}} @. 
\Z_2 c_1^R(R)
\end{CD}
$$
we have
\begin{equation*}
\pi_*(\tau_{S^1})\cong \pi_*(te) =0,\qquad \pi_*(h_{\mathsf{pm}}) \cong \pi_*(t^{1/2}e\chi) =  t^{1/2}e = c_1^R(R).
\end{equation*}

\noindent
{\bf (Case $c_1^R = t^{1/2}e$).}
For the Gysin sequence, we need
$$
\begin{array}{c|c|c|c|}
n = 4 & \Z_2 t^{1/2}e & \Z_2 t^2 & \\
\hline
n = 3 & \Z_2 t^{1/2} & \Z_2 te & \shade{\Z_2} \\
\hline
n = 2 & 0 & \Z_2 t & \shade{\Z_2} \\
\hline
n = 1 & 0 & \Z e & \shade{\Z} \\
\hline
n = 0 & 0 & \Z & \shade{\Z} \\
\hline
& H^{n-2}_\pm(S^1_{\mathrm{triv}}) & H^n_{\Z_2}(S^1_{\mathrm{triv}}) &
H^n_{\Z_2}(S(R))
\end{array}
$$
It turns out that $H^3_{\Z_2}(S(R))=H^3_{\ZZ_2}(T^2_{\sf{cm}}) \cong \Z_2$, and that its generator $h_{\sf{cm}}$ is not representable by a group $2$-cocycle \cite{Gomi2}. A part of the Gysin sequence is:
$$
\begin{CD}
H^3_{\Z_2}(S^1_{\mathrm{triv}}) @>{\pi_R^*}>>
H^3_{\Z_2}(S(R)) @>{(\pi_R)_*}>>
H^2_{\pm}(S^1_{\mathrm{triv}}), \\
@| @| @| \\
\Z_2 \tau_{S^1} @.
\Z_2 h_{\sf{cm}} @. 
\Z_2 c_1^R(R)
\end{CD}
$$
$$\pi_R^*\tau_{S^1} = 0,\qquad (\pi_R)_*h_{\sf{cm}} = c_1^R(R).$$

\medskip

From the computations above, we get from Definition \ref{defn:Realdualpair} that:

\begin{proposition}\label{prop:Real_dual_over_S1_triv}
The following pairs are `Real' T-dual over $S^1_{\mathrm{triv}}$:
\begin{align*}
(S^1_{\mathrm{triv}} \times S^1_{\mathrm{flip}}, 0)
&\leftrightarrow
(S^1_{\mathrm{triv}} \times S^1_{\mathrm{flip}}, 0), \\
(S^1_{\mathrm{triv}} \times S^1_{\mathrm{flip}}, 
\tau_{S^1})
&\leftrightarrow
(S^1_{\mathrm{triv}} \times S^1_{\mathrm{flip}}, 
\tau_{S^1}), \\
(S^1_{\mathrm{triv}} \times S^1_{\mathrm{flip}}, h_{\mathsf{pm}})
&\leftrightarrow
(S(R), 0), \\
(S^1_{\mathrm{triv}} \times S^1_{\mathrm{flip}}, 
h_{\mathsf{pm}} + \tau_{S^1})
&\leftrightarrow
(S(R), 0), \\
(S(R), h_{\sf{cm}}) 
&\leftrightarrow
(S(R), h_{\sf{cm}}).
\end{align*}
\end{proposition}

\begin{remark}\label{rem:automorphism}
Notice that $S^1_{\mathrm{triv}} \times S^1_{\mathrm{flip}}$ has the $\Z_2$-equivariant automorphism given by multiplying $-1$ with $S^1_{\mathrm{flip}}$. From $(-1)^*\chi = \chi + t^{1/2}$, we have
$$
(-1)^*h_{\mathsf{pm}} = (-1)^*(t^{1/2}e\chi)
= t^{1/2}e(\chi + t^{1/2})
= t^{1/2}e\chi + te
= h_{\mathsf{pm}} + \tau_{S^1}.
$$
\end{remark}



We summarise the $K$-theories of the T-dual pairs over $S^1_{\mathrm{triv}}$, which verify Proposition \ref{prop:Real_dual_over_S1_triv}. (In the following table, two errors/typo in \cite{Gomi1} are replaced by the correct results, which are shaded.)
$$
\begin{array}{|c|c|c|c|}
\hline
h & w & 
K^{(h,w) + 0}_{\Z_2}(S^1_{\mathrm{triv}} \times S^1_{\mathrm{flip}}) & 
K^{(h,w) + 1}_{\Z_2}(S^1_{\mathrm{triv}} \times S^1_{\mathrm{flip}}) \\
\hline
0 & 0 &
R(\Z_2) \oplus (1 - t) & R(\Z_2) \oplus (1 - t) \\
\hline
0 & \mathrm{id} &
R(\Z_2) \oplus (1 - t) & R(\Z_2) \oplus (1 - t) \\
\hline
\tau_{S^1} & 0 &
(1 + t) & \shade{(1 + t)} \oplus \Z/2 \\
\hline
\tau_{S^1} & \mathrm{id} &
\shade{(1 + t)} \oplus \Z/2 & (1 + t) \\
\hline
h_{\mathsf{pm}}, h_{\mathsf{pm}} + \tau_{S^1} & 0 &
(1 + t) \oplus (1 - t) & R(\Z_2) \\
\hline
h_{\mathsf{pm}}, h_{\mathsf{pm}} + \tau_{S^1} & 
\mathrm{id} &
(1 + t) \oplus (1 - t) & R(\Z_2) \\
\hline
\end{array}
$$
$$
\begin{array}{|c|c|c|c|}
\hline
h & w & 
K^{(h,w) + 0}_{\Z_2}(S(R)) & 
K^{(h,w) + 1}_{\Z_2}(S(R)) \\
\hline
0 & 0 & R(\Z_2) & (1 + t) \oplus (1 - t) \\
\hline
0 & \mathrm{id} & R(\Z_2) & (1 + t) \oplus (1 - t) \\
\hline
h_{\sf{cm}} & 0 & (1 + t) & (1 + t) \\
\hline
h_{\sf{cm}} & \mathrm{id} & (1 + t) & (1 + t) \\
\hline
\end{array}
$$


\subsection{$\Z_2$-equivariant T-duality over $S^1_{\mathrm{flip}}$}
\label{sec:Z2_equivariant_duality_over_S1_flip}


By $H^2_{\Z_2}(S^1_{\mathrm{flip}}) \cong \Z_2 \oplus \Z_2$, there are four inequivalent $\Z_2$-equivariant line bundles on $S^1_{\mathrm{flip}}$ whose $\Z_2$-equivariant Chern classes (Euler classes) $c_1^{\ZZ_2}$ are the four elements $0$, $t$, $t^{1/2}\chi$ and $t + t^{1/2}\chi$. (Here the basis $\chi \in \tilde{H}^2_\pm(\tilde{S}^1)$ is chosen with respect to the base point $1 \in S^1_{\mathrm{flip}} \subset \C$.)
\begin{itemize}
\item
($c_1^{\Z_2} = 0$)
The $\Z_2$-equivariant line bundle associated to the trivial representation. Its unit circle bundle is $S^1_{\mathrm{triv}} \times S^1_{\mathrm{flip}}$, which is identified with the torus $T^2_{\mathsf{pm}}$.

\item
($c_1^{\Z_2} = t$)
The $\Z_2$-equivariant line bundle associated to the sign representation. Its unit circle bundle is $S^1_{\mathrm{flip}} \times S^1_{\mathrm{free}}$, which is identified with the torus $T^2_{\mathsf{pg}}$ with free involution and quotient manifold the Klein bottle ${\bf K}$.

\item
($c_1^{\Z_2} = t^{1/2}\chi$)
The product bundle $L = S^1_{\mathrm{flip}} \times \C$ with the $\Z_2$-equivariant structure $(u, z) \mapsto (\bar{u}, uz)$. We write $\pi_L : S(L) \to S^1_{\mathrm{flip}}$ for its unit circle bundle. The total space $S(L)$ can be identified with the torus $T^2_{\mathsf{cm}}$.

\item
($c_1^{\Z_2} = t + t^{1/2}\chi$)
The product bundle $L' = S^1_{\mathrm{flip}} \times \C$ with the $\Z_2$-equivariant structure $(u, z) \mapsto (\bar{u}, -uz)$. We write $\pi_{L'} : S(L') \to S^1_{\mathrm{flip}}$ for its unit circle bundle. The total space $S(L')$ can also be identified with $T^2_{\mathsf{cm}}$.
\end{itemize}
\begin{remark}\label{rem:automorphismL}
Notice that $L$ and $L'$ are non-isomorphic equivariant line bundles on $S^1_{\mathrm{flip}}$ despite both being $T^2_{\mathsf{cm}}$ as $\ZZ_2$-spaces. Nevertheless, we can take the bundle map $L \to L'$, ($(u, z) \mapsto (-u, z)$) covering the base automorphism $-1 : S^1_{\mathrm{flip}} \to S^1_{\mathrm{flip}}$, then $(-1)^*t^{1/2}\chi = t + t^{1/2}\chi$ exchanges the equivariant Chern classes.
\end{remark}

\medskip
In order to find T-dual partners for these $\ZZ_2$-equivariant circle bundles, we need to compute the push-forward of each of their possible twists to the base $\Sflip$ under their respective bundle projection maps $\tilde{\pi}$.

\noindent
{\bf (Case $c_1^{\Z_2} = t$.)}
Since the quotient Klein bottle is two-dimensional, 
$${H^3_{\Z_2}(S^1_{\mathrm{flip}} \times S^1_{\mathrm{free}}) \cong H^3({\bf K})= 0}.$$ 

\noindent
{\bf (Case $c_1^{\Z_2} = 0$.)}
Recall from \S\ref{sec:Real_duality_over_S1_triv} the basis for $H^n_{\ZZ_2}(S^1_{\mathrm{triv}} \times S^1_{\mathrm{flip}})$ in low degrees:
$$
\begin{array}{|c|c|c|c|c|}
\hline
& n = 0 & n = 1 & n = 2 & n = 3  \\
\hline
H^n_{\Z_2}(S^1_{\mathrm{triv}} \times S^1_{\mathrm{flip}}) &
\Z & \Z & \Z_2 \oplus \Z_2 & \Z_2 \oplus \Z_2 \\
\hline
\mbox{basis} & 1 & e & t, t^{1/2}\chi & te, t^{1/2}e\chi \\
\hline
\end{array}
$$
In the Gysin sequence for $\tilde{\pi} : S^1_{\mathrm{triv}} \times S^1_{\mathrm{flip}} \to S^1_{\mathrm{flip}}$,
$$
\begin{CD}
H^3_{\Z_2}(S^1_{\mathrm{flip}}) @>\tilde{\pi}^*>>
H^3_{\Z_2}(S^1_{\mathrm{triv}} \times S^1_{\mathrm{flip}}) @>\tilde{\pi}_*>>
H^2_{\Z_2}(S^1_{\mathrm{flip}}), \\
@| @| @| \\
0 @. \Z_2 te \oplus \Z_2 t^{1/2}e\chi @. \Z_2 t \oplus \Z_2 t^{1/2}\chi
\end{CD}
$$
the push-forward $\tilde{\pi}_*$ is bijective, and we have
\begin{align*}
\tilde{\pi}_*(\tau_{S^1}) 
&\equiv \tilde{\pi}_*(te) = t 
= c_1^{\Z_2}(S^1_{\mathrm{triv}} \times S^1_{\mathrm{free}}), \\
\tilde{\pi}_*(h_{\mathsf{pm}}) 
&\equiv \tilde{\pi}_*(t^{1/2}e\chi) = t^{1/2}\chi
= c_1^{\Z_2}(L), \\
\tilde{\pi}_*(\tau_{S^1} + h_{\mathsf{pm}}) 
&\equiv \tilde{\pi}_*(te + t^{1/2}e\chi) = t + t^{1/2}\chi
= c_1^{\Z_2}(L').
\end{align*}

\noindent
{\bf (Cases $c_1^{\Z_2} = t^{1/2}\chi$ and $t + t^{1/2}\chi$.)}
We have $\Z_2$-equivariant homeomorphism $S(L) \cong S(R) \cong T^2_{\textsf{cm}}$, and $H^n_{\ZZ_2}(S(R))$ was already computed in \S\ref{sec:Real_duality_over_S1_triv}. The relevant groups in the Gysin sequence for $\tilde{\pi}_L : S(L) \to S^1_{\mathrm{flip}}$ are:
$$
\begin{array}{c|c|c|c}
\hline
n = 4 & \Z_2 t \oplus \Z_2 t^{1/2}\chi & \Z_2 t^2 \oplus \Z_2 t^{3/2}\chi & \\
\hline
n = 3 & 0 & 0 & \Z_2 \\
\hline
n = 2 & \Z & \Z_2 t \oplus \Z_2 t^{1/2}\chi & \Z_2 \\
\hline
n = 1 & 0 & 0 & \Z \\
\hline
n = 0 & 0 & \Z & \Z \\
\hline
& H^{n-2}_{\Z_2}(S^1_{\mathrm{flip}}) & H^n_{\Z_2}(S^1_{\mathrm{flip}}) &
H^n_{\Z_2}(S(L))
\end{array}
$$
Putting $c_1^{\Z_2}(L) = t^{1/2}\chi$ into the Gysin sequence and Eq.\ \eqref{Sflipring}, we have for the generator $h_{\sf{cm}}\in H^3_{\Z_2}(S(L))$,
$$
(\tilde{\pi}_L)_*(h_{\mathsf{cm}}) = t + t^{1/2}\chi = c_1^{\Z_2}(L'),
$$
The computation for $S(L') \cong S(L)$ is similar, except that in the Gysin sequence for $\tilde{\pi}_{L'} : S(L') \to S^1_{\mathrm{flip}}$, we have $c_1^{\Z_2}(L') = t + t^{1/2}\chi$. Therefore we find
$$
(\tilde{\pi}_{L'})_*(h_{\mathsf{cm}}) = t^{1/2}\chi = c_1^{\Z_2}(L).
$$

\medskip

Summarizing the calculations above, we have:

\begin{proposition} \label{prop:Z2_dual_over_S1_flip}
The following pairs are $\Z_2$-equivariant T-dual over $S^1_{\mathrm{flip}}$:
\begin{align*}
(S^1_{\mathrm{triv}} \times S^1_{\mathrm{flip}}, 0)
&\leftrightarrow
(S^1_{\mathrm{triv}} \times S^1_{\mathrm{flip}}, 0), \\
(S^1_{\mathrm{triv}} \times S^1_{\mathrm{flip}}, 
\tau_{S^1})
&\leftrightarrow
(S^1_{\mathrm{flip}} \times S^1_{\mathrm{free}}, 0), \\
(S^1_{\mathrm{triv}} \times S^1_{\mathrm{flip}}, h_{\mathsf{pm}})
&\leftrightarrow
(S(L), 0), \\
(S^1_{\mathrm{triv}} \times S^1_{\mathrm{flip}}, 
h_{\mathsf{pm}} + \tau_{S^1})
&\leftrightarrow
(S(L'), 0), \\
(S(L), h_{\sf{cm}})
&\leftrightarrow
(S(L'), h_{\sf{cm}}).
\end{align*}
\end{proposition}


In Proposition \ref{prop:Z2_dual_over_S1_flip}, all the pairs except the second one are also T-dual in the `Real' sense by Proposition \ref{prop:Real_dual_over_S1_flip}, and all their $K$-theories had been presented at the end of \S\ref{sec:Real_duality_over_S1_triv} except for $S^1_{\mathrm{flip}} \times S^1_{\mathrm{free}}$ which is given at the end of \S\ref{sec:duality_over_S1_free}. These $K$-theory computations verify the dualities in Proposition \ref{prop:Z2_dual_over_S1_flip}.

\subsection{$\Z_2$-equivariant and `Real' T-duality over $S^1_{\mathrm{free}}$}\label{sec:duality_over_S1_free}

From $H^2_{\ZZ_2}(\Sfree)\cong H^2(S^1,\ZZ)=0$, $S^1_{\mathrm{free}}$ admits only the trivial $\Z_2$-equivariant line bundle, whose unit circle bundle is $S^1_{\mathrm{triv}} \times S^1_{\mathrm{free}}$ and has no nontrivial twists.
 Also, $H^2_\pm(\Sfree)=0$ so $S^1_{\mathrm{free}}$ admits only the trivial `Real' line bundle, with unit circle bundle $S^1_{\mathrm{flip}} \times S^1_{\mathrm{free}}$ (also $T^2_{\sf{pg}}$ encountered earlier in \S\ref{sec:Z2_equivariant_duality_over_S1_flip}) having no $H^3$-twists. From these computations, we get:
\begin{proposition}\label{prop:duality_over_S1_free}
The following is the $\Z_2$-equivariant T-dual pair over $S^1_{\mathrm{free}}$:
$$
(S^1_{\mathrm{triv}} \times S^1_{\mathrm{free}}, 0) 
\leftrightarrow
(S^1_{\mathrm{triv}} \times S^1_{\mathrm{free}}, 0).
$$
The following is the `Real' T-dual pair over $S^1_{\mathrm{free}}$:
$$
(S^1_{\mathrm{flip}} \times S^1_{\mathrm{free}}, 0) 
\leftrightarrow
(S^1_{\mathrm{flip}} \times S^1_{\mathrm{free}}, 0).
$$
\end{proposition}

To verify Proposition \ref{prop:duality_over_S1_free}, we compute the relevant $K$-theories. First of all, we regard $S^1_{\mathrm{free}}$ as the unit circle bundle of the $\Z_2$-equivariant line bundle associated to the sign representation. Thus, its Euler class in $K$-theory is $1 - t \in K^0_{\Z_2}(\pt)$. Using the Gysin exact sequence in $K$-theory, we get:
\begin{align*}
K^0_{\Z_2}(S^1_{\mathrm{free}})
&\cong (1 + t), &
K^1_{\Z_2}(S^1_{\mathrm{free}})
&\cong (1 + t), \\
K^0_{\pm}(S^1_{\mathrm{free}})
&\cong 0, &
K^1_{\pm}(S^1_{\mathrm{free}})
&\cong \Z/2.
\end{align*}
Then, regarding $S^1_{\mathrm{flip}} \times S^1_{\mathrm{free}}$ as the trivial `Real' circle bundle on $S^1_{\mathrm{free}}$, we apply the splitting of the Gysin sequence to have
\begin{align*}
K^0_{\Z_2}(S^1_{\mathrm{flip}} \times S^1_{\mathrm{free}})
&\cong (1 + t) \oplus \Z/2, &
K^1_{\Z_2}(S^1_{\mathrm{flip}} \times S^1_{\mathrm{free}})
&\cong (1 + t), \\
K^0_{\pm}(S^1_{\mathrm{flip}} \times S^1_{\mathrm{free}})
&\cong (1 + t), &
K^1_{\pm}(S^1_{\mathrm{flip}} \times S^1_{\mathrm{free}})
&\cong (1 + t) \oplus \Z/2.
\end{align*}
Similarly, by the splitting of the Gysin sequence for $S^1_{\mathrm{triv}} \times S^1_{\mathrm{free}} \to S^1_{\mathrm{free}}$, we get
\begin{align*}
K^0_{\Z_2}(S^1_{\mathrm{triv}} \times S^1_{\mathrm{free}})
&\cong (1 + t) \oplus (1 + t), &
K^1_{\Z_2}(S^1_{\mathrm{triv}} \times S^1_{\mathrm{free}})
&\cong (1 + t) \oplus (1 + t), \\
K^0_{\pm}(S^1_{\mathrm{triv}} \times S^1_{\mathrm{free}})
&\cong \Z/2, &
K^1_{\pm}(S^1_{\mathrm{triv}} \times S^1_{\mathrm{free}})
&\cong \Z/2.
\end{align*}

\begin{remark}
There are also $\ZZ_2$-equivariant T-dualities for circle bundles over $\Striv$, but we omit these as we do not use them in our examples.
\end{remark}



\section{2D crystallographic T-dualities}\label{sec:2Ddualities}
In this section, we apply crystallographic T-duality, ${\rm T}_{\mathscr{G}}$ of Eq.\ \eqref{crystalTdual}, to the 17 wallpaper groups. When the point group $G$ is $\ZZ_2$, $T^2_{\mathscr{G}}$ fibres as a `Real' or $\ZZ_2$-equivariant circle bundle over another circle (or even both), so that we can T-dualise only the fibre circle while keeping the base circle fixed. This ``partial Fourier transform'' produces an intermediate 2-torus which we denote by $\acute{T}^2$; if there is a second way to fibre and T-dualise, we denote the resulting space by $\grave{T}^2$. Subsequently, we may be able to re-fibre $\acute{T}^2$ (or $\grave{T}^2$) such that what was considered the base is now a fibre. A second ``partial Fourier transform'' produces $\hat{T}^2$ on the right-hand-side of Eq.\ \eqref{crystalTdual}, showing that ${\rm T}_{\mathscr{G}}$ factorises into partial T-dualities. In these cases, the factorisation is essentially a combination of the circle bundle T-dualities in Propositions \ref{prop:Real_dual_over_S1_flip}, \ref{prop:Real_dual_over_S1_triv}, \ref{prop:Z2_dual_over_S1_flip}. The $K$-theory groups appearing in 2D crystallographic T-dualities are listed in Table \ref{table:2Ddualitytable}.

\subsection{Trivial point group: \sf{p1}}
Despite ${\rm T}_{\sf{p1}}: K^\bullet(T^2_{\sf{p1}})\rightarrow K^\bullet(\hat{T}^2_{\sf{p1}})$, this is \emph{not} simply the identity map. First, the spaces $T^2_{\sf{p1}}$ and $\hat{T}^2_{\sf{p1}}$ are not naturally identified, and second, the rank and Hopf generators for their $K^0$ theory are actually exchanged under ${\rm T}_{\sf{p1}}$ \cite{Hori, MT1}.

Pick a principal fibration $T^2_{\sf{p1}}=S^1_x\times S^1_y\rightarrow S^1_y$, then the circle bundle T-dual is $\acute{T}^2_{\sf{p1}}=\hat{S}^1_x\times S^1_y\rightarrow S^1_y$; fibring over $S^1_y$ instead will yield $\grave{T}^2_{\sf{p1}}=\hat{S}^1_y\times S^1_x\rightarrow S^1_x$. Refibring $\acute{T}^2_{\sf{p1}}$ to exchange base and fibre allows a second circle bundle T-duality transformation to arrive at $\hat{T}^2_{\sf{p1}}$; similarly for $\grave{T}^2_{\sf{p1}}$. In summary, we have a web of T-duality isomorphisms
$$
\xymatrix@R.8pc@C.8pc{
& K^{\bullet-1}(\acute{T}^2_{\sf{p1}}) \ar[dr]^{{\rm T}_y}&  \\
 K^\bullet(T^2_{\sf{p1}})\ar[ur]^{{\rm T}_x}\ar[dr]_{{\rm T}_y} \ar[rr]^{{\rm T}_{\sf{p1}}}&    & K^\bullet(\hat{T}^2_{\sf{p1}}) \\
& K^{\bullet-1}(\grave{T}^2_{\sf{p1}})\ar[ur]_{{\rm T}_x} & 
}
$$

\subsection{Order-2 point group}
\subsubsection{$G=\ZZ_2$: $\sf{p2}$}
The wallpaper group $\sf{p2}$ has point group $\ZZ_2$ comprising $\pi$ rotations about an origin in Euclidean space. $T^2_{\sf{p2}}=\Sflip\times\Sflip$ is a `Real' circle bundle over $\Sflip$ (in two different ways, see Fig.\ \ref{fig:fiberings}) whose T-dual is $\acute{T}^2_{\sf{p2}}=\hSflip\times\Sflip\rightarrow \Sflip$ (or $\grave{T}^2_{\sf{p2}}=\Sflip\times\hSflip\rightarrow\Sflip$) by Proposition \ref{prop:Real_dual_over_S1_flip}. Exchanging base and fibre for $\acute{T}^2_{\sf{p2}}$ (or $\grave{T}^2_{\sf{p2}}$) and T-dualising again gives $\hat{T}^2_{\sf{p2}}=\hSflip\times\hSflip$; in summary,
$$
\xymatrix@R.8pc@C.8pc{
& K^{\bullet-1}_\pm(\acute{T}^2_{\sf{p2}}) \ar[dr]^{\TR}&  \\
 K^\bullet_{\ZZ_2}(T^2_{\sf{p2}})\ar[ur]^{\TR}\ar[dr]_{\TR} \ar[rr]^{{\rm T}_{\rm p2}}&    & K^\bullet_{\ZZ_2}(\hat{T}^2_{\sf{p2}}) \\
& K^{\bullet-1}_\pm(\grave{T}^2_{\sf{p2}})\ar[ur]_{\TR} & 
}
$$
We may also view the vertical isomorphism $K^{\bullet-1}_\pm(\acute{T}^2_{\sf{p2}})\cong K^{\bullet-1}_\pm(\grave{T}^2_{\sf{p2}})$ as crystallographic T-duality for $\sf{p2}$ in the presence of $c$-twist, cf.\ Remark \ref{rem:extratwist}.

\subsubsection{$G=D_1\cong\ZZ_2$: $\sf{pm},\sf{pg},\sf{cm}$}\label{sec:pmpgcm}
The three wallpaper groups $\sf{pm},\sf{pg},\sf{cm}$ have point group $D_1\cong\ZZ_2$ acting on $T^2$ by an orientation-reversing involution. For $\sf{pm}$ and $\sf{pg}$, the linear actions $\alpha, \hat{\alpha}$ on $T^2, \hat{T^2}$ both correspond to the ${\rm GL}(2,\ZZ)$ subgroup generated by $\begin{pmatrix} 1 & 0 \\ 0 & -1\end{pmatrix}$. Thus the Brillouin torus $\hat{T}^2$ in these two cases is identified with $\hat{T}^2_{\sf{pm}}$. The difference is that $T^2_{\sf{pg}}$ has a further translational component, because $\sf{pg}$ is nonsymmorphic. Explicitly, $\sf{pg}\cong\ZZ\rtimes\ZZ$ with the second copy of $\ZZ$ acting on the first by reflection, and it is a non-split extension
$$0\rightarrow \Pi=\ZZ\oplus\ZZ\overset{(\times 1, \times 2)}{\longrightarrow}\ZZ\rtimes\ZZ\overset{(-1)^{n_2}}{\longrightarrow}\ZZ_2=G\rightarrow 1,$$
with 2-cocycle $\nu(-1,-1)=(0,1)$. The Fourier transform of $\nu$ is the ${\rm U}(C(\hat{T}^2_{\sf{pm}}))$-valued 2-cocycle $\tau_{\sf{pg}}(-1,-1)=\{(k_1,k_2)\mapsto e^{\im k_2}\}$, which is $\tau_{S^1}$ of Eq.\ \eqref{2cocyclecircle} (pulled back under an inclusion $S^1_{\rm triv}\hookrightarrow \hat{T}^2_{\sf{pm}}$).

For $\sf{cm}$, both $\alpha, \hat{\alpha}$ correspond to the ${\rm GL}(2,\ZZ)$ subgroup generated by $\begin{pmatrix} 0 & 1 \\ 1 & 0 \end{pmatrix}$. 

In each of $\sf{pm}, \sf{pg}, \sf{cm}$, the orientability obstruction homomorphism $c_\mathscr{G}$ is the unique nontrivial homomorphism $c={\rm id}$. Thus, their crystallographic T-dualities, Eq.\ \eqref{crystalTdual}, are
\begin{align*}
K^{\bullet+c}_{\ZZ_2}(T^2_{\sf{pm}})&\cong K^{\bullet}_{\ZZ_2}(\hat{T}^2_{\sf{pm}}),\\
K^{\bullet+c}_{\ZZ_2}(T^2_{\sf{pg}})&\cong K^{\bullet+\tau_{S^1}}_{\ZZ_2}(\hat{T}^2_{\sf{pm}}),\\
K^{\bullet+c}_{\ZZ_2}(T^2_{\sf{cm}})&\cong K^{\bullet}_{\ZZ_2}(\hat{T}^2_{\sf{cm}}).
\end{align*}

\noindent
{\bf Crystallographic T-duality diagram for $\sf{pm}$. }
Using Propositions \ref{prop:Real_dual_over_S1_triv}, \ref{prop:Z2_dual_over_S1_flip}, we can apply $\TR$ to the trivial `Real' circle bundle $T^2_{\sf{pm}}=\Sflip\times\Striv\rightarrow\Striv$ to get another trivial `Real' circle bundle $\acute{T}^2_{\sf{pm}}=\hSflip\times\Striv\rightarrow\Striv$. Now regard $\acute{T}^2_{\sf{pm}}$ as a trivial $\ZZ_2$-equivariant circle bundle over $\hSflip$, and apply $\TZtwo$ to get $\hat{T}^2_{\sf{pm}}=\hStriv\times\hSflip$. Similarly, we can start from $T^2_{\sf{pm}}$ as a $\ZZ_2$-equivariant circle bundle over $\Sflip$, and apply $\TZtwo$ to obtain $\grave{T}^2_{\sf{pm}}=\hStriv\times\Sflip$. Now T-dualise $\grave{T}^2_{\sf{pm}}$ as a `Real' circle bundle over $\hStriv$ to get $\hat{T}^2_{\sf{pm}}$. To summarise, we have
\begin{equation*}
\xymatrix@R.8pc@C.8pc{
& K^{\bullet-1}_{\ZZ_2}(\acute{T}^2_{\sf{pm}}) \ar[dr]^{\TZtwo}&  \\
 K^\bullet_\pm(T^2_{\sf{pm}})\ar[ur]^{\TR}\ar[dr]_{\TZtwo} \ar[rr]^{{\rm T}_{\sf{pm}}}&    & K^{\bullet}_{\ZZ_2}(\hat{T}^2_{\sf{pm}}) \\
& K^{\bullet-1}_\pm(\grave{T}^2_{\sf{pm}})\ar[ur]_{\TR} & 
}\label{pm-diagram}
\end{equation*}

\noindent
{\bf Crystallographic T-duality diagram for $\sf{pg}$. }
Using Propositions \ref{prop:Z2_dual_over_S1_flip}, \ref{prop:duality_over_S1_free}, we can apply $\TR$ to $T^2_{\sf{pg}}=\Sflip\times\Sfree\rightarrow\Sfree$ as a trivial `Real' circle bundle to get $\acute{T}^2_{\sf{pg}}=\hSflip\times\Sfree\rightarrow\Sfree$, then regard $\acute{T}^2_{\sf{pg}}$ as a $\ZZ_2$-equivariant circle bundle over $\hSflip$ and apply $\TZtwo$ to get $\hat{T}^2_{\sf{pm}}$ twisted by $\tau_{S^1}$. As in the $\sf{pm}$ case, there is a second factorisation route, applying $\TZtwo$ then $\TR$. To summarise, 
\begin{equation}
\xymatrix@R.8pc@C.8pc{
& K^{\bullet-1}_{\ZZ_2}(\acute{T}^2_{\sf{pg}}) \ar[dr]^{\TZtwo}&  \\
 K^\bullet_\pm(T^2_{\sf{pg}})\ar[ur]^{\TR}\ar[dr]_{\TZtwo} \ar[rr]^{{\rm T}_{\sf{pg}}}&    & K^{\bullet+\tau_{S^1}}_{\ZZ_2}(\hat{T}^2_{\sf{pm}}) \\
& K^{\bullet-1+\tau_{S^1}}_\pm(\grave{T}^2_{\sf{pm}})\ar[ur]_{\TR} & 
}\label{pg-diagram}
\end{equation}

\begin{remark}
The $H^1$-twist $c$ may be identified as the orientation class of the Klein bottle ${\bf K}$ after passing to the quotient in $H^1_{\ZZ_2}(T^2_{\sf{pg}},\ZZ_2)=H^1({\bf K},\ZZ_2)$. Then the map $\TR$ on the top left of Eq.\ \eqref{pg-diagram} says that the T-dual of a Klein bottle is another Klein bottle with orientation twist. The same isomorphism was obtained in \cite{Baraglia}, \S9.1, in the context of T-dualising general (non-principal) circle bundles such as ${\bf K}$.
\end{remark}

\begin{remark}\label{rem:pgsimplification}
An easy way \cite{GT} to compute the groups in Eq.\ \eqref{pg-diagram} is to calculate the $K$-homology groups 
$$K_0({\bf K})\cong H_{\rm even}({\bf K})=\ZZ,\qquad K_1({\bf K})\cong H_{\rm odd}({\bf K})=\ZZ\oplus\ZZ/2,$$
where passage to ordinary homology is justified by low-dimensionality of ${\bf K}$. By the Baum--Connes isomorphism for $\sf{pg}$, this computes the RHS of Eq.\ \eqref{pg-diagram}, and also the LHS by Poincar\'{e} duality. The top entry of Eq.\ \eqref{pg-diagram} is just the ordinary $K$-theory of ${\bf K}$, which is
$$K^1({\bf K})\cong H^{\rm odd}({\bf K})=\ZZ,\qquad K^0({\bf K})\cong H^{\rm even}({\bf K})=\ZZ\oplus\ZZ/2,$$
thus independently verifying the `Real' T-duality in the diagram.
\end{remark}

\subsubsection{Exotic non-cocycle $H^3$-twists from partial T-duality: $\sf{cm}$}\label{sec:partialTflux}
In \S\ref{sec:Real_duality_over_S1_triv}, we saw that $T^2_{\sf{pm}}$ and $T^2_{\sf{cm}}$ admitted equivariant $H^3$-twists $h_{\sf{pm}}, h_{\sf{cm}}$ which are not representable by cocycles. These somewhat mysterious twists are \emph{not} from the family of ``special twistings'' associated with groupoid central extensions \cite{FHT1} which had appeared naturally in solid state physics in \cite{FM}. 

A natural question is whether these extra $H^3$-twists have any realisation in solid state physics. We answer this in the affirmative, namely, they appear when we do a \emph{partial} Fourier transform adapted to $\sf{cm}$; that is, an exotic twist is required in a \emph{mixed} position-momentum space description. A similar observation was made in \cite{HMT}, where partial Fourier transform for the nonabelian integer Heisenberg group (roughly: a screw-dislocated 3D lattice) was found to induce (non-equivariant) H-flux on a 3-torus (cf.\ \S\ref{sec:screw}).

\noindent
{\bf Crystallographic T-duality diagram for $\sf{cm}$. }
By Proposition \ref{prop:Real_dual_over_S1_triv}, we can take $\TR$ for the `Real' circle bundle $T^2_{\sf{cm}}$ over $\Striv$ to get $\acute{T}^2_{\sf{pm}}$ twisted either by 
$h_{\sf{pm}}$ or $h_{\sf{pm}}+\tau_{S_1}$, the two options being related by the automorphism $f$ of $\acute{T}^2_{\sf{pm}}$ which multiplies its $\Sflip$ circle by $-1$ (Remark \ref{rem:automorphism}). By Proposition \ref{prop:Z2_dual_over_S1_flip}, we can then take $\TZtwo$ to arrive at $\hat{T}^2_{\sf{cm}}$. A second (pair of) factorisations is possible by first performing $\TZtwo$ on $T^2_{\sf{cm}}$ regarded as either of the $\ZZ_2$-equivariant circle bundles $S(L)$ or $S(L')$ over $\Sflip$ (cf.\ Remark \ref{rem:automorphismL}), and then $\TR$. To summarise,
\begin{equation*}
\xymatrix@R.8pc@C.8pc{
& K^{\bullet-1+h_{\sf{pm}}+\tau_{S^1}}_{\ZZ_2}(\acute{T}^2_{\sf{pm}}) \ar[ddr]^{\TZtwo} &\\
& K^{\bullet-1+h_{\sf{pm}}}_{\ZZ_2}(\acute{T}^2_{\sf{pm}}) \ar[u]^{f^*} \ar[dr]^{\TZtwo}&  \\
 K^\bullet_\pm(T^2_{\sf{cm}})\ar[ur]^{\TR}\ar[uur]^{\TR}\ar[dr]_{\TZtwo}\ar[ddr]_{\TZtwo} \ar[rr]^{{\rm T}_{\sf{cm}}}&    & K^{\bullet}_{\ZZ_2}(\hat{T}^2_{\sf{cm}}) \\
& K^{\bullet-1+h_{\sf{pm}}}_\pm(\grave{T}^2_{\sf{pm}})\ar[d]^{f^*}\ar[ur]_{\TR} & \\
& K^{\bullet-1+h_{\sf{pm}}+\tau_{S^1}}_\pm(\grave{T}^2_{\sf{pm}})\ar[uur]_{\TR} &
}\label{cm-diagram}
\end{equation*}

The exotic twist $h_{\sf{cm}}$ appears in the last pair of Proposition \ref{prop:Z2_dual_over_S1_flip} and of Proposition \ref{prop:Real_dual_over_S1_triv}, and these lead to a modified web of dualities,
\begin{equation*}
\xymatrix@R.8pc@C.8pc{
& K^{\bullet-1+h_{\sf{cm}}}_{\ZZ_2}(\acute{T}^2_{\sf{cm}}) \ar[dr]^{\TZtwo}&  \\
 K^{\bullet+h_{\sf{cm}}}_\pm(T^2_{\sf{cm}})\ar[ur]^{\TR}\ar[dr]_{\TZtwo} \ar[rr]^{{\rm T}^{h_{\sf{cm}}}_{\sf{cm}}}&    & K^{\bullet+h_{\sf{cm}}}_{\ZZ_2}(\hat{T}^2_{\sf{cm}}) \\
& K^{\bullet-1+h_{\sf{cm}}}_\pm(\grave{T}^2_{\sf{cm}})\ar[ur]_{\TR} & 
}\label{cm-diagram2}
\end{equation*}
where the composition ${\rm T}^{h_{\sf{cm}}}_{\sf{cm}}$ of the two partial T-dualities may be interpreted as crystallographic T-duality ${\rm T}_{\sf{cm}}$ enhanced by twisting both sides with $h_{\sf{cm}}$.

\subsection{Point group $D_3$: Orbifold change under crystallographic T-duality}\label{sec:p31mp3m1}
\begin{figure}
\begin{center}
\begin{subfigure}[b]{0.4\textwidth}
\begin{tikzpicture}[scale=1.2]

\draw[dotted] (0,0) -- (2,0) -- (3,1.73) -- (1,1.73) -- (0,0);

\draw[fill] (1,1.73) circle [radius=0.05];
\draw[fill] (1,0.58) circle [radius=0.05];
\draw[fill] (2,1.15) circle [radius=0.05];
\draw[fill] (0,0) circle [radius=0.05];
\draw[fill] (2,0) circle [radius=0.05];
\draw[fill] (3,1.73) circle [radius=0.05];

\draw[thick] (1,0) to (1,1.73);
\draw[thick] (2,0) to (2,1.73);
\draw[thick] (2,0) to (0.5,0.87);
\draw[thick] (1,1.73) to (2.5,0.87);
\draw[thick] (0,0) to (3,1.73);

\end{tikzpicture}
\caption{$\sf{p3m1}$}\label{fig:p3m1}
\end{subfigure}
\qquad\qquad
\begin{subfigure}[b]{0.4\textwidth}
\begin{tikzpicture}[scale=1.2]

\draw[thick] (0,0) -- (0,2) -- (1.73,1) -- (1.73,-1) -- (0,0);

\draw[thick] (0,0) to (1.73,1);

\draw[fill] (0,0) circle [radius=0.05];
\draw[fill] (0,2) circle [radius=0.05];
\draw[fill] (1.73,1) circle [radius=0.05];
\draw[fill] (1.73,-1) circle [radius=0.05];
\node at (0.57,1) {$\circ$};
\node at (1.15,0) {$\circ$};

\end{tikzpicture}
\caption{$\sf{p31m}$}\label{fig:p31m}
\end{subfigure}
\caption{Unit cells for $\sf{p3m1}$ and $\sf{p31m}$. Solid lines indicate reflection axes, solid circles indicate points with full point group $D_3$ isotropy, hollow circles (only applicable to $\sf{p31m}$) indicate points with only $\frac{2\pi}{3}$ rotation symmetry. The lattice for $\sf{p3m1}$ has reciprocal lattice identifiable with the lattice for $\sf{p31m}$. }\label{fig:p3m1p31m}
\end{center}
\end{figure}
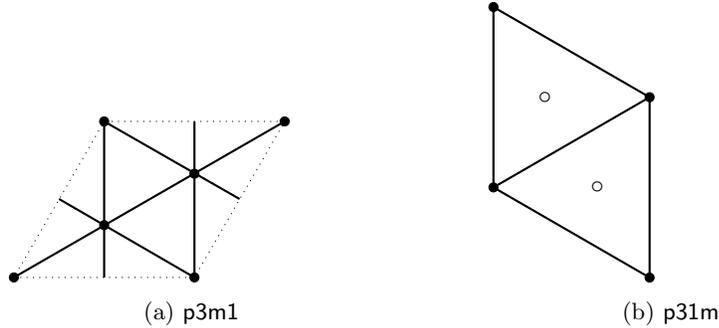

The two symmorphic wallpaper groups $\sf{p3m1}, \sf{p31m}$ have the dihedral group $D_3$ as their point group, and have the interesting feature that the $\alpha$ for $\sf{p31m}$ has dual action $\hat{\alpha}$ equivalent to the $\alpha$-action of $\sf{p3m1}$, and vice versa (see Eq.\ 27 of \cite{Michel}, cf.\ Lemma 2.4 of \cite{Gomi2}), as illustrated in Fig.\ \ref{fig:p3m1p31m}. Thus the position space orbifold $T^2_{\sf{p3m1}}$ dualises into $\hat{T}^2_{\sf{p31m}}$ while $T^2_{\sf{p31m}}$ dualises into $\hat{T}^2_{\sf{p3m1}}$.

The orientability homomorphism $c:D_2\rightarrow\ZZ_2$ is the unique surjective one, and we have the dualities
$$
K^{\bullet+c}_{D_2}(T^2_{\sf{p3m1}})\overset{{\rm T}_{\sf{p3m1}}}{\longrightarrow} K^\bullet_{D_2}(\hat{T}^2_{\sf{p31m}}),
$$
$$
K^{\bullet+c}_{D_2}(T^2_{\sf{p31m}})\overset{{\rm T}_{\sf{p31m}}}{\longrightarrow} K^\bullet_{D_2}(\hat{T}^2_{\sf{p3m1}}).
$$
Without the $c$-twist, the equivariant $K$-theories of both $T^2_{\sf{p3m1}}$ and $T^2_{\sf{p31m}}$ were computed in \cite{SSG2}, and also via the $K$-groups of the group $C^*$-algebra in \cite{Yang,LS}, to be $\ZZ^5$ in degree 0 and $\ZZ$ in degree 1. It is also possible, although not detailed in this paper, to compute directly that the $c$-twisted equivariant $K$-theories are also $\ZZ^5$ or $\ZZ$, verifying the crystallographic T-dualities for these two wallpaper groups.

\subsection{Remaining cases}
In each of the remaining cases, $\hat{\alpha}$ on $\hat{T}^2$ is conjugate to $\alpha$ on $T^2$, and there are no equivariant fibrings as circle bundles for factorisation of ${\rm T}_{\mathscr{G}}$.

\subsubsection{Cyclic point groups $\sf{p3}$, $\sf{p4}$, $\sf{p6}$}
For point groups $\ZZ_3,\ZZ_4,\ZZ_6$ comprising order-3,4 or 6 rotations with respective wallpaper groups $\sf{p3}$, $\sf{p4}$, $\sf{p6}$, there are only the crystallographic T-dualities:
$$
K^\bullet_{\ZZ_{3/4/6}}(T^2_{\sf{p3/4/6}})\overset{{\rm T}_{\sf{p3/4/6}}}{\longrightarrow} K^\bullet_{\ZZ_{3/4/6}}(\hat{T}^2_{\sf{3/4/6}}).
$$

\subsubsection{Point group $D_2, D_4, D_6$}
The orientability homomorphism for the remaining seven wallpaper groups (with point group $D_2, D_4$ or $D_6$) is $c_\mathscr{G}:D_n=\ZZ_n\rtimes\ZZ_2\rightarrow\ZZ_2$. Of these, the symmorphic ones $\sf{pmm}, \sf{cmm}, \sf{p4m}, \sf{p6m}$ and also $\sf{p4g}$ have $H^3$ obstruction $\omega$ (pulled back from $\pt$). As for $\sf{pmg}, \sf{pgg}$, nonsymmorphic versions of $\sf{pmm}$ in which one/both of the generating reflections in $D_2$ is/are replaced by glide reflections, the pullback of $\omega$ from $\pt$ to $T^2_{\mathscr{G}}$ trivialises, because there are maps of groupoids $T^2_{\sf{pmg}/\sf{pgg}}/\!\!/ D_2\rightarrow S^1_{{\rm flip}\times{\rm free}}/\!\!/D_2\rightarrow\pt/\!\!/D_2$, so after quotienting, $H^3_{D_2}(S^1_{{\rm flip} \times {\rm free}},\ZZ)\cong H^3_{\ZZ_2}(\Sflip,\ZZ)=0$ (\S\ref{sec:cohomologycalculations}). The nonsymmorphic $\sf{pmg}, \sf{pgg}, \sf{p4g}$ give the cocycle $H^3$-twists $\tau_{\sf{pmg}},\tau_{\sf{pgg}},\tau_{\sf{p4g}}$ on the Brillouin torus $\hat{T}^2$ as usual.

\noindent
{\bf Crystallographic T-dualities for $\mathscr{G}$ with $G=D_2,D_4,D_6$. }
\begin{equation}
K^{\bullet+(\omega,c_{\mathscr{G}})}_G(T^2_{\mathscr{G}})\overset{{\rm T}_{\mathscr{G}}}{\longrightarrow} K^\bullet_G(\hat{T}^2_{\mathscr{G}}),\qquad \mathscr{G}=\sf{pmm}, \sf{cmm}, \sf{p4m}, \sf{p6m},\label{wallpaperD2dualities}
\end{equation}
$$
K^{\bullet+c_{\sf{pmg}/\sf{pgg}}}_{D_2}(T^2_{\sf{pmg}/\sf{pgg}})\overset{{\rm T}_{\sf{pmg}/\sf{pgg}}}{\longrightarrow} K^{\bullet+\tau_{\sf{pmg}/\sf{pgg}}}_{D_2}(\hat{T}^2_{\sf{pmm}}),
$$
$$
K^{\bullet+(\omega,c_{\sf{p4g}})}_{D_4}(T^2_{\sf{p4g}})\overset{{\rm T}_{\sf{p4g}}}{\longrightarrow} K^{\bullet+\tau_{\sf{p4g}}}_{D_4}(\hat{T}^2_{\sf{p4m}}).
$$

\subsubsection{Combination of cocycle and Spin${}^c$ obstruction twists}\label{rem:trivialpullback}
In \S\ref{sec:topologychange}, we recalled that in (nonequivariant) T-duality of circle bundles over some base space, the Chern class $c_1$ and H-flux $h$ invariants get mixed around under T-duality according to the relations Eq.\ \eqref{chern-flux-exchange}, and \emph{both} can be nontrivial on each side of a duality. The mixing up of topological invariants is an interesting and intricate feature of T-duality. 

In the basic crystallographic T-duality Eq.\ \eqref{crystalTdual}, equivariant twists serve as topological invariants, with $\sigma_\sG$ the equivariant Spin${}^c$ obstruction for $T^d$ on one side, and $\tau_\sG$ the 2-cocycle twist for the dual torus $\hat{T}^d$ on the other side. Also, the affine point group action on $T^d$ is related to the 2-cocycle twist $\tau_\sG$ on $\hat{T}^d$ through the crystallographic group. More generally, it is possible for more complicated combinations of 2-cocycle twists and Spin${}^c$ obstruction twists to occur on the \emph{same} side of a T-duality, and we give an example here. 

In the previous subsection, we obtained the crystallographic T-duality 
\begin{equation}
{\rm T}_{\sf{pmg}}:K^{\bullet+c_{\sf{pmg}}}_{D_2}(T^2_{\sf{pmg}})\longrightarrow K^{\bullet+\tau_{\sf{pmg}}}_{D_2}(\hat{T}^2_{\sf{pmm}}),\label{pmgTduality}
\end{equation}
where on the left side $\sigma_{\sf{pmg}}$ reduced to the orientability homomorphism $H^1$ twist $c_{\sf{pmg}}$ because the $H^3$ part (the pullback of $\omega$ from a point to $T^2_{\sf{pmg}}$) necessarily vanished. The pullback of $\omega$ to the Brillouin torus $\hat{T}^2_{\sf{pmm}}$ is, however, nonzero. Furthermore, it was observed in \cite{SSG2} that there is a $D_2$-equivariant automorphism of $\hat{T}^2_{\sf{pmm}}$ which takes $\tau_{\sf{pmg}}$ to $\tau_{\sf{pmg}}+\omega$ (see the Appendix for a concrete calculation). This means that the right side of Eq.\ \eqref{pmgTduality} is isomorphic to $K^{\bullet+\tau_{\sf{pmg}}+\omega}_{D_2}(\hat{T}^2_{\sf{pmm}})$, and there is another T-duality
\begin{equation}
K^{\bullet+c_{\sf{pmg}}}_{D_2}(T^2_{\sf{pmg}})\overset{\rm T}{\longleftrightarrow} K^{\bullet+\tau_{\sf{pmg}}+\omega}_{D_2}(\hat{T}^2_{\sf{pmm}})\label{pmgTdualitywithomega}
\end{equation}
in which a cocycle twist $\tau_{\sf{pmg}}$ and a (S)pin$^{c}$ obstruction twist $\omega$ both appear on the right side. On the left side, the $c_{\sf{pmg}}$ part of the  Spin$^{c}$ obstruction twist remains, and a 2-cocycle is implicit in the affine part of the $\sf{pmg}$ action (rather than appearing as a $H^3$-twist). A similar result holds for $\sf{pgg}$ in place of $\sf{pmg}$.


\begin{table}
\begin{center}
 \begin{tabular}{|| c | c | c | c | c | c | c | c ||} 
 \hline
 \multicolumn{3}{|| c |}{Primal $\bullet=0$} & \multirow{2}{*}{$K^\bullet$} & \multirow{2}{*}{$K^{\bullet-1}$} & \multicolumn{3}{| c ||}{$\bullet=1$ Circle bundle dual}\\
 \cline{1-3}\cline{6-8}
Space & $H^3$-twist & $H^1$-twist &  &  & $H^1$-twist & $H^3$-twist & Space \\ [0.5ex] 
 \hline
 $T^2_{\sf{p1}}$ & N/A & 0 & $\ZZ^2$ & $\ZZ^2$ & \multicolumn{2}{| c |}{Self T-dual} & $T^2_{\sf{p1}}$\\
 \hline\hline
  \multicolumn{3}{|| c |}{Primal $\bullet=0$} &$K_{\ZZ_2}^\bullet$ &$K_{\ZZ_2}^{\bullet-1}$ & \multicolumn{3}{| c ||}{$\bullet=1$ `Real' circle bundle dual}\\
\hline
 $T^2_{\sf{p2}}$ & N/A & 0 & $\ZZ^6$ & 0 & \shade$c$ & \shade N/A & $T^2_{\sf{p2}}$\\
 \hline
  $T^2_{\sf{pg}}$ & \shade N/A & \shade $c$ & \multirow{2}{*}{$\ZZ$} & \multirow{2}{*}{$\ZZ\oplus\ZZ/2$} & 0 & N/A & $T^2_{\sf{pg}}$\\
 \cline{1-3}\cline{6-8}
\multirow{2}{*}{$T^2_{\sf{pm}}$} & $\tau_{S^1}$ & 0 &  &   & \shade $c$ & \shade $\tau_{S^1}$ & \multirow{4}{*}{$T^2_{\sf{pm}}$} \\ 
 \cline{2-7}
 & 0 & 0 & $\ZZ^3$ & $\ZZ^3$  & \shade $c$ & \shade 0 & \\
 \cline{1-7}
 \multirow{3}{*}{$T^2_{\sf{cm}}$} & 0 & 0 & \multirow{2}{*}{$\ZZ^2$} &  \multirow{2}{*}{$\ZZ^2$} & \shade $c$ & \shade $h_{\sf{pm}}$\,\,or\,\,$h_{\sf{pm}}+\tau_{S^1}$ &\\
  \cline{2-3}\cline{6-7}
    & \shade  0  & \shade $c$ &  &  & \shade $0$ & \shade $h_{\sf{pm}}$\,\,or\,\,$h_{\sf{pm}}+\tau_{S^1}$ &\\
  \cline{2-8}
& \shade  $h_{\sf{cm}}$ & \shade  0 & $\ZZ$ & $\ZZ$  & \shade $c$ & \shade $h_{\sf{cm}}$ & $T^2_{\sf{cm}}$ \\ 
 \hline
 \hline
  \multicolumn{3}{|| c |}{Primal $\bullet=0$} &$K_G^\bullet$ &$K_G^{\bullet-1}$ & \multicolumn{3}{| c ||}{$\bullet=0$ Crystallographic dual}\\
\hline
\multirow{2}{*}{$T^2_{\sf{p31m}}$} & 0 & 0 & $\ZZ^5$ & $\ZZ$  & \shade $c$ & \shade 0 & \multirow{2}{*}{$T^2_{\sf{p3m1}}$} \\ 
 \cline{2-7}
  & \shade $0$ & \shade $c$ & $\ZZ^5$ & $\ZZ$  & 0 & 0  & \\
\hline
 $T^2_{\sf p3}$ & N/A & 0 & $\ZZ^8$ & 0 & \multicolumn{2}{| c |}{Self crystal T-dual}  &  $T^2_{\sf p3}$\\
\hline
 $T^2_{\sf p4}$ & N/A & 0 & $\ZZ^9$ & 0 & \multicolumn{2}{| c |}{Self crystal T-dual}  &  $T^2_{\sf p4}$\\
\hline
 $T^2_{\sf p6}$ & N/A & 0 & $\ZZ^{10}$ & 0 & \multicolumn{2}{| c |}{Self crystal T-dual}  &  $T^2_{\sf p6}$\\
\hline
 $T^2_{\sf pmm}$ & \shade $\omega$ & \shade $c_{\sf{pmm}}$ & $\ZZ^9$ & 0  & 0 & 0  &   \multirow{3}{*}{$T^2_{\sf pmm}$}\\
\cline{1-7}
 $T^2_{\sf pmg}$ & \shade 0 &\shade  $c_{\sf{pmg}}$ & $\ZZ^4$ & $\ZZ$  & 0 & $\tau_{\sf{pmg}}$  &  \\
\cline{1-7}
 $T^2_{\sf pgg}$ & \shade 0 & \shade $c_{\sf{pgg}}$ & $\ZZ^3$ & $\ZZ/2$  & 0 & $\tau_{\sf{pgg}}$  & \\
\hline
 $T^2_{\sf cmm}$ & \shade $\omega$  & \shade $c_{\sf{cmm}}$ & $\ZZ^6$ & 0  & 0 & 0  &  $T^2_{\sf cmm}$\\
\hline
 $T^2_{\sf p4m}$ & \shade  $\omega$ & \shade $c_{\sf{p4m}}$ & $\ZZ^9$ & 0  & 0 & 0  &   \multirow{2}{*}{$T^2_{\sf p4m}$}\\
\cline{1-7}
 $T^2_{\sf p4g}$ & \shade $\omega$ & \shade $c_{\sf{p4g}}$ & $\ZZ^6$ & 0  & 0 & $\tau_{\sf{p4g}}$  &  \\
 \hline
 $T^2_{\sf p6m}$ & \shade  $\omega$ & \shade $c_{\sf{p6m}}$ & $\ZZ^8$ & 0  & 0 & 0  &  $T^2_{\sf p6m}$\\
\hline
\end{tabular}
\caption{List of $K$-theory groups appearing in 2D crystallographic T-dualities, note the use of Eq.\ \eqref{KpmKZtworelation}. Unshaded entries were computed in \cite{Yang} through $K_\bullet(C^*_r(\mathscr{G}))$ (see also \cite{LS}), and directly as twisted $K$-theory groups in \cite{SSG2}. Shaded entries indicate $K$-theories with graded twists that are further implied by various T-dualities, which for point group $\ZZ_2$ (middle set of rows) were independently computed in \S\ref{sec:circleTdualcalculations} (and partially in \S5.5 of \cite{Gomi1}).}\label{table:2Ddualitytable}
\end{center}
\end{table}

\section{1D crystallographic T-dualities}\label{sec:1Ddualities}
\subsection{1D space groups, frieze groups, and graded twists}\label{sec:frieze}
The two 1D space groups $\ZZ, \ZZ\rtimes\ZZ_2$ are special cases of frieze groups, and we shall use the international notation $\sf{p1}, \sf{p1m1}$. Our convention is to regard the $\ZZ$ symmetry to be along the horizontal direction. A frieze group is a generalisation of a 1D space group to include an extra ``internal'' direction. Such generalisations of space groups are called \emph{subperiodic groups} \cite{KL}. Sometimes, the extra direction is taken to be a time direction which can be reversed by symmetry group operations, and 1D frieze groups are examples of \emph{magnetic space groups}. 

This internal direction is crucial in the bulk-boundary correspondence, where a 1D boundary line should be thought of as sitting in 2D, whence it has a notion of ``above'' and ``below'' the line \cite{GT}. For example, even though reflection of the vertical coordinate in 2D restricts to the trivial action on the invariant horizontal axis, the internal label ``above/below'' is changed, and this is recorded by giving the reflection the \emph{odd} grading.

The seven frieze groups (see Fig.\ \ref{fig:frieze}), with their natural gradings, are summarised in the following table. The point groups are either trivial, $\ZZ_2$, or $D_2$. In a semidirect product $\ZZ\rtimes\ZZ_2$, the point group $\ZZ_2$ acts on $\ZZ$ by reflection, while in $\ZZ\rtimes D_2$, the second factor of $D_2=\ZZ_2\times\ZZ_2$ acts on $\ZZ$ by reflection while the first factor acts trivially. The projection of $D_2$ onto the $i$-th $\ZZ_2$-factor is denoted by $p_i, i=1,2$.

\begin{center}
 \begin{tabular}{||c | c  | c ||} 
 \hline
 IUC Name & Graded point group & Abstract graded group \\ [0.5ex] 
 \hline\hline
 $\sf{p1}$ & 1 & $\ZZ$  \\ 
 \hline
 $\sf{p1m1}$ & $\ZZ_2\rightarrow 1$ &$ \ZZ\rtimes\ZZ_2$  \\
 \hline
 $\sf{p2}$ & $\ZZ_2\overset{\rm id}{\rightarrow}\ZZ_2$ & $\ZZ\rtimes\ZZ_2\overset{1,{\rm id}}{\rightarrow}\ZZ_2$  \\
 \hline
 $\sf{p11m}$ & \multirow{2}{*}{$\ZZ_2\overset{\rm id}{\rightarrow}\ZZ_2$} & $\ZZ\times\ZZ_2\overset{1,{\rm id}}{\rightarrow}\ZZ_2$ \\
 \cline{1-1}\cline{3-3}
 $\sf{p11g}$  &  & $\ZZ\overset{{(-1)^n}}{\rightarrow}\ZZ_2$  \\
 \hline
  $\sf{p2mm}$  & \multirow{2}{*}{$D_2\overset{p_1}{\rightarrow}\ZZ_2$} & $\ZZ\rtimes D_2\overset{1,p_1}{\rightarrow}\ZZ_2$  \\
 \cline{1-1}\cline{3-3}
 $\sf{p2mg}$ & & $\ZZ\rtimes\ZZ_2\overset{{(-1)^n},1}{\rightarrow}\ZZ_2$  \\ [1ex] 
 \hline
\end{tabular}\label{table:friezegroups}
\end{center}

\begin{figure}
\begin{tikzpicture}[scale=1.7,every node/.style={scale=0.8}]

\draw[dotted] (0,0.5) -- (1,0.5);
\node[above] at (0.5,0.7) {$\sf{p1}$};
\node at (0,0.5) {$\bullet$};
\node at (0.5,0.5) {$\bullet$};
\node at (1,0.5) {$\bullet$};
\node at (0.25,0.5) {b};
\node at (0.75,0.5) {b};

\draw[dotted] (1,1.5) -- (2,1.5);
\node[above] at (1.5,1.7) {$\sf{p1m1}$};
\node at (1,1.5) {$\bullet$};
\node at (1.5,1.5) {$\bullet$};
\node at (2,1.5) {$\bullet$};
\node at (1,1.5) {$|$};
\node at (1.25,1.5) {$|$};
\node at (1.5,1.5) {$|$};
\node at (1.75,1.5) {$|$};
\node at (2,1.5) {$|$};
\node at (1.25,1.5) {$\vee$};
\node at (1.75,1.5) {$\vee$};

\draw[dotted] (2,0.5) -- (3,0.5);
\node[above] at (2.5,0.7) {$\sf{p2}$};
\node at (2,0.5) {$\bullet$};
\node at (2.5,0.5) {$\bullet$};
\node at (3,0.5) {$\bullet$};
\draw (2,0.5) circle (2.2pt);
\draw (2.25,0.5) circle (2.2pt);
\draw (2.5,0.5) circle (2.2pt);
\draw (2.75,0.5) circle (2.2pt);
\draw (3,0.5) circle (2.2pt);
\node at (1,1.5) {$|$};
\node at (1.25,1.5) {$|$};
\node at (1.5,1.5) {$|$};
\node at (1.75,1.5) {$|$};
\node at (2,1.5) {$|$};
\node at (2.25,0.5) {S};
\node at (2.75,0.5) {S};

\draw[thick] (3,1.5) -- (4,1.5);
\node[above] at (3.5,1.7) {$\sf{p11m}$};
\node at (3,1.5) {$\bullet$};
\node at (3.5,1.5) {$\bullet$};
\node at (4,1.5) {$\bullet$};
\node at (3.25,1.5) {B};
\node at (3.75,1.5) {B};

\draw[thick,dashed] (4,0.5) -- (5,0.5);
\node[above] at (4.5,0.7) {$\sf{p11g}$};
\node at (4,0.5) {$\bullet$};
\node at (4.5,0.5) {$\bullet$};
\node at (5,0.5) {$\bullet$};
\node at (4.125,0.5) {p};
\node at (4.375,0.5) {b};
\node at (4.625,0.5) {p};
\node at (4.875,0.5) {b};

\draw[thick] (5,1.5) -- (6,1.5);
\node[above] at (5.5,1.7) {$\sf{p2mm}$};
\node at (5,1.5) {$\bullet$};
\node at (5.5,1.5) {$\bullet$};
\node at (5,1.5) {$\bullet$};
\node at (5,1.5) {$|$};
\node at (5.25,1.5) {$|$};
\node at (5.5,1.5) {$|$};
\node at (5.75,1.5) {$|$};
\node at (6,1.5) {$|$};
\node at (5.25,1.5) {H};
\node at (5.75,1.5) {H};

\draw[thick,dashed] (6,0.5) -- (7,0.5);
\node[above] at (6.5,0.7) {$\sf{p2mg}$};
\node at (6,0.5) {$\bullet$};
\node at (6.5,0.5) {$\bullet$};
\node at (7,0.5) {$\bullet$};
\node at (6.125,0.5) {$|$};
\node at (6.375,0.5) {$|$};
\node at (6.625,0.5) {$|$};
\node at (6.875,0.5) {$|$};
\node at (6.125,0.5) {$\vee$};
\node at (6.375,0.5) {$\wedge$};
\node at (6.625,0.5) {$\vee$};
\node at (6.875,0.5) {$\wedge$};


\end{tikzpicture}
\caption{For each frieze group, three lattice points $\bullet$ (two unit cells) are drawn. There may be additional symmetry operations which preserve a pattern of symbols indicated b, p, $\vee, \wedge$, S, B, or H. Thick horizontal lines are horizontal reflection axes, dashed lines are glide reflection axes, circles are $\pi$ rotation centres, and vertical lines $|$ are vertical reflection axes. The $\pi$-rotation symmetries for $\sf{p2mm}$ and $\sf{p2mg}$ are not independent group generators, and are omitted.}\label{fig:frieze}
\end{figure}
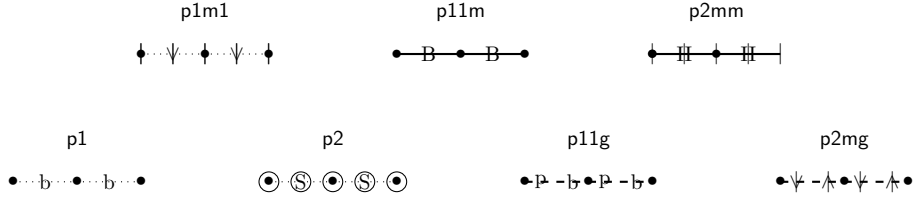

\subsubsection{Point group $1$: $\sf{p1}$} 
{\bf $\sf{p1}$ case.} $T^1_{\sf{p1}}=S^1$, $\hat{T}^1_{\sf{p1}}=\hat{S}^1$, and ${\rm T}_{\sf{p1}}$ is the basic T-duality in Eq.\ \eqref{circledual}.

\subsubsection{Point group $\ZZ_2$ acting by reflection: $\sf{p1m1},\sf{p2}$}
The frieze groups $\sf{p1m1}$ and $\sf{p2}$ are both $\ZZ\rtimes\ZZ_2$ and come with the nontrivial orientability homomorphism $c_\mathscr{G}=c:\ZZ_2\overset{{\rm id}}{\rightarrow}\ZZ_2$. 
In both cases, $T^1_{\mathscr{G}}=R^1/\ZZ=\Sflip$ has the flip involution, and the momentum space is $\hSflip$. However, the point group in $\sf{p1m1}$ implements reflection of the horizontal coordinate, whereas in $\sf{p2}$ it implements $\pi$-rotation; only the latter exchanges the internal ``above/below'' label and is non-trivially graded by $c:\ZZ_2\overset{{\rm id}}{\rightarrow}\ZZ_2$.

The crystallographic T-duality, Eq.\ \eqref{crystalTdual}, for $\sf{p1m1}$ is
$$
{\rm T}_{\sf{p1m1}}: K^\bullet_\pm(\Sflip)\cong K^{\bullet+c_{\sf{p1m1}}}_{\ZZ_2}(\Sflip)\overset{\cong}{\longrightarrow}K^{\bullet-1}_{\ZZ_2}(\hSflip),\label{p1m1duality}
$$
which is also $\TR$ for the `Real' T-dual circle bundles $\Sflip, \hSflip$ over a point (Proposition \ref{prop:Real_dual_over_pt}). Adding a $c$-twist to both sides, (or exchanging the roles of $\Sflip$ and $\hSflip$), we get the crystallographic T-duality for the graded group $\sf{p2}$,
$$
{\rm T}_{\sf{p2}}: K^{\bullet+c_{\sf{p2}}+c}_{\ZZ_2}(\Sflip)=K^{\bullet}_{\ZZ_2}(\Sflip)\overset{\cong}{\longrightarrow}K^{\bullet-1+c}_{\ZZ_2}(\hSflip)\cong K^{\bullet-1}_\pm(\hSflip).
$$

\subsubsection{Point group $\ZZ_2$ acting trivially: $\sf{p11m}, \sf{p11g}$}\label{sec:p11g}
The frieze group $\sf{p11m}$ is $\ZZ\times\ZZ_2$, with the point group $\ZZ_2$ reflecting the vertical coordinate and thus nontrivially graded. It has $T^1_{\sf{p11m}}=B\ZZ=R^1/\ZZ=\Striv$ and Brillouin torus $\hStriv$. 

As $\ZZ_2$-equivariant circle bundles over a point, $\Striv, \hStriv$ are T-dual,
\begin{equation*}
\TZtwo:K^{\bullet}_{\ZZ_2}(\Striv)\overset{\cong}{\longrightarrow} K^{\bullet-1}_{\ZZ_2}(\hStriv),\label{basicZ2duality2}
\end{equation*}
and adding a $c$-twist on both sides gives crystallographic T-duality for $\sf{p11m}$,
$$
{\rm T}_{\sf{p11m}}: K^{\bullet+c}_{\ZZ_2}(\Striv)\overset{\cong}{\longrightarrow} K^{\bullet-1+c}_{\ZZ_2}(\hStriv).
$$

\begin{remark}
In \S\ref{sec:H1twistexamples}, we saw that $H^1_{\ZZ_2}(\Striv,\ZZ_2)=\ZZ/2\oplus\ZZ/2$ with generators $c$ and $M$ the M\"{o}bius bundle over $\Striv$ made $\ZZ_2$-equivariant in a trivial way. We sketch a strategy to T-dualise $(\Striv,M)$ and $(\Striv, M+c)$ in \S\ref{sec:Mobiustwist}. 
\end{remark}

In \S\ref{sec:H3twistexamples}, we also saw that $H^3_{\ZZ_2}(\Striv,\ZZ)\cong\ZZ/2$ generated by $\tau_{S^1}$ defined by Eq.\ \eqref{2cocyclecircle}. This $H^3$-twist appears in the crystallographic T-dual of $\sf{p11g}$, the \emph{graded} group $\ZZ$ generated by an odd glide reflection, i.e.\ reflection of vertical coordinate followed by half a lattice translation, see Fig.\ \ref{fig:glideaxis}.
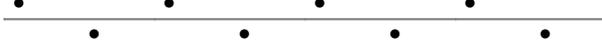
\begin{figure}
\begin{tikzpicture}[every node/.style={scale=1.5}]
\draw[help lines,thick,step=2] (-2,0) grid (6,0);

\foreach \p in {(-1.8,0.2),(0.2,0.2),(2.2,0.2),(4.2,0.2)}{\draw \p node[scale=0.6] {$\bullet$};}

\foreach \p in {(-0.8,-0.2),(1.2,-0.2),(3.2,-0.2),(5.2,-0.2)}{\draw \p node[scale=0.6] {$\bullet$};}

\end{tikzpicture}
\caption{The generating translation along a glide axis also effects an exchange of the internal label ``above/below the axis''.}\label{fig:glideaxis}
\end{figure}
The (even) lattice subgroup $\Pi$ is a \emph{proper} subgroup of $\sf{p11g}$ of index 2,
\begin{equation}
0\rightarrow \Pi\cong\ZZ\overset{\times 2}{\longrightarrow} \sf{p11g}\cong\ZZ\overset{(-1)^{(\cdot)}}{\longrightarrow}\ZZ_2\rightarrow 1\label{p11gextension}
\end{equation}
On $T^1=R^1/\Pi$, the translational part of the $\ZZ_2$-action is $s_{\sf{p11g}}:-1\mapsto e^{\im \pi}=-1$, so $T^1_{\sf{p11g}}=\Sfree$. The Brillouin zone is $\wh{\Pi}=\hStriv$ but has the cocycle twist $\tau_{\sf{p11g}}=\tau_{S^1}$ due to the 2-cocycle $\nu_{\sf{p11g}}(-1,-1)=1$ for Eq.\ \eqref{p11gextension}. By Proposition \ref{prop:Z2_dual_over_pt}, $(\Sfree,0)$ and $(\hStriv,\tau_{S^1})$ are T-dual pairs in the $\ZZ_2$-equivariant sense, so
$$
\TZtwo:K^\bullet_{\ZZ_2}(\Sfree)\overset{\cong}{\longrightarrow}K^{\bullet-1+\tau_{S^1}}_{\ZZ_2}(\hStriv).
$$
Adding a $c$-twist to both sides gives crystallographic T-duality for $\sf{p11g}$,
\begin{equation}
{\rm T}_{\sf{p11g}}: K^{\bullet+c}_{\ZZ_2}(\Sfree)\overset{\cong}{\longrightarrow}K^{\bullet-1+(\tau_{S^1},c)}_{\ZZ_2}(\hStriv).\label{p11gTduality}
\end{equation}

\begin{remark}
In \cite{SSG1}, the odd glide reflection generating $\sf{p11g}$ was called a \emph{nonsymmorphic chiral symmetry}, and the $K$-theory of the $(\tau_{S^1},c)$-twisted $\hStriv$ was computed to be $\ZZ/2$. This $K$-theory group is important for the crystallographic bulk-boundary correspondence for $\mathscr{G}=\sf{pg}$, as studied in \cite{GT} and briefly discussed in \S\ref{sec:bulkboundaryapplication}.
\end{remark}

The T-dualities associated to frieze groups with $G\subset\ZZ_2$ appear in Table \ref{table:1Ddualitiestable}.

\subsubsection{$D_2$ point group: $\sf{p2mm},\sf{p2mg}$}\label{sec:D2pointgroup1D}
For $\sf{p2mm}$, the point group is $D_2=\ZZ_2\times\ZZ_2$ with the first (resp.\ second) generator reflecting the vertical (resp.\ horizontal) coordinate, so $T^1_{\sf{p2mm}}=S^1_{{\rm triv}\times{\rm flip}}$, and similarly for the Brillouin torus $\hat{T}^1_{\sf{p2mm}}=\hat{S}^1_{{\rm triv}\times{\rm flip}}$. Let the $i$-th projection homomorphism $p_i:D^2=\ZZ_2\times\ZZ_2\rightarrow\ZZ_2$ define the $H^1$-twist $c_i, i=1,2$, then the grading $p_1$ on $D_2$ gives the twist $c_1$, while the orientability homomorphism $p_2$ gives $c_{\sf{p2mm}}=c_2=c_{\sf{p2mg}}$.

If $\sf{p2mm}$ is regarded as an \emph{ungraded} group $\ZZ\rtimes D_2$, the crystallographic T-duality Eq.\ \ref{crystalTdual} would give
\begin{equation*}
{\rm T}^{\rm ungraded}_{\sf{p2mm}}:K^{\bullet+c_2}_{D_2}(S^1_{{\rm triv}\times{\rm flip}})\overset{\cong}{\rightarrow}K^{\bullet-1}_{D_2}(\hat{S}^1_{{\rm triv}\times{\rm flip}}),\label{p2mmTduality}
\end{equation*}
We anticipate that when the grading twist $c_1$ is added, we will obtain
\begin{equation*}
{\rm T}_{\sf{p2mm}}:K^{\bullet+(\omega,c_1+c_2)}_{D_2}(S^1_{{\rm triv}\times{\rm flip}})\overset{\cong}{\rightarrow}K^{\bullet-1+c_1}_{D_2}(\hat{S}^1_{{\rm triv}\times{\rm flip}}),
\end{equation*}
where on the LHS, we recall that $c_1+c_2\equiv(0,c_1)+(0,c_2)=(\omega,c_1+c_2)$ in the group of graded $D_2$-equivariant twists (pulled back from $\pt$) due to Eq.\ \eqref{D2twistmodified}.

For $\sf{p2mg}$, the vertical coordinate reflection is replaced by a glide reflection so that $T^1_{\sf{p2mg}}=S^1_{{\rm free}\times{\rm flip}}$. The Brillouin torus is again $\hat{S}^1_{{\rm triv}\times{\rm flip}}$, with a 2-cocycle twist $\tau_{\sf{p2mg}}$ from the nonsymmorphicity. The crystallographic T-duality for the group $\sf{p2mg}$ is 
\begin{align*}
{\rm T}^{\rm ungraded}_{\sf{p2mg}}&:K^{\bullet+c_2}_{D_2}(S^1_{{\rm free}\times{\rm flip}})\overset{\cong}{\rightarrow}K^{\bullet-1+\tau_{\sf{p2mg}}}_{D_2}(\hat{S}^1_{{\rm triv}\times{\rm flip}}),\nonumber\\
{\rm T}_{\sf{p2mg}}&:K^{\bullet+c_1+c_2}_{D_2}(S^1_{{\rm free}\times{\rm flip}})\overset{\cong}{\rightarrow}K^{\bullet-1+(\tau_{\sf{p2mg}},c_1)}_{D_2}(\hat{S}^1_{{\rm triv}\times{\rm flip}})
\label{p2mgTduality}
\end{align*}
where we note that $H^3_{D^2}(S^1_{{\rm free}\times{\rm flip}},\ZZ)\cong H^3_{\ZZ_2}(S^1_{{\rm flip}},\ZZ)=0$ so $c_1$, $c_2$ add as graded twists in the na\"{i}ve way as in $H^1_{D_2}(S^1_{{\rm free}\times{\rm flip}},\ZZ_2)$.

\begin{table}
\begin{center}
 \begin{tabular}{|| c | c | c | c | c | c | c | c ||} 
 \hline
 \multicolumn{3}{|| c |}{Primal $\bullet=0$} & \multirow{2}{*}{$K_{\ZZ_2}^\bullet$} & \multirow{2}{*}{$K_{\ZZ_2}^{\bullet-1}$} & \multicolumn{3}{| c ||}{$\bullet=1$ Dual}\\
 \cline{1-3}\cline{6-8}
Space & $H^3$-twist & $H^1$-twist &  &  & $H^1$-twist & $H^3$-twist & Space \\ [0.5ex] 
 \hline\hline
 $S^1$ & N/A & 0 & $\ZZ$ & $\ZZ$ & 0 & N/A & $S^1$\\
 \hline\hline
 $\Sflip$ & N/A & 0 & $\ZZ^3$ & 0 & $c$ & N/A & $\Sflip$\\
 \hline
\multirow{2}{*}{$\Striv$} & 0 & 0 & $\ZZ^2$ & $\ZZ^2$  & 0 & 0 & \multirow{4}{*}{$\Striv$} \\ 
 \cline{2-7}
 & 0 & $c$ & $\ZZ$ & $\ZZ$ & $c$ & 0 & \\
 \cline{1-7}
 \multirow{2}{*}{$\Sfree$} & N/A & 0 & $\ZZ$ & $\ZZ$ & 0 & $\tau_{S^1}$ & \\ 
 \cline{2-7}
 & N/A & $c$ & 0 & $\ZZ/2$ & $c$ & $\tau_{S^1}$  & \\
 \hline

\end{tabular}
\caption{$S^1$ and the three involutive circles with all possible graded twists, except those of $M$-type, are T-dualised as above. To T-dualise $(S^1,M)$ we pass to $(\Sfree,c)$ instead. To T-dualise $(\Striv,M)$, we need to pass to a double cover and take a conjectured $D_2$-equivariant T-dual.}\label{table:1Ddualitiestable}
\end{center}
\end{table}

\subsection{T-duality with M\"{o}bius twists and $G$-equivariant T-duality}\label{sec:Mobiustwist}
So far, the $H^1$-twists that we have considered are of $c$-type, coming from a homomorphism $G\rightarrow\ZZ_2$. Consider $(S^1,M)$ where $M\in H^1(S^1,\ZZ_2)$ is the M\"{o}bius twist. Passing to the double cover $\Sfree$, the generating twist $c\in H^1_{\ZZ_2}(\Sfree,\ZZ_2)\cong\ZZ/2$ corresponds to $M$, and it is possible to show that
$$
K^{\bullet+M}(S^1)\cong K^{\bullet+c}_{\ZZ_2}(\Sfree)\cong \begin{cases}0,\quad\,\qquad \bullet=0,\\ \ZZ/2,\qquad \bullet=1. \end{cases}
$$
We had already seen that $\Sfree$ can be identified with $T^1_{\sf{p11g}}$ and found the $\ZZ_2$-equivariant T-dual of $(\Sfree,c)$ in Eq.\ \eqref{p11gTduality}. Thus $(S^1,M)$ has a T-dual pair via passage to an equivariant double cover $\Sfree$.

Now consider $\Striv$ which has $H^1_{\ZZ_2}(\Striv,\ZZ_2)\cong\ZZ/2\oplus\ZZ/2$ generated by $c$ and by $M$ made $\ZZ_2$-equivariant in the trivial way. A similar strategy to T-dualise $(\Striv,M)$ is to pass to a double cover $\tilde{S}^1_{\rm triv}=S^1_{{\rm free}\times{\rm triv}}\rightarrow S^1_{\rm triv}$ which has a $D_2=\ZZ_2\times\ZZ_2$ action with the first factor acting by deck transformations. Now regard $M\rightarrow\Striv$ as a $D_2$-equivariant real line bundle $\tilde{M}\rightarrow S^1_{{\rm free}\times{\rm triv}}$ --- explicitly, $\tilde{M}$ is the product bundle with $D_2$ acting via its first $\ZZ_2$ factor by the deck transformation on the base and $-1$ on the fibre. Thus $\tilde{M}$ can be regarded as the twist $c_1\in H^1_{D_2}(S^1_{{\rm free}\times {\rm triv}},\Z_2)$ coming from the homomorphism $p_1:D_2\rightarrow\ZZ_2$.

Conjecturally, there is a notion of $D_2$-equivariant T-duality $T_{D_2}$ (and also for more general groups $G$), generalising $\TZtwo$ and $\TR$ in a natural way. Then we can T-dualise $(\Striv,M)$ by first passing to $(S^1_{{\rm free}\times{\rm triv}},c_1)$ and then taking $T_{D_2}$. Furthermore, the frieze group dualities in \S\ref{sec:D2pointgroup1D}, as well as the wallpaper group crystallographic T-dualities in Eq.\ \eqref{wallpaperD2dualities}, would be expected to be implemented by ${\rm T}_{D_2}$. The circle can be made a $D_2$ space in several other ways such as $S^1_{{\rm free}\times{\rm triv}}$. These $D_2$-actions can arise from more general subperiodic groups such as the \emph{rod groups} \cite{KL} associated to symmetries of a line in 3D space. For example, a two fold screw axis which is also on a reflection plane is preserved by a point group $D_2$. The $D_2$ action on a unit cell (a circle) for the lattice translation along the axis gives $S^1_{{\rm free}\times{\rm triv}}$. Generalising the particular case of $\sf{p11g}$ studied in \cite{GT}, we expect that the $K$-theories associated to rod groups will be important for crystallographic bulk-boundary correspondences with screw axes.

\section{3D dualities and applications}
\subsection{H-flux from partial T-duality: screw dislocations}\label{sec:screw}
In \cite{HMT}, it was observed that $H^3$-flux (in the nonequivariant sense) is ``produced'' when a screw-dislocated lattice is partially Fourier transformed. In string theory, one might start with $T^3$ with ``one unit of H-flux'', meaning that $K^{\bullet+h}(T^3)$ is needed, for $h$ a generator of $H^3(T^3,\ZZ)$, to study D-brane charges. As a circle bundle over $T^2$, the T-dual of the pair $(T^3,h)$ is $({\rm Nil},0)$ where the nilmanifold ${\rm Nil}$ is the circle bundle over $T^2$ with Chern class the generator of $H^2(T^2,\ZZ)$ as required. The name ``Nil'' comes from the fact that $\pi_1({\rm Nil})={\rm {Heis}}^\ZZ$, the integer Heisenberg group
$${\rm {Heis}}^\ZZ=\left\{\begin{pmatrix} 1 & a & c \\ 0 & 1 & b \\ 0 & 0 & 1 \end{pmatrix}\,:\,a,b,c\in\ZZ\right\},$$
and ${\rm Nil}={\rm {Heis}}^\RR/{\rm {Heis}}^\ZZ$ is a $B{\rm {Heis}}^\ZZ$. This example illustrated ``topology change from H-flux'' as in \cite{BEM}.

The story is run from a different angle in \cite{HMT}, where the nonabelian lattice ${\rm {Heis}}^\ZZ$ was considered to be a screw-dislocated version of the standard (Euclidean) lattice $\ZZ^3$. One sees this by noticing that the commutator of the $a$ and $b$ translations in ${\rm {Heis}}^\ZZ$ is not zero (corresponding to a closed loop) but rather a translation in the third direction (corresponding to helical motion along an axis in this third direction). The fundamental domain (``unit cell'' in position space) is no longer the 3-torus, but ${\rm Nil}$. Despite ${\rm {Heis}}^\ZZ$ being nonabelian, it is built up from three copies of $\ZZ$, and the noncommutative ``momentum space'' $C^*_r({\rm {Heis}}^\ZZ)$ can be understood as a field of noncommutative tori parameterised by the circle dual to the central $\ZZ$ \cite{MR}. Note that ${\rm {Heis}}^\ZZ$ is a discrete cocompact subgroup of the continuous version ${\rm {Heis}}^\RR$ (with real number entries), and then ``crystallographic T-duality'' for ${\rm {Heis}}^\ZZ$ can be defined as Poincar\'{e} duality for ${\rm Nil}=B{\rm {Heis}}^\ZZ={\rm {Heis}}^\RR/{\rm {Heis}}^\ZZ$ composed with the Baum--Connes assembly map. In summary, we have
\begin{equation*}
\label{eqn:factorization}
\xymatrix{K^{\bullet+h} (T^3) \ar[rr]_{{\rm T}_{\rm torus}}&&
K_\bullet(C^*_r({\rm Heis}^\ZZ))\\
&K^{\bullet+1} ({\rm Nil})  \ar[ul]_{{\rm T}_{\rm circle}}\ar[ur]_{{\rm T}_{\rm crystal}}^{\sim} &}
\end{equation*}
Then we see that ${\rm T}_{\rm circle}$, interpreted as a partial Fourier transform, means that the mixed position-momentum space $T^3$ comes with a $H^3$-twist. This is an other instance of the observation in \S\ref{sec:partialTflux}.

\subsection{Crystallographic bulk-boundary correspondence and super-indices for boundary zero modes}\label{sec:bulkboundaryapplication}
In a crystalline version of the bulk-boundary correspondence, a $d$-dimensional crystalline topological insulator should be detectable on some codimension-1 layer fixed under some point group operation. Such a layer need not only have an ordinary $d-1$-dimensional space group symmetry, but the isotropy can also contribute by toggling the ``above/below'' degree of freedom.

For example, if $\mathscr{G}=\sf{pg}$, there is a 2-torsion Class AIII phase because of $K^{-1+\tau_{S^1}}_{\ZZ_2}(\hat{T}^2_{\sf{pm}})\cong\ZZ\oplus\ZZ/2$ (recall that $\tau_{\sf{pg}}=\tau_{S^1}$). In \cite{GT}, it was shown that this phase is detected by zero modes localised on a cut along a glide axis, and such an axis has precisely the 1D frieze group $\sf{p11g}$ symmetry with the generator given the odd grading, see Fig.\ \ref{fig:glideaxis}. The graded group $\sf{p11g}$ gives $K^{0+c+\tau_{S^1}}_{\ZZ_2}(\hStriv)\cong\ZZ/2$ on the RHS of crystallographic T-duality, and may be understood as $K^{\rm graded}_\bullet(C^*_r(\sf{p11g}))$ of the \emph{graded group} $C^*$-algebra for ${\sf{p11g}}\cong\ZZ\overset{(-1)^n}{\rightarrow}\ZZ_2$. A natural `Real' Gysin map takes $\pi_*:K^{-1+\tau_{S^1}}_{\ZZ_2}(\hat{T}^2_{\sf{pm}})\rightarrow K^{0+c+\tau_{S^1}}_{\ZZ_2}(\hStriv)$ along the fibre projection $\pi:\hat{T}^2_{\sf{pm}}=\hSflip\times\hStriv\rightarrow\hStriv$, and it was shown in \cite{GT} that $\pi_*$ realises an analytic index map for a $\tau_{S^1}$-twisted family of Toeplitz-like operators parameterised by $\hStriv$. In this way, the target group 
$$K^{\rm graded}_\bullet(C^*_r({\sf{p11g}}))\cong K^{0+c+\tau_{S^1}}_{\ZZ_2}(\hStriv)$$
is the ``super-higher index'' group for the $\sf{p11g}$-symmetric topological boundary zero modes of $\sf{pg}$-symmetric insulators.

\subsection{Spectral sequence extension problems and halving computations of topological phases}\label{sec:extensionproblem}
Consider the symmorphic 3D space groups $\mathscr{G}=\sf{P222},\sf{C222},\sf{F222},\sf{I222}$, which have point group ${\sf 222}\cong D_2=\ZZ_2\times\ZZ_2\rightarrow{\rm O}(3)$ whose three nontrivial elements are $\pi$-rotations about three mutually orthogonal axes (say, $x, y, z$). Note that $H^3_{\sf{222}}(\pt,\ZZ)\cong\ZZ/2$ with generator $\omega$ which pulls back to a nontrivial twist on $T^3_{\mathscr{G}}$. We do not need the precise description of these space groups, but just the fact that the arithmetic crystal classes for $\sf{P222},\sf{C222}$ are self-dual but those for $\sf{F222}$ and $\sf{I222}$ are dual to each other, i.e.\ $\alpha$ for one is $\hat{\alpha}$ for the other (see (45)-(47) of \cite{Michel} for a list of dual pairs of arithmetic crystal classes).

In \cite{SSG3}, the Atiyah--Hirzebruch spectral sequence (AHSS) was used to compute the $K$-theories of $T^3_{\sf{P222}}, T^3_{\sf{C222}}, T^3_{\sf{F222}}, T^3_{\sf{I222}}$, with the results\footnote{In Table 4 of \cite{SSG3}, the entries for a space group are the $E_2$ terms for the untwisted ($+0$) and $\omega$-twisted ($-1/2$) $K$-theory for the corresponding Brillouin torus with \emph{dual} point group action. So, e.g.\ the $\sf{F222}$ entries compute the $K$-theory for $T^3_{\sf{I222}}$.}
\begin{align*}
K^0_{D_2}(T^3_{\sf{P222}})\cong \ZZ^{13},\qquad & K^1_{D_2}(T^3_{\sf{P222}})\cong \ZZ\;{\rm or}\;\ZZ\oplus\ZZ/2,\\
K^{0+\omega}_{D_2}(T^3_{\sf{P222}})\cong \ZZ,\qquad & K^{1+\omega}_{D_2}(T^3_{\sf{P222}})\cong \ZZ^{13},\\
K^0_{D_2}(T^3_{\sf{C222}})\cong \ZZ^8,\qquad & K^1_{D_2}(T^3_{\sf{C222}})\cong \ZZ^2\;{\rm or}\;\ZZ^2\oplus\ZZ/2,\\
K^{0+\omega}_{D_2}(T^3_{\sf{C222}})\cong \ZZ^2,\qquad & K^{1+\omega}_{D_2}(T^3_{\sf{C222}})\cong \ZZ^8,\\
K^0_{D_2}(T^3_{\sf{I222}})\cong \ZZ^7\oplus\ZZ/2,\qquad & K^1_{D_2}(T^3_{\sf{I222}})\cong \ZZ\;{\rm or}\;\ZZ\oplus\ZZ/2,\\
K^{0+\omega}_{D_2}(T^3_{\sf{I222}})\cong \ZZ\oplus\ZZ/2,\qquad & K^{1+\omega}_{D_2}(T^3_{\sf{I222}})\cong \ZZ^7,\\
K^0_{D_2}(T^3_{\sf{F222}})\cong \ZZ^7,\qquad & K^1_{D_2}(T^3_{\sf{F222}})\cong \ZZ\oplus\ZZ_2^2\;{\rm or}\;\ZZ\oplus\ZZ/2,\\
K^{0+\omega}_{D_2}(T^3_{\sf{F222}})\cong \ZZ,\qquad & K^{1+\omega}_{D_2}(T^3_{\sf{F222}})\cong \ZZ^7\;{\rm or}\;\ZZ^7\oplus\ZZ/2.
\end{align*}
It is possible to resolve the ambiguity for the untwisted $K^1_{D_2}$ by a direct Mayer--Vietoris computation, but let us instead show how crystallographic T-duality comes to the rescue. In anticipation of this, notice that groups in the left column also appear in the right column.

First, note that the ${\sf 222}\cong D_2$ point group action is orientable but has the ${\rm Spin}^c$ obstruction $\omega$, by Lemma \ref{lem:222W3obstruction}, and this is pulled back faithfully to the $K_{D_2}$-orientability obstruction for $T^3_{\mathscr{G}}$. Consequently, the crystallographic T-dualities are (dropping the hats for now)
\begin{align*}
{\rm T}_{\sf{P222}}&:K^{\bullet+\omega}_{D_2}(T^3_{\sf{P222}})\overset{\cong}{\longrightarrow}K^{\bullet-1}_{D_2}(T^3_{\sf{P222}}),\\
{\rm T}_{\sf{C222}}&:K^{\bullet+\omega}_{D_2}(T^3_{\sf{C222}})\overset{\cong}{\longrightarrow}K^{\bullet-1}_{D_2}(T^3_{\sf{C222}}),\\
{\rm T}_{\sf{F222}}&:K^{\bullet+\omega}_{D_2}(T^3_{\sf{F222}})\overset{\cong}{\longrightarrow}K^{\bullet-1}_{D_2}(T^3_{\sf{I222}}),\\
{\rm T}_{\sf{I222}}&:K^{\bullet+\omega}_{D_2}(T^3_{\sf{I222}})\overset{\cong}{\longrightarrow}K^{\bullet-1}_{D_2}(T^3_{\sf{F222}}).
\end{align*}
These dualities enable the resolution of the $K^1$ ambiguities by referring to the unambiguous $K^0$ groups on the T-dual side, i.e., 
\begin{align*}
K^1_{D_2}(T^3_{\sf{P222}}) &\cong \ZZ,\qquad K^1_{D_2}(T^3_{\sf{C222}}) \cong \ZZ^2, \qquad K^1_{D_2}(T^3_{\sf{I222}}) \cong \ZZ, \\
K^1_{D_2}(T^3_{\sf{F222}}) &\cong \ZZ\oplus\ZZ/2,\qquad K^{1+\omega}_{D_2}(T^3_{\sf{F222}}) \cong \ZZ^7\oplus\ZZ/2.
\end{align*} 
These examples demonstrate how our crystallographic T-duality supplements the powerful general machinery of the AHSS. In effect, the number of computations is halved, some twisted $K$-theories can be computed more easily on the T-dual side (cf.\ Remark \ref{rem:pgsimplification}), and extension problems may be resolved by inspecting the T-dual computations. In the physics context, these $K^1$ groups (with no $\omega$-twist) classify the so-called Class AIII topological insulators with respective space group symmetries, and in particular (restoring the hat) $K^1_{D_2}(\hat{T}^3_{\sf{F222}}) \cong \ZZ\oplus\ZZ/2$ shows that there is a 2-torsion chiral symmetric and $\sf{I222}$-symmetric phase.

\section*{Acknowledgements}
G.C.T.~is supported by Australian Research Council grant DE170100149, and K.G.~by JSPS KAKENHI Grant Number JP15K04871. Both authors would like to thank Siye Wu for his hospitality at the National Center for Theoretical Sciences (Physics Division) of Taiwan, where the ideas for this paper crystallised.

\appendix
\section{Appendix}
In order to determine the $K_G$-orientability obstruction $\sigma_\mathscr{G}$ for a torus with action induced from a space group $\mathscr{G}$, Eq.\ \eqref{Korientabilityobstruction}, we need to first compute the obstruction class $W_3^G(\RR^d_\rho)$ of Eq.\ \eqref{W3obstructionpoint} associated to the point group $\rho:G\rightarrow{\rm O}(d)$ as follows. To shorten notation, we just write $g$ for $\rho(g)\in {\rm O}(d)$. Choose lifts $\wt{g}\in {\rm Pin}^c(d)$ of $g\in G\subset {\rm O}(d)$ in the central extension
\begin{equation}
1\rightarrow {\rm U}(1)\rightarrow {\rm Pin}^c(d)\overset{\varpi^c}{\rightarrow}{\rm O}(d)\rightarrow 1,\label{pincextensionappendix}
\end{equation}
giving a \emph{projective} representation of $G$ whose cocycle $\zeta\in Z^2(G,{\rm U}(1))$ is
$$ \wt{g}\wt{h}=\zeta(g,h)\wt{gh},\qquad g,h\in G.$$
If $[\zeta]=0\in H^2_{\rm group}(G,{\rm U}(1))$, then we can actually choose the lifts $\wt{g}$ to give a genuine representation of $G$ factoring through ${\rm Pin}^c(d)$, otherwise there is an obstruction and $(W_3^G(\RR^d_\rho),W_1^G(\RR^d_\rho))\neq 0$.
\begin{lemma}\label{lem:nontrivialcocycle}
Let $\zeta\in Z^2(G,{\rm U}(1))$ be a cocycle for a finite group $G$, and let $\epsilon(g,h)\coloneqq\zeta(g,h)\zeta(h,g)^{-1}$. If $\epsilon(g,h)\neq 1$ for some $g,h$ such that $gh=hg$, then $[\zeta]\neq 0 \in H^2_{\rm group}(G,{\rm U}(1))$.
\end{lemma}
\begin{proof}
If $\zeta$ is a coboundary, we can verify that $\epsilon(g,h)=1$ whenever $gh=hg$.
\end{proof}

{\bf Spin${}^c$ obstruction for $D_2, D_4, D_6$ point groups.} Let us analyse the 2D point groups $G=D_2, D_4, D_6\subset {\rm O}(2)$. They are generated in ${\rm O}(2)$ by a reflection $\varsigma=\begin{pmatrix} -1 & 0 \\ 0 & 1 \end{pmatrix}$ and a rotation $r_\theta=\begin{pmatrix}\cos \theta & -\sin \theta \\ \sin \theta & \cos \theta\end{pmatrix}$ with $\theta=\frac{2\pi}{n}$. Let $\CC l(2)$ be the complex Clifford algebra generated by $e_1, e_2$ with $e_1^2=e_2^2=-1, e_1e_2+e_2e_1=0$. The ${\rm Pin}(2)$ group is a double cover of ${\rm O}(2)$ and can be realised concretely inside $\CC l(2)$ as
$$ {\rm Pin}(2)=\{\cos \theta + \sin \theta e_1e_2\}\cup \{\cos \theta e_1+ \sin \theta e_2\}_{\theta\in[0,2\pi]}$$
The double-cover projection ${\rm Pin}(2)\overset{\varpi}{\rightarrow} {\rm O}(2)$ is then given by
$$ \varpi(\cos \theta + \sin \theta e_1e_2)= r(2\theta),\qquad \varpi(\cos \theta e_1+ \sin \theta e_2)= r(2\theta)\varsigma$$
The ${\rm Pin}^c(2)$ group is defined to be ${\rm Pin}(2)\times {\rm U}(1)/_{\{(1,1),(-1,-1)\}}$, and the projection in the central extension Eq.\ \eqref{pincextensionappendix}
is $\varpi^c[x,u]=\varpi(x)$.
\begin{lemma}\label{lem:2DW3obstruction}
For the point groups $G=D_2, D_4, D_6\overset{\rho}{\rightarrow}{\rm O}(2)$, the class $W_3^G(\RR^2_\rho)\in H^3_G(\pt,\ZZ)\cong\ZZ/2$ is the unique generator $\omega$.
\end{lemma}
\begin{proof}
Note that $D_2, D_4, D_6$ each contain the commuting elements $\varsigma$ and $r_\pi$. Choose the lifts
$$ \wt{r_\pi}=[e_1e_2,-\im],\qquad \wt{\varsigma}=[e_1,\im]\qquad \wt{r_\pi\varsigma}\equiv \wt{\varsigma r_\pi} =\wt{r_\pi}\wt{\varsigma}=[e_2,1].$$
Then Lemma \ref{lem:nontrivialcocycle} applied to the computation
$$\zeta(\varsigma,r_\pi)\equiv \wt{\varsigma}\wt{r_\pi}\wt{\varsigma r_\pi}^{-1}=[e_1,\im][e_1e_2,-\im][e_2,-1]=[e_1^2e_2^2,-1]=-1 $$
$$\zeta(r_\pi,\varsigma)\equiv \wt{r_\pi}\wt{\varsigma}\wt{r_\pi \varsigma}^{-1}=[e_1e_2,-\im][e_1,\im][e_2,-1]=[-e_1^2e_2^2,-1]=1 $$
shows that $\rho$ is not ${\rm Pin}^c(2)$, and the obstruction $W_3^G(\RR^2_\rho)$ is nontrivial.
\end{proof}

As a concrete representative of the cohomology class $\omega$, let us tabulate the representative 2-cocycle $\zeta\in Z^2(D_2,{\rm U}(1))$ for the above choice of lift of $D_2$. It is convenient to rewrite the $D_2$ generators as $\sigma_x=\varsigma$ and $\sigma_y=r_\pi\varsigma$, corresponding respectively to reflections of the $x$ and $y$ coordinates of 2D Euclidean space. 
$$
\begin{array}{|c|c|c|c|c|}
\hline
\zeta(g,h)& h = 1 & h = \sigma_x & h=\sigma_y & h=r_\pi \\
\hline
g=1 & 1 & 1 & 1 & 0  \\
\hline
g=\sigma_x & 1 & 1 & -1 & -1  \\
\hline
g=\sigma_y & 1 & 1 & -1 & -1  \\
\hline
g=r_\pi & 1 & 1 & 1 & 1  \\
\hline
\end{array}
$$
{\bf Cocycle twist for \sf{pmg}.} We also compute the dual 2-cocycle twist $\tau_{\sf{pmg}}$ on $\hat{T}^2_{\sf{pmm}}$ induced by the nonsymmorphic space group $\sf{pmg}$ (which we recall has point group $D_2$). The dual $D_2$ action on $\hat{T}^2_{\sf{pmm}}$ is concretely given by $\sigma_x:(u,v)\mapsto (\bar{u},v)$ and $\sigma_y:(u,v)\mapsto (u,\bar{v})$ in terms of standard unit-complex number coordinates $u,v$. Using the tilde notation this time to denote lifts of $\sigma_x,\sigma_y, r_\pi=\sigma_x\sigma_y\in D_2$ inside $\sf{pmg}$, a standard choice is: $\wt{\sigma_x}$ remains the reflection of the $x$-coordinate, $\wt{\sigma_y}$ becomes a \emph{glide} reflection\footnote{As in $\sf{pg}$, reflect $y$ then half-translate along $x$.}, and $\wt{r_\pi}=\wt{\sigma_y}\wt{\sigma_x}$. Noting that $\wt{\sigma_y}^2$ equals the unit translation $T_x$ along $x$ (which Fourier transforms to $u$), and that $(\wt{\sigma_x}\wt{\sigma_y})^2=1=(\wt{\sigma_y}\wt{\sigma_x})^2$, it is easy to see that the ${\rm U}(C(\hat{T}_{\sf{pmm}}))$-valued 2-cocycle $\tau_{\sf{pmg}}$ has values as follows:
$$
\begin{array}{|c|c|c|c|c|}
\hline
\tau_{\sf{pmg}}(g,h)& h = 1 & h = \sigma_x & h=\sigma_y & h=r_\pi \\
\hline
g=1 & 1 & 1 & 1 & 0  \\
\hline
g=\sigma_x & 1 & 1 & \bar{u} & \bar{u}  \\
\hline
g=\sigma_y & 1 & 1 & u & u  \\
\hline
g=r_\pi & 1 & 1 & 1 & 1  \\
\hline
\end{array}
$$
The map $(u,v)\mapsto(-u,v)$ is a $D_2$-equivariant automorphism of $\hat{T}_{\sf{pmm}}$, and by comparing the tables for $\zeta$ and $\tau_{\sf{pmg}}$, we see that pullback under this automorphism converts $\tau_{\sf{pmg}}\mapsto \tau_{\sf{pmg}}+\omega$.

{\bf Spin${}^c$ obstruction for \sf{222} point group.} Next, we analyse the 3D point group ${\sf 222}\cong\ZZ_2\times\ZZ_2\subset{\rm O}(d)$ which contains $\pi$ rotations $r_x, r_y, r_z=r_xr_y=r_yr_x$ about the $x,y,z$ axes. Since $\sf{222}\subset{\rm SO}(3)$, we only need to check whether it is also ${\rm Spin}^c$, and we recall the exact sequence
$$1\rightarrow {\rm U}(1)\rightarrow {\rm U}(2)= {\rm Spin}^c(3)\rightarrow {\rm PU}(2)={\rm SO}(3)\rightarrow 1.$$

\begin{lemma}\label{lem:222W3obstruction}
For the point group $G=\sf{222}\overset{\rho}{\rightarrow}{\rm SO}(3)\subset{\rm O}(3)$, the class $W_3^G(\RR^3_\rho)\in H^3_G(\pt,\ZZ)\cong\ZZ/2$ is the unique generator $\omega$.
\end{lemma}
\begin{proof}
Choose lifts of $r_i\in G, i=x,y,z$ to be $\wt{r_i}=e^{\frac{\im \pi}{2}\sigma_i}=\im\sigma_i\in{\rm U}(2)={\rm Spin}^c(3)$ where $\sigma_i$ are the Pauli spin matrices. Then the cocycle $\zeta$ has
$$ \zeta(r_x,r_y)=\wt{r_x}\wt{r_y}(\wt{r_xr_y})^{-1}=(\im\sigma_x)(\im\sigma_y)(\im\sigma_z)^{-1}=-1,$$
and is nontrivial by Lemma \ref{lem:nontrivialcocycle}. Thus $\rho$ is not ${\rm Spin}^c(3)$, and its obstruction $W_3^G(\RR^3_\rho)$ is nontrivial.
\end{proof}

\end{document}